\documentclass[11pt]{article}

\PassOptionsToPackage{
	backref=page,
	colorlinks=true,
	linktoc=all,
	linkcolor=RoyalBlue4,
	citecolor=SeaGreen4,
	urlcolor=Purple4,
	hypertexnames=false
}{hyperref}

\usepackage[T1]{fontenc}
\usepackage[utf8]{inputenc}
\usepackage[margin=1in]{geometry}
\usepackage{fullpage}
\usepackage{microtype}
\usepackage{xurl}
\usepackage{doi}

\usepackage{mathtools}
\usepackage{amssymb,amsfonts}
\usepackage{amsthm}

\numberwithin{equation}{section}
\newtheorem{theorem}{Theorem}[section]
\newtheorem{lemma}[theorem]{Lemma}
\newtheorem{corollary}[theorem]{Corollary}

\newtheorem{definition}[theorem]{Definition}
\newtheorem{claim}[theorem]{Claim}
\newtheorem{problem}[theorem]{Problem}

\theoremstyle{remark}
\newtheorem{remark}[theorem]{Remark}
\newtheorem{observation}[theorem]{Observation}

\usepackage{graphicx}
\usepackage{epstopdf}
\usepackage{array}
\usepackage{multirow}
\usepackage{enumitem}
\usepackage{appendix}
\usepackage{placeins}
\usepackage{colonequals}
\usepackage{csquotes}

\usepackage{algorithm}
\usepackage{algpseudocode}
\algrenewcommand\textproc{\mathsf}
\algrenewcommand\algorithmiccomment[1]{\hfill\(\triangleright\)~#1}
\algrenewcommand\alglinenumber[1]{\scriptsize #1:}

\usepackage[table,x11names]{xcolor}
\usepackage{hyperref}
\usepackage[capitalize,nameinlink,noabbrev]{cleveref}

\crefname{theorem}{Theorem}{Theorems}
\crefname{lemma}{Lemma}{Lemmas}
\crefname{corollary}{Corollary}{Corollaries}
\crefname{definition}{Definition}{Definitions}
\crefname{problem}{Problem}{Problems}
\crefname{claim}{Claim}{Claims}
\crefname{remark}{Remark}{Remarks}
\crefname{observation}{Observation}{Observations}
\crefname{section}{Section}{Sections}
\crefname{subsection}{Section}{Sections}
\crefname{subsubsection}{Section}{Sections}
\crefname{algorithm}{Algorithm}{Algorithms}
\crefname{appendix}{Appendix}{Appendices}

\newcommand{\email}[1]{\href{mailto:#1}{\texttt{#1}}}

\newcommand{\den}{\mathsf{den}}
\DeclareMathOperator{\polylog}{polylog}
\DeclareMathOperator{\poly}{poly}

\DeclareMathOperator{\dist}{dist}

\newcommand{\eps}{\epsilon}
\newcommand{\seq}{\mathsf{seq}}
\newcommand{\R}{\mathbb{R}}

\newcommand{\bx}{\mathsf{box}}

\DeclareMathOperator{\E}{\textbf{E}}

\newcommand\restr[2]{{
		\left.\kern-\nulldelimiterspace
		#1
		\vphantom{\big|}
		\right|_{#2}
}}

\def\set#1{\ensuremath{\left\{#1\right\}}}
\def\Abs#1{\left|#1\right|}

\def\vec#1{\mathbf{#1}}

\def\veci{\mathbf{i}}
\def\vecj{\mathbf{j}}
\def\veck{\mathbf{k}}

\def\vecv{\mathbf{v}}

\def\vecx{\mathbf{x}}
\def\vecy{\mathbf{y}}

\definecolor{mygreen}{RGB}{190,255,190}
\definecolor{myred}{RGB}{255,220,220}
\definecolor{mygray}{RGB}{240,240,240}

\newcommand{\ignore}[1]{}

\title{Testing forbidden order-pattern properties on hypergrids}

\author{
	Harish Chandramouleeswaran\thanks{Chennai Mathematical Institute, India (\email{harishc@cmi.ac.in}).}
	\and
	Ilan Newman\thanks{Department of Computer Science, University of Haifa, Israel (\email{ilan@cs.haifa.ac.il}). Research supported by the Israel Science Foundation (grant 379/21).}
	\and
	Tomer Pelleg
	\and
	Nithin Varma\thanks{Department of Mathematics and Computer Science, University of Cologne, Germany (\email{varma@cs.uni-koeln.de}).}
}

\date{}

\begin{document}
	\maketitle
	
	\begin{abstract}
		Given a permutation $\pi \colon [k] \to [k]$, a function $f \colon [n]^d \to \mathbb{R}$ is said to be $\pi$-free if 
		there are no $k$ indices $x_1 \prec \dots \prec x_k \in [n]^d$ such that 
		$f(x_i) < f(x_j)$ and $\pi(i) < \pi(j)$ for all $i,j \in [k]$, where $\prec$ is the natural partial order over $[n]^d$. 
		For a fixed $\pi$ and $\epsilon \in (0,1)$, the problem of
		$\epsilon$-testing $\pi$-freeness is to distinguish  the case that $f$ is $\pi$-free
		from the case that at least $\epsilon n^d$ values of $f$ need to be
		modified in order to make it  $\pi$-free. When $k=2$, the problem is identical to monotonicity testing,
		which is extensively studied in property testing.
		The case of $k > 2$ has also received significant attention for functions $f \colon [n] \to \mathbb{R}$.
		
		We initiate a systematic study of testing of $\pi$-freeness for higher
		dimensional grids; specifically, for permutations of size $3$ over hypergrids of dimension $2$.
		We  design an adaptive one-sided error
		$\pi$-freeness tester with query complexity $O(n^{4/5 + o(1)})$ that works for all permutations of size $3$.
		We then argue that for general permutations of size $3$,
		every nonadaptive tester must have query complexity
		$\Omega(n)$, and further, that every \emph{adaptive} tester 
		must have query complexity $\Omega(\sqrt{n})$. This is the
		first general lower bound for testing $\pi$-freeness that is higher than
		$\Theta(\log n)$.
		
		For the important special case of $\pi = (1,2,3)$ or
		$(3,2,1)$, we  design a nonadaptive tester with a
		significantly improved (nearly optimal) query complexity of
		$\operatorname{polylog} n$. Such an exponential 
		gap in the complexity of testing freeness of nonmonotone and monotone patterns is not known, in general, in the one-dimensional case.
		
		Erasure-resilient (ER) monotonicity testers are important subroutines in the design of our 
		$\pi$-freeness testers. 
		We  design
		$\delta$-ER $\epsilon$-testers for
		monotonicity for functions $f \colon [n]^d \to \mathbb{R}$ 
		with query complexity $O\left(\frac{\log^{O(d)} n}{\epsilon (1-\delta)}\right)$, where $\delta \in (0,1)$ is an upper bound on the fraction of erasures. Earlier erasure-resilient monotonicity testers worked only when $\delta = O(\epsilon/d)$. The complexity of our nonadaptive monotonicity tester is nearly optimal as evidenced by a lower bound of Pallavoor, Raskhodnikova, and Waingarten (Random Struct. Algorithms, 2022).
		
		Lastly, we argue that the current techniques in function
		property testing cannot give us sublinear-query testers for patterns
		of length $4$ even for $2$-dimensional hypergrids.
	\end{abstract}

%%%%%%%%%%%%%%%%%%%%%%%%%%%%%%%%%%%%%%%%%%%%%%%%%%%%%%%%%%%%

\section{Introduction}

Order relations between values capture both local and global structural properties of numerical datasets. The study of specific order patterns in inputs is a fundamental topic spanning algorithms, combinatorics and applications. 
Knuth, in an early work~\cite{Knuth68}, showed that arrays that avoid the pattern $(2,3,1)$ are exactly the ones that can be sorted by a single stack. Recent research in the same spirit has focused on designing faster algorithms for optimization problems such as Euclidean TSP and searching in BSTs when inputs avoid certain order patterns~\cite{ChalermsookGKMS15,ChalermsookGJAPY23,ChalermsookPY24,BerendsohnKO24}.

There is also a long line of work in combinatorics on the number of permutations that avoid specific order patterns~\cite{Bona97, Bona99, Arratia99, Klazar00, AlonF00}, culminating in the celebrated Marcus-Tardos theorem~\cite{MarcusT04}. There are measures based on the frequencies of various order patterns such as permutation entropy\footnote{The presence or absence of specific order patterns is used to empirically determine whether a time series is stochastic or deterministic in nature. Specifically, the time series is inferred to be deterministic if it avoids some specific order pattern~\cite{Rosso07,ZuninoSR12} There are also dedicated libraries to perform various ordinal analyses on sequential datasets~\cite{BergerKSJ19,PessaR21}.}, which are extensively used in time series analysis~\cite{BandtP02} and have applications to varied fields such as biomedical sciences and econophysics among others~\cite{ZaninZRP12,Bandt16}.

Given a real-valued function \( f \colon [n] \to \mathbb{R} \) and a permutation (order pattern) \( \pi \colon [k] \to [k] \), we say that \( f \) contains a \emph{\(\pi\)-appearance} if there exist distinct indices \( x_1, \dots, x_k \in [n] \) such that the values of \( f \) respect the ordering dictated by \( \pi \). That is, for all \( i, j \in [k] \), \( f(x_i) < f(x_j) \) if and only if \( \pi(i) < \pi(j) \). The function \( f \) is said to be \emph{\( \pi \)-free} otherwise. 
There are several classical algorithms for deciding whether a given function is $\pi$-free for short patterns \( \pi \)~\cite{AlbertAAH01,AhalR08,BerendsohnKM2021}, including linear-time algorithms~\cite{GuillemotM14, Fox13}.
The quest for faster algorithms has led to the study of $\pi$-freeness in the property testing framework~\cite{RubinfeldS96, GoldreichGR98} of sublinear-time algorithms, although much less is known in this setting.
Investigations on forbidden order pattern testing has so far focused mostly on the case of real-valued
functions on the line $[n]$.
Formally, for a parameter \( \epsilon \in (0,1) \), a function \( f \colon [n] \to \mathbb{R}\) is \( \epsilon \)-far from being \( \pi \)-free if at least \( \epsilon n \) values of \( f \) must be modified to make it \( \pi \)-free. An \( \epsilon \)-tester for \( \pi \)-freeness is an algorithm that, given oracle access to \( f \), decides with constant probability whether \( f \) is \( \pi \)-free or \( \epsilon \)-far from \( \pi \)-free.
For general
constant length patterns and constant $\epsilon$,
the best known $\pi$-freeness testers have query complexity
$\widetilde{O}(n^{o(1)})$~\cite{NV22}.
In contrast, for the important special case of monotonicity testing (i.e., $(2,1)$-freeness testing), optimal \( \epsilon \)-testers~\cite{EKK+00, Fischer04, BGJRW12, CS13a, Belovs18} have query complexity \( \Theta\left(\frac{\log n}{\epsilon}\right) \).
Research on testing order pattern freeness has also resulted in the development of several interesting combinatorial techniques~\cite{NewmanRRS19, Ben-EliezerCLW19, Ben-EliezerC18, Ben-EliezerLW19, NV22}.

The significance of the study of short ordered patterns in one-dimensional functions is evident from the discussion above.
Many of the aforementioned problems extend naturally to higher-dimensional functions and are motivated by both theoretical interest and practical applications, as, for instance, in analyzing higher dimensional time series data~\cite{KellerL03,RibeiroZLSM12}. 
In the context of sublinear-time algorithms, monotonicity testing has been extensively studied for functions over general hypergrids, i.e., for $f \colon [n]^d
\to \R$, and there are optimal $\epsilon$-testers with complexity
$\Theta\left(\frac{d\log n}{\epsilon}\right)$~\cite{GGLRS00,DGLRRS99,CS13a,CS14,CDJS17}. However, little is known for testing general pattern freeness over domains other than the line, and in
particular, over hypergrids of dimension larger than $1$.

\subsection{Our results}

In this paper, we address this knowledge gap by initiating a
systematic study of pattern freeness testing of real-valued functions
over higher dimensional hypergrid domains. Let $\prec$ denote the natural  partial order over $[n]^d$, i.e., $x \prec y$ if, for all $i \in [d]$, we have $x[i] \leq y[i]$, where $p[i]$ refers to the $i^{\text{th}}$ coordinate of $p \in [n]^d$.
Given a permutation $\pi\in\mathcal{S}_k,$ a function $f \colon [n]^d \to \R$ is $\pi$-free if 
there are no $k$ indices $x_1 \prec \dots \prec x_k \in [n]^d$ such that 
$f(x_i) < f(x_j)$ and $\pi(i) < \pi(j)$ for all $i,j \in [k]$.
For a fixed $\pi$ and $\epsilon \in (0,1)$, the problem of
$\epsilon$-testing $\pi$-freeness is to distinguish  the case that $f$ is $\pi$-free
from the case that at least $\epsilon n^d$ values of $f$ need to be
modified in order to make it  $\pi$-free.  The problem is interesting in its own right and as a natural 
generalization of monotonicity testing for functions over the hypergrids. 

Our results can be summarized as follows. We develop techniques 
to design sublinear algorithms and strong lower bounds for testing $\pi$-freeness over hypergrids for all patterns $\pi \in \mathcal{S}_3$. It turns out that 
there are interesting new phenomena and significant new challenges to
deal with for patterns of length $3$ even over $2$-dimensional hypergrids. For testing freeness of larger order patterns, we argue that the current techniques in function
property testing cannot give strong sublinear-query testers even over $2$-dimensional hypergrids. Erasure-resilient as well as partial function monotonicity testing are important components
of our algorithms. Naturally, we design nearly optimal fully erasure-resilient monotonicity testers~\cite{DixitRTV18} over the hypergrids and efficient monotonicity testers over high dimensional partial functions.

\subsubsection*{Order-patterns over hypergrids} 	

We show that there are fundamental distinctions between nonmonotone and monotone patterns as well as between large and small patterns
for $\pi$-freeness testing over the hypergrids. We emphasize that this is very much unlike the case for the line $[n]$. We also show a 
distinction between adaptive and nonadaptive algorithms for $\pi$-freeness over hypergrids.
The bounds we obtain are summarized in \cref{table:results}.

\paragraph{Monotone vs.\ nonmonotone patterns.} Much of the research on $\pi$-freeness of functions $f \colon [n] \to \R$ has focused on the case of $\pi$ being monotone~\cite{NewmanRRS19, Ben-EliezerC18, Ben-EliezerCLW19, Ben-EliezerLW19}.
For constant $k$ and $\epsilon$, Ben-Eliezer, Letzter, and Waingarten~\cite{Ben-EliezerLW19} designed $\epsilon$-testers for $\pi$-freeness with query complexity $O(\log n)$ for monotone patterns $\pi$ of length $k$, i.e., for $\pi \in \{(1,2,\dots, k), \, (k,\dots,2,1)\}$. The best-known $\epsilon$-tester for general $\pi$-freeness, due to Newman and Varma~\cite{NV22}, has query complexity $\widetilde{O}(n^{o(1)})$ for constant $k$ and $\epsilon$.
Additionally, the only known general lower bound for $\epsilon$-testing freeness of patterns of \emph{any} length is $\Omega(\log n)$, which is the lower bound for monotonicity testing~\cite{Fischer04}.
Newman and Varma~\cite{NV22} conjecture that 
one can design $\pi$-freeness testers for all patterns $\pi$ of constant length with query complexity $\polylog n$.
They base their conjecture mainly on the fact that one can test $\pi$-freeness for all patterns $\pi$ of length $3$ with query complexity $\polylog n$~\cite{NewmanRRS19}. 

In this paper, we show that an analogous conjecture cannot hold for pattern-freeness testing over general hypergrids. Specifically, we prove that testing nonmonotone pattern freeness is exponentially harder than testing monotone pattern freeness \emph{even for patterns of length $3$ over $2$-dimensional hypergrids}. 

Our first main result is an $\Omega(\sqrt{n})$ lower bound for $\pi$-freeness testing
of patterns $\pi$ of length $3$ of functions $f \colon [n]^2 \to \R$. We emphasize that our lower bound holds even when the tester is allowed to make its queries adaptively.\footnote{A tester is \emph{nonadaptive} if its queries do not depend on the answers to previous queries. It is \emph{adaptive} otherwise.}

\begin{theorem} \label{thm:132_a_lb_intro}
	Any one-sided error\footnote{A tester has \emph{one-sided error} if it accepts a $\pi$-free function with probability $1$. It has \emph{two-sided error} otherwise.} $\epsilon$-tester for $(1,3,2)$-freeness of functions $f: [n]^2 \to \R$, has query complexity $\Omega(\sqrt{n})$, for every $\epsilon \leq 4/81$.
\end{theorem}

The above is the first nontrivial adaptive lower bound for pattern freeness that does not directly follow from the lower bounds for monotonicity testing~\cite{Fischer04,CS14}. We mention that our techniques can be generalized to higher-dimensional functions $f \colon [n]^d \to \mathbb{R}$ to show an $\Omega(n^{(d-1)/2})$ lower bound for general $\pi$-freeness of patterns $\pi$ of length $3$. The lower bounds that we obtain are exponentially stronger than the $\Omega(d\log n)$ bounds that directly comes out of monotonicity testing over $f \colon [n]^d \to \R$~\cite{Fischer04,CS14}.

In contrast, for monotone $3$-patterns, i.e., when $\pi = (1,2,3)$, or, equivalently, $(3,2,1)$, we design testers with exponentially better query complexity. Our testers have the additional feature of being nonadaptive. In particular, our result can be seen as a higher dimensional counterpart of the fact that testing monotone patterns of constant length can be done \emph{nonadaptively} using $\polylog n$ queries for real-valued functions over the line $[n]$~\cite{NewmanRRS19,Ben-EliezerCLW19}. However, proving our result requires very different ideas.

\begin{theorem} \label{thm:123_polylog}
	There exists a one-sided error nonadaptive $\epsilon$-tester with query complexity $\log^{O(1)} n$ for $(1,2,3)$-freeness of functions $f \colon [n]^2 \to
	\R$ for constant $\epsilon \in (0,1)$.
\end{theorem}

We remark that such a gap between testing nonmonotone and monotone pattern freeness is known for the line $[n]$ when restricted to \emph{nonadaptive testers}. Ben-Eliezer and Canonne~\cite{Ben-EliezerC18} showed that there are nonmonotone patterns $\pi$ of length $k$ such that every nonadaptive $\pi$-freeness tester requires $\Omega(n^{1-(1/(k-1))})$ queries. However, there are nonadaptive testers for all monotone patterns of length $k$ that has query complexity $\Theta((\log  n)^{\log k})$~\cite{Ben-EliezerCLW19}. 

\paragraph{Adaptivity vs.\ nonadaptivity.}
We also complement our lower bound in~\cref{thm:132_a_lb_intro} with efficient sublinear-time $\pi$-freeness testers for general patterns $\pi$ of length $3$. 
Using standard birthday paradox arguments, a simple $\epsilon$-tester 
with query complexity $O(n^{4/3}/\epsilon^{1/3})$ for testing $\pi$-freeness
of patterns $\pi$ of length $3$ of functions $f \colon [n]^2 \to \R$ (see~\cref{obs:32}) can be obtained. This is already a sublinear-query tester as the size of the domain is $n^2$.\footnote{This argument is more general and applies to longer patterns and higher dimensions as long as some additional requirements on the (Hamming) distance to pattern freeness are met. See \cref{sec:ham-del-intro} for more details.}\  Furthermore, the tester has one-sided error and is nonadaptive. However, this tester is only \emph{moderately} sublinear, and in particular, its query complexity is larger than $n$. 

We show that, in general, nonadaptive testers for $\pi$-freeness
of patterns $\pi$ of length $3$ of functions $f \colon [n]^2 \to \R$ cannot do significantly better. 

\begin{theorem} \label{thm:132_na_lb_intro}
	Any one-sided error nonadaptive $\epsilon$-tester for $\pi$-freeness of functions $f: [n]^2 \to \R$, for general $\pi \in \mathcal{S}_3$, has query complexity $\Omega(n)$.
\end{theorem}

Our lower bound technique can also yield a query complexity lower bound of $\Omega(n^{d-1})$ on the nonadaptive query complexity of testing $\pi$-freeness of functions $f \colon [n]^d \to \R$ for general $\pi \in \mathcal{S}_3$.

We also design $o(n)$-query adaptive testers for the same problem proving that adaptivity can and does significantly help $\pi$-freeness testing. 

\begin{theorem}\label{thm:3-pat-general}
	Let $\pi \in \mathcal{S}_3$. There exists a one-sided error
	$\epsilon$-tester with query complexity $O(n^{4/5 + o(1)})$ for $\pi$-freeness of functions $f \colon [n]^2 \to
	\R$ for constant $\epsilon \in (0,1)$.
\end{theorem}

Our result is in contrast to the case of monotonicity, where adaptivity is known to not help, in general, even over hypergrid domains~\cite{CS14}.
Moreover, these are the first nontrivial sublinear-query $\pi$-freeness testers for functions
$f \colon [n]^d \to \R$ and $\pi \in \mathcal{S}_3$ for $d > 1$.

\paragraph{Large vs.\ small order patterns.} For patterns of length $4$ or more, we provide evidence that the
current techniques cannot give us
sublinear-query testers for dimension larger than $1$. The reason is that all known $\pi$-freeness testers exploit the fact that functions that are far from being $\pi$-free have many disjoint $\pi$-appearances. We show in \cref{sec:ham-del-intro} that such a statement is not true in general for patterns of length $4$ or more over hypergrids. This is in stark contrast to the case of functions defined on the line, where this fact is crucially used to obtain sublinear-query testers for all constant-length patterns~\cite{NV22}.

\subsubsection*{Erasure-resilient monotonicity testing}

The \emph{erasure-resilient model} (henceforth referred to as the ER model), introduced by Dixit, Raskhodnikova, Thakurta, and Varma~\cite{DixitRTV18}, is the setting where some of the function values are \enquote{erased} by an adversary and cannot be
used by the algorithm. An ER algorithm gets an upper bound $\delta \in
(0,1)$ on the fraction of erased values. The actual erased points are
unknown and the algorithm learns whether the value at a specific domain point has been erased by the adversary only when that point is queried. Given parameters $\delta \in (0,1)$ and $\epsilon \in (0,1)$, a $\delta$-ER $\epsilon$-tester
for a property $\mathcal{P}$ of functions $\{f \colon D \to \R\}$ gets oracle access to a function $f$
with at most $\delta$-fraction of adversarially erased values in unknown places.
The tester must decide whether $f$ can be
completed to a function in $\mathcal{P}$  or whether every completion of $f$ 
must be changed in at least an $\epsilon$-fraction of the nonerased
domain in order to belong to $\mathcal{P}$.\footnote{A completion of a partially erased function is a total function agreeing with it on the nonerased points.}

Erasure-resilient monotonicity testers are crucial components of our pattern freeness testers. Dixit et al.~\cite{DixitRTV18} designed optimal $\delta$-ER
$\epsilon$-testers for monotonicity of functions $f \colon [n]^d \to
\R$ with complexity $O\left(\frac{d\log
	n}{\epsilon(1-\delta)}\right)$, {\em but}  these testers only work when the fraction of
erasures $\delta$ is $O(\epsilon/d)$. 
However, Pallavoor, Raskhodnikova, and Waingarten~\cite{PallavoorRW22}  proved that every nonadaptive $\delta$-ER $\epsilon$-tester for monotonicity of real-valued functions $f
\colon\{0,1\}^d \to \R$ on the Boolean hypercube must have query complexity $2^{d^\kappa}$ even
when $\delta = \Theta\left(\frac{\epsilon}{d^{0.5 - \kappa}}\right)$,
where $\kappa \in (0, 0.5)$.
We design
optimal ER-testers for functions $f \colon [n]^d \to \R$ that work for any fraction of erasures.

\begin{theorem}\label{thm:mon-ER}
	There exists a $\delta$-ER $\epsilon$-tester for monotonicity of functions $f \colon [n]^d \to \R$ with query complexity $O\left(\frac{\log^{O(d)} n}{\epsilon (1-\delta)}\cdot \log^2 \frac{1}{\epsilon}\right)$, for any $0 < \delta, \epsilon < 1$.
\end{theorem}

\subsubsection*{Monotonicity testing of partial functions}

Testing properties of partial functions is a problem closely related to erasure-resilient testing. A {\em partial} function is a function $f \colon P \to \R$, for a given (and known) subset $P \subseteq [n]^d$. For the line,
testing pattern freeness of partial functions is identical to
testing total functions (for any pattern), as the grid order induced on
$P$ is a total order (that is, identical to the line order). However,
for higher-dimensional grids ($d \geq 2$), the order induced on $P$ is
not identical to the grid order. In particular, while any tester for
total functions on the line can be applied to test properties of partial
functions with the same complexity (w.r.t.\ the domain size $\Abs{P}$),
this is no longer true for hypergrids. We design monotonicity
testers for partial functions on the $d$-dimensional grid with
query complexity $O(\log^{O(d)} n)$, which is independent of the size of $P$.
We emphasize that this result does not directly follow by treating the points in $[n]^d \setminus P$ as erased, and then using the erasure-resilient tester as a blackbox.

\begin{table}[htbp]
	\caption{Summary of the relevant results in pattern freeness testing}
	\label{table:results}
	\centering
	
	\renewcommand{\arraystretch}{1.2}
	\setlength{\tabcolsep}{4pt}
	\setlength{\arrayrulewidth}{0.6pt}
	\arrayrulecolor{black}
	
	\resizebox{0.85\textwidth}{!}{%
		\begin{tabular}{
				|>{\centering\arraybackslash}p{1.6cm}
				|>{\centering\arraybackslash}p{2.2cm}
				|>{\centering\arraybackslash}p{2.6cm}
				|>{\centering\arraybackslash}p{2.6cm}
				|>{\centering\arraybackslash}p{2.6cm}
				|>{\centering\arraybackslash}p{2.6cm}|
			}
			\hline
			\shortstack{\rule{0pt}{1.5em}Pattern\\length} &
			\shortstack{\rule{0pt}{1.5em}Pattern\\type} &
			\multicolumn{2}{c|}{\shortstack{$d = 1$\\(Line)}} &
			\multicolumn{2}{c|}{\shortstack{$d = 2$\\(Hypergrid)}} \\
			\cline{3-6}
			{} & {} & Nonadaptive & Adaptive & Nonadaptive & Adaptive \\
			\hline
			
			\rule{0pt}{3.5em}\shortstack{$k = 2$} &
			\multicolumn{5}{c|}{%
				\raisebox{0pt}[2.5em][2.5em]{%
					\begin{tabular}{@{}c@{}}
						$\Theta(\log n)$ \\[3pt]
						\cite{EKK+00,Fischer04,CS13a,CS14}
					\end{tabular}
				}%
			} \\ \hline
			
			\multirow{2}{*}{\rule{0pt}{5em}\shortstack{$k = 3$}} &
			\shortstack{Monotone\\patterns} &
			\raisebox{0pt}[2.7em][2.7em]{%
				\begin{tabular}{@{}c@{}}
					$\Theta(\log n)$ \\[3pt]
					\cite{Ben-EliezerCLW19}
				\end{tabular}
			} &
			\raisebox{0pt}[2.7em][2.7em]{%
				\begin{tabular}{@{}c@{}}
					$\Theta(\log n)$ \\[3pt]
					\cite{Ben-EliezerLW19}
				\end{tabular}
			} &
			\cellcolor{mygreen}\raisebox{0pt}[2.9em][2.9em]{%
				\begin{tabular}{@{}c@{}}
					$\operatorname{polylog}(n)$ \\[3pt]
					(\cref{thm:123_polylog})
				\end{tabular}
			} &
			\raisebox{0pt}[2.7em][2.7em]{\textbf{?}} \\
			\cline{2-6}
			
			& \shortstack{Nonmonotone\\patterns} &
			\raisebox{0pt}[3.8em][3.8em]{%
				\begin{tabular}{@{}c@{}}
					$\Theta(\sqrt{n})$ \\[3pt]
					\cite{Ben-EliezerC18, NewmanRRS19}
				\end{tabular}
			} &
			\raisebox{0pt}[3.8em][3.8em]{%
				\begin{tabular}{@{}c@{}}
					$\operatorname{polylog}(n)$ \\[3pt]
					\cite{NewmanRRS19}
				\end{tabular}
			} &
			\cellcolor{mygreen}\raisebox{0pt}[4.6em][4.6em]{%
				\begin{tabular}{@{}c@{}}
					$O(n^{4/3})$ \\[3pt]
					(\cref{obs:32}), \\[3pt]
					$\Omega(n)$ \\[3pt]
					(\cref{thm:132_na_lb_intro})
				\end{tabular}
			} &
			\cellcolor{mygreen}\raisebox{0pt}[4.6em][4.6em]{%
				\begin{tabular}{@{}c@{}}
					$O(n^{4/5 + o(1)})$ \\[3pt]
					(\cref{thm:3-pat-general}), \\[3pt]
					$\Omega(\sqrt{n})$ \\[3pt]
					(\cref{thm:132_a_lb_intro})
				\end{tabular}
			} \\ \hline
		\end{tabular}%
	}
\end{table}

\subsection{Related work}

The study of testing forbidden order-patterns  on real-valued arrays (i.e.,
real-valued functions on the line) was
initiated by Newman, Rabinovich, Rajendraprasad, and
Sohler~\cite{NewmanRRS19}. For monotone patterns $\pi$ of length $k$,
they designed a nonadaptive one-sided error $\epsilon$-tester for
$\pi$-freeness with query complexity ${({\epsilon}^{-1} \log
	n)}^{O(k^2)}$. For nonmonotone patterns, they gave an
$\Omega\left(n^{1 - \frac{2}{k+1}}\right)$ lower bound and an $O\left(n^{1 - \frac{1}{k}}\right)$ upper bound
in the nonadaptive setting, and an adaptive polylog-query tester for
$3$-patterns. Ben-Eliezer and
Canonne~\cite{Ben-EliezerC18} improved the bounds for nonadaptive
testing of nonmonotone
$k$-patterns to the essentially tight $\Theta\left({{n}^{1 - \frac{1}{k - 1}}}\right)$.
Subsequently,  the complexity of testing monotone patterns was improved in a
sequence of works~\cite{Ben-EliezerCLW19,Ben-EliezerLW19}, to
$O_{k,\epsilon}(\log n)$, which is optimal for constant $\epsilon$ even for the special case of testing $(2,1)$-freeness~\cite{Fischer04}. For general $k$-patterns $\pi$, Newman and Varma~\cite{NV22}
designed a ${\widetilde{O}}\left(n^{o(1)}\right)$-query adaptive
one-sided error $\epsilon$-tester for $\pi$-freeness.
ER-testing of general patterns on the line has not been considered per se, but the results in~\cite{NV22} can be made ER.

Yet another related direction is pattern freeness testing of functions over $[n]^d$ {\em with a constant-sized range}. Fischer and Newman \cite{FischerN07} designed a one-sided error $\epsilon$-tester with query complexity that depends only on $d$ and $\epsilon$, for arbitrary (including non-permutation) patterns of constant length. We mention that the dependence of the query complexity on $d$ and $\epsilon$ can even be doubly exponential in some cases. 

Lastly, we note that there is also a large amount of literature devoted to studying monotonicity ($(2,1)$-freeness) of Boolean
functions and real-valued functions over high-dimensional
domains~\cite{GGLRS00,DGLRRS99,FLNRRS02,CS16,CST14,CDST15,CWX17,KMS18,BCS18,PRV18,BCS20,BKR23,BKKM23,BCS23a,BCS23b},
but these are less relevant to the current work.

\subsection{Technical ideas and challenges}

All testing algorithms for $\pi$-freeness where $\pi$ is a permutation pattern,
i.e.,  $\pi \in \mathcal{S}_k$ where $\mathcal{S}_k$ is the symmetric
group on $k$ elements (including the restricted case of $\pi \in \mathcal{S}_2$ (i.e., monotonicity))
are crucially based on two ingredients.

The first ingredient is a relation between two different distances, the \emph{Hamming distance} and the \emph{deletion distance} (\cref{def:hamming}). All testers in the literature are designed to
distinguish between functions that are $\pi$-free (i.e., satisfy the property), and functions that are
far from being $\pi$-free with respect to the {\em deletion distance}.
However,  the standard definition of testing is with respect to the
Hamming distance. Fortunately, for the properties of $\pi$-freeness
on the line, and for monotonicity (in any dimension), these
two 
distances are equal. We show that for $d \geq 2$, this is not the
case in general, even for some $4$-patterns.  However, for $3$- and $2$-patterns, 
this equivalence still holds. This is partially the
reason why our results are limited to patterns of size $2$ and $3$.

The second ingredient for testing larger patterns is to perform some
sort of reduction to testing smaller patterns on partial domains, i.e., testing properties of partial functions. This is done
extensively in
\cite{NewmanRRS19,Ben-EliezerCLW19,Ben-EliezerLW19,NV22}. For a subdomain $P \subseteq [n]$ of the line, the order induced on $P$ is a
total order isomorphic to the order on a (smaller) line. Thus, when the domain is a line, testing partial functions is the same as testing total functions. However, for
$P \subseteq [n]^d, ~ d\geq 2$, the order on $P$ {\em is not}, in
general, isomorphic to a grid order, making the reduction more involved. Testing of partial
functions is a different problem in nature, which is of independent interest. 
When the density of the subdomain $P$ (i.e., $\Abs{P}/n^d$) is
a large constant, it may be regarded as an instance of
ER-testing. However, we use such reductions extensively
for low-density sets. Hence, we were required to develop an ER-monotonicity tester for partial
functions, and for (some special) partial functions for $3$-patterns. 
This, in turn, introduces some new challenges and techniques. We now describe
these in more detail for the $2$-dimensional case,
i.e., for functions $f:[n]^2 \to \R$.

\paragraph{Monotonicity testing.} 
For testing monotonicity of partial functions, we build on the idea of
~\cite{NewmanRRS19} which is referred to as \enquote{median-split and
	concatenate} (or the \emph{median argument} for short). The idea for the
line is as follows: suppose that a function $f$ is $\epsilon$-far from
$(1,2)$-freeness (that is, far from being monotone decreasing). The
median argument consists of sampling two disjoint intervals $L$ and $R$
on the line, where $L$ is on the `left' of $R$, and such that: (1) $f$
contains a large set of pairwise disjoint $(1,2)$-appearances  $A = \{(p_1, p_2) \in
L \times R, ~ f(p_1) < f(p_2) \}$, and (2)  one can sample a point $x$
such that for a large number of appearances $(p_1,p_2)$ in $A$, $f(p_1) <
f(x)$, and $f(x) < f(p_2)$.
Then, such a sampled point $x$ can be used to find a $(1,2)$-appearance of the type described above as follows:
sample a point $p \in L$, and a point $q \in R$, uniformly and independently at random. By (1) and (2), with high probability, $p$ and $q$ are the left and right coordinates of  some (not
necessarily the same) $(1,2)$-appearance in $A$, where $f(p) < f(x) < f(q)$. This is how a $(1,2)$-appearance $(p,q)$ is found.

Directly applying this argument for monotonicity on $[n]^d, ~ d\geq 2$
raises two difficulties.
Firstly, 
the probability of success (even for the line) is proportional to
$\epsilon^2$, which will result in a nonoptimal dependence on the
distance even for the $1$-dimensional case. Secondly, for partial functions
and in the ER setting, we must make sure that the value of the function
$f(x)$ for the `median point' $x$ described above is defined and not erased. We further note that for the argument above, one does not
need to know $x$, but merely assure its existence (w.h.p.). 
We show
that the median argument can be generalized to subgrids of $[n]^d$,
which will play the roles of the intervals $L$ and $R$ above. 
This can be done even for partial functions (regardless of the
density $\Abs{P}/n^d$), and in the presence
of any fraction of erasures (i.e., an ER-test).
To be able to use  
this in our tests for $3$-patterns, we require, and obtain, some extra guarantee from the
`median point' $x$ -- namely, that it is somewhat `typical' (see \cref{rem:needed-main}).

\paragraph{Testing $(1,2,3)$-freeness.} Testing $(1,2,3)$-freeness (and larger
monotone patterns) on the line was done in~\cite{NewmanRRS19}  by extending the
median argument as follows. 
It is shown that for two suitably chosen intervals 
$L$ and $R$ as above, 
there are many $(1,2,3)$-appearances that induce $(1,2)$-appearances in
$L$, and a corresponding $3$-appearance in $R$. 
Then, by the median argument, one can
find a suitable $(1,2)$-appearance in $L$, and complete it to a
$(1,2,3)$-appearance in $L \cup R$ by sampling a point in $R$. However, this is true only for specific ranges of distances
between points forming $(1,2,3)$-appearances\footnote{In later
	works, much more careful and elaborate partitions were used. But,
	these are also sensitive to the distances between the points forming the appearance.} and
cannot be directly generalized to $[n]^2$.

Instead,
we partition $[n]^2$ into $m^2$ blocks of $\left(\frac{n}{m}\right)^2$ size for a certain parameter $m$. Now, the
$(1,2,3)$-appearances in $f$ can be partitioned into a constant number of
configurations, according to the relative order of the blocks in which
the corresponding $1$-, $2$-, and $3$-legs fall. See e.g., \cref{fig_M3}. For some of these configurations, we can directly apply
the median argument and obtain a tester. For the others, a more sophisticated
argument is used to further partition a corresponding region and to
finally reduce the problem to monotonicity
testing. This is done by extending the
median argument to not only the corresponding values of the
points as described above, \emph{but also to the coordinates of the
	points}. These cannot be done simultaneously but can be achieved
using the `correct' order of sampling points, and the `typical nature of
the median point', as mentioned earlier. This is done, for instance, in
Case~\ref{step:3} of the algorithm given in \cref{sec:O(n)}.
Unlike in the $1$-dimensional case, two additional features to be taken care of are: (1) the reduction is to
monotonicity testing with subconstant distances, and (2) the reduced problem is on a partial domain which
is not isomorphic to a grid. 
Using, on top of this, a delicate recursion for an appropriate setting of $m$ brings the query complexity to $\polylog n$.

\paragraph{Testing $(1,3,2)$-freeness.} 
In this case, we build on the ideas of Newman and Varma~\cite{NV22} in which they designed sublinear-query $\pi$-freeness testers for general constant-length patterns $\pi$ for functions 
$f \colon [n] \to \R$.  Their basic tool is to look at the $2$-dimensional
`graph' of the function, and approximate it by a
$2$-dimensional grid of $m^2$ boxes which are $\{0,1\}$-tagged using what they call \textsf{Layering} and
\textsf{Gridding} (for a certain parameter $m \ll n$).
Then, using a theorem of Marcus-Tardos \cite{MarcusT04}, it
can be shown that one may assume that only $O(m)$ of the boxes are
nonempty (tagged `1'). Finally, with this reduced number of boxes, they
show that the forbidden appearances can be partitioned into a constant
number of configurations, each of which is testable by reducing to testing smaller forbidden patterns.

In this work, we follow the same grand scheme of ideas, adapted to $f \colon [n]^2 \to \R$. Now,
the `graph' of the function becomes a $3$-dimensional collection of points.
We partition the domain into $m^2$ blocks, and perform \textsf{Layering} and \textsf{Gridding} in a fashion similar to~\cite{NV22}. Since there is a corresponding generalization of the
Marcus-Tardos theorem~\cite{MarcusT04} to Boolean grids of the form $[n]^r$ given by Jang, Ne\v{s}et\v{r}il, and Ossona de Mendez~\cite{JNO21} (in our case, $r \colonequals d+1 = 3$), this implies that after $\widetilde{O}(m^2)$ queries, we either find a $\pi$-appearance, or
we partition the corresponding $3$-dimensional graph of the function into
a grid $[m]^3$ of boxes with at most $m^2$ nonempty boxes. Generally, we now
analyse the constantly many possible configurations as done in the
$1$-dimensional case, and design sublinear-query testers for each of those configurations. However, while doing so, we face
two new difficulties: In the $1$-dimensional case, the procedure \textsf{Gridding} results
in a grid $[m]^2$ of boxes with $O(m)$ nonempty boxes, and in which every row and
every column contains $O(1)$ nonempty boxes (after a simple
`cleaning'). This is crucial to carry out the reduction to testing patterns of smaller size. In
our case, the corresponding grid is a grid $[m]^3$ of boxes with
$O(m^2)$ nonempty boxes. Hence, on average, every column along any
dimension has $O(1)$ nonempty boxes. However, unlike in the $1$-dimensional
case, we can only guarantee this fact deterministically for every
vertical column, but not for horizontal columns.  This requires a
different strategy altogether.

An additional difficulty was hinted at previously for
monotonicity testing. As in~\cite{NV22}, and the testers for $(1,2,3)$-freeness in this work, most cases have complexity
growing with the parameter $m$, and this is true for the
\textsf{Gridding} procedure as well. However, there is one case -- namely, when every appearance has all
its points in one layer, in which one cannot reduce the testing to
monotonicity testing, and one can only guarantee
finding an appearance w.h.p.\ by directly sampling
$\poly(n)/\poly(m)$ points. Namely, the query complexity in this case is inversely proportional to $m$.
Optimizing w.r.t.\ $m$ still only results in an overall query
complexity of $\widetilde{O}(n)$. A way to decrease the complexity, which worked in~\cite{NV22} for a similar situation, was to use recursion in
the above special case, using smaller values of $m$. However, here,
the induced problem on a single layer (on which we would like to
perform recursion) is not identical to the original problem, but rather to a
test on a partial function on a domain of sublinear size. As
explained above, this is not the same as testing a total
function as the subdomain is not isomorphic to a rectangular grid anymore. A different kind of recursion has to be worked out, and this results 
in a complexity of $\widetilde{O}(n^{4/5+ o(1)})$.

\paragraph{Lower bounds for testing $(1,3,2)$-freeness.} To prove our lower bounds, the
high-level idea is to follow Newman, Rabinovich,
Rajendraprasad, and Sohler \cite{NewmanRRS19}, in which they proved an
$\Omega(\sqrt{n})$-query lower bound for {\em nonadaptively} 
testing $(1,3,2)$-freeness on the line $[n]$. We define a suitable \emph{intersection search} problem (\cref{problem:intersection_search-main}) for two \emph{$2$-dimensional} `monotone arrays' (w.r.t.\ $\prec$) and show a reduction to
to the problem of testing $(1,3,2)$-freeness on the
$2$-dimensional hypergrid. This reduction is blackbox, and as such, any lower bound for it
(nonadaptive or adaptive), implies a corresponding lower bound
for testing $(1,3,2)$-freeness. However, unlike in \cite{NewmanRRS19},
where the corresponding $1$-dimensional problem could be solved
adaptively in $O(\log n)$ queries, we show that our $2$-dimensional variant
requires $\Omega(\sqrt{n})$ queries, even for adaptive
randomized algorithms. The latter is shown conceptually by
observing that the intersection search problem `embeds'
in it independent instances of a variant of the \emph{birthday problem} (\cref{problem:birthday-search-main}).

\subsection{Conclusions and open directions}

We initiate a systematic study of order pattern freeness of real-valued functions over hypergrids. 
On the positive side, we design sublinear-query algorithms for testing $\pi$-freeness for patterns of length $3$ over
$2$-dimensional hypergrids. 
We also prove lower bounds to argue that our algorithms are almost optimal.
Additionally, we show that the equality between Hamming and deletion distances does not hold for larger patterns in general. This, in turn, implies that current ideas used in the testing literature are
insufficient to obtain sublinear-query $\pi$-freeness testers for larger patterns. 

In the following, we list a few additional interesting open directions that arise from our work.

\begin{enumerate}
	\item \textbf{Testing larger order patterns over high-dimensional hypergrids.} One natural direction is to design $\pi$-freeness testers for $\pi \in \mathcal{S}_3$ for hypergrids with dimension larger than $2$. We believe that our techniques might be generalizable to this case and we leave it as an open direction.
		As noted above, since the Hamming distance is not equal to the deletion distance, in general, for patterns of length $4$ over high-dimensional domains, current techniques do not result in sublinear-query testers.
	We leave open the testing of larger order patterns with respect to the Hamming distance for high-dimensional domains.
	Note that the question remains open even for the case of
	$2$-dimensional grids.
	
	\item \textbf{Testing monotone patterns.} 
	An open question is whether one can obtain efficient testers for monotone patterns of length $4$ or more over hypergrids of arbitrary dimension and improve our results. We remark (see \cref{clm:deletion}) that the Hamming distance and the deletion distance are equal for monotone patterns of any length, for functions over hypergrids of arbitrary dimension.
	
	\item \textbf{Optimality.} Determining the optimality of our $\pi$-freeness testers is a well-motivated problem. It is also important to know
	whether one can design ER monotonicity testers with better query complexity over high-dimensional hypergrids that work for all fractions of erasures. Additionally, lower bound techniques developed for pattern freeness might find use in proving lower bounds for other related combinatorial problems, as evidenced in the $1$-dimensional case~\cite{NV20}. 
	
	\item \textbf{Non-permutation patterns.} We consider only
	order-patterns that are permutations in this paper. It is equally interesting to 
	consider testing freeness of forbidden patterns that are not permutations. 
	In the $1$-dimensional case, any order pattern is a permutation, due
	to the fact that any subdomain of $[n]$ is a total
	order. However, for $[n]^d$, ~$d \geq 2$, there are order
	patterns (of length $3$ and higher) that are not permutations.
	The Hamming distance is not equal to the deletion distance even for such patterns of length $3$, in general. We briefly discuss testing freeness of such patterns in \cref{sec:last}.
\end{enumerate}

\subsection{Organization}

The paper is organized as follows.
\cref{sec:prelim} contains preliminary notations and claims regarding the relationship between the deletion and Hamming distances. Our erasure-resilient and partial function monotonicity testers for real-valued functions over hypergrids are presented in \cref{sec:mono-start-main}.
\cref{sec:123} contains the description of the $\polylog n$-query tester for $(1,2,3)$-freeness. Our $O(n^{4/5 + o(1)})$-query tester for $(1,3,2)$-freeness is presented in \cref{sec:132}. 
Lastly, \cref{sec:132_lbs-main} contains the adaptive and nonadaptive lower bounds for $(1,3,2)$-freeness.
Many of the proofs, technical details, and pseudocodes are deferred to the appendices (\crefrange{sec:mono-start}{sec:hamdel}).
In particular, \cref{sec:gridding} contains the procedures \textsf{Gridding} and \textsf{Layering} that are crucial components of our $(1,3,2)$-freeness testers. 

%%%%%%%%%%%%%%%%%%%%%%%%%%%%%%%%%%%%%%%%%%%%%%%%%%%%%%%%%%%%

\section{Preliminaries}\label{sec:prelim}

For a function $f \colon [n]^d \to \R$, we denote the
range of $f$ by $R(f)$. The elements of the domain $[n]^d$ of $f$ are referred
to as \emph{indices} or \emph{points} and the elements of the range $R(f)$ are referred to as \emph{values} or $f$-values. For an index $x \in [n]^d$, we use $x_i$ to denote the $i$-th coordinate of $x$.
For two distinct indices $x, y \in [n]^d$, we say that $x \prec y$ if $x_i \leq y_i$ for all $i \in [d]$.

Let $\pi \in \mathcal{S}_k$ be a permutation of length $k$. The
permutation (also called an order pattern or simply a
pattern) is monotone if it is either $(1,2,\dots, k)$ or $(k, k-1,
\dots ,1)$. It is nonmonotone otherwise. The function $f \colon [n]^d \to
\R$ has a $\pi$-appearance if for some points $x^{(1)}
\prec \dots \prec x^{(k)}, $ for all $i, j \in [k]$, $\pi(i) <
\pi(j)$ implies that $f(x^{(i)}) < f(x^{(j)})$. 
Namely, the order $f$
induces on $\{f(x^{(i)})\mid ~ i \in [k]\}$ is identical to the order $\pi$
induces on $[k]$. 
For $i \in [k]$, the pair $(x^{(i)}, f(x^{(i)}))$ is called the 
$i$-th leg of the $\pi$-appearance and the pair $(x^{(\pi^{-1}(i))}, f(x^{(\pi^{-1}(i))}))$ is called the $i$-leg of the appearance.
For example, if $x_1, x_2, x_3$ are the indices of a $(1,3,2)$-appearance, then $(x_2, f(x_2))$ is the $3$-leg and $(x_3, f(x_3))$ is the $2$-leg.
The function $f$ is $\pi$-free if it has no $\pi$-appearance. Let $P^d_{\pi}$ denote the
property of $\pi$-freeness of functions $f \colon [n]^d \to
\R$. Thus, $f$ is monotone nondecreasing if it satisfies $P_\pi^d$ for $\pi=(2,1)$.

Most of the algorithms in this paper are for functions over the domain $[n]^2$. We use $x(p)$ and $y(p)$ to denote the $x$ and $y$ coordinates of the point $p \in [n]^2$, respectively. 

For ease of presentation, all of our algorithms are designed to have success probability $1 - 1/n^{\Omega(\log n)}$, and we ignore the dependence on polylogarithmic factors in $n$ in their query complexities.

\subsection{Deletion and Hamming distances}\label{sec:ham-del-intro}

A major tool in forbidden order pattern testing is the connection
between the deletion distance and the Hamming distance.

\begin{definition}[Deletion and Hamming distances]\label{def:hamming}
	Let $f \colon [n]^d \to \R$ and $\pi$ be an order pattern. The deletion distance of $f$ from being $\pi$-free
	is $\min \{|S|: ~ S \subseteq [n]^d, ~
	f|_{[n]^d\setminus S} ~ \mbox
	{is } \pi\mbox{-free} \}$. Namely, it is the cardinality of the smallest set $S \subseteq
	[n]^d$ that intersects each $\pi$-appearance in $f$.
	The Hamming distance of $f$ from being $\pi$-free is the minimum of $\dist(f,f') = |\{i:~ i \in [n]^d,~ f(i) \neq f(i') \}|$ over all
	functions $f': [n]^d \to \R$ that are $\pi$-free.
\end{definition}

For the $1$-dimensional case, namely for $f \colon [n] \to \R$  and 
any $\pi \in {\mathcal S}_k$, the Hamming distance of $f$ to $P_{\pi}^1$ is
equal to the deletion distance of $f$ to $P_{\pi}^1$. This fact is the
basis for every known tester for pattern
freeness in the $1$-dimensional case.
However, for $d \geq 2$ this equality is not true in general. Moreover, for
some patterns (and functions), there
is an unbounded gap between the two distances. We show that equality
between the two distances does hold for some
patterns, and, in particular, for all
patterns $\pi \in {\mathcal{S}}_k, ~ k \leq 3$.

\begin{claim}\label{clm:deletion}
	Let $k \geq 1$ and $\pi \in {\mathcal S}_k$ with $\{\pi(1),\pi(k)\}
	\cap \{1,k\} \neq \emptyset$.  Then the Hamming distance of $f \colon [n]^d \to \R$ to $P_{\pi}^d$ is equal to the deletion distance of $f$ to $P_{\pi}^d$.
\end{claim}

The proof of \cref{clm:deletion} can be found in \cref{sec:hamdelproof}.

\begin{corollary} \label{cor:hamdel3}
	Let $\pi \in {\mathcal{S}}_k, ~k \leq 3$. Let $n,d \geq 1$. The Hamming distance of $f \colon [n]^d \to \R$ to $P_{\pi}^d$ is equal to the deletion distance of $f$ to $P_{\pi}^d$.
\end{corollary}

\begin{proof}
	The claim holds for $(1,2) \in {\mathcal{S}}_2$, $(1,2,3) \in {\mathcal{S}}_3$, and $(1,3,2) \in {\mathcal{S}}_3$, by \cref{clm:deletion}. All the other elements of ${\mathcal{S}}_2$ and ${\mathcal{S}}_3$ are equivalent to one of these via looking at the domain ${[n]}^d$ upside-down, or via multiplying all $f$-values by $-1$. 
\end{proof}

Due to \cref{cor:hamdel3}, henceforth, for patterns of length at most $3$, we simply say that $f$ is $\epsilon$-far from $\pi$-freeness without specifying the distance measure.

Unfortunately, a statement analogous to \cref{cor:hamdel3} does not hold in general for all
patterns of length $4$ or higher, even when $d = 2$. This is a
significant point of departure from pattern freeness on arrays
(i.e., $d = 1$) where such a statement holds for every
pattern. We show an example of a $4$-pattern $\pi$ and a
function $f \colon [n]^2 \to \R$ for which the multiplicative gap
between the Hamming and deletion distances from $\pi$-freeness is
$\Omega(n)$. This example can be straightforwardly generalized to longer patterns and for larger values of $d$.

\begin{claim} \label{clm:hamneqdelexample}
	There exists a $\pi \in {\mathcal{S}}_4$ and a function $f \colon [n]^2 \to \R$ such that the deletion distance of $f$ from $\pi$-freeness is $1$ and the Hamming distance of $f$ from $\pi$-freeness is $\Theta(n)$.
\end{claim}

The proof of \cref{clm:hamneqdelexample} can be found in \cref{sec:last}.

A {\em matching} of
$\pi$-appearances in $f$ is a collection of $\pi$-appearances 
that are pairwise disjoint as sets of indices in $[n]^d$. 
The following claim is folklore and immediate from the fact that the size of a minimum
vertex cover of a $k$-uniform hypergraph is at most $k$ times the 
cardinality of a maximal matching.

\begin{observation}\label{obs:matching}
	Let $\pi \in \mathcal{S}_k$. If  $f \colon [n]^d \to \R$  is $\epsilon$-far from being
	$\pi$-free both in Hamming and deletion distances, then there exists a matching of $\pi$-appearances of size at least $\epsilon n^d /k$. 
\end{observation}

We also need the following (folklore) lemma that generalizes the
birthday paradox. 
\begin{lemma}\label{lem:31}
	Let $N$ be a set and $M$ be a set of pointwise disjoint $k$-tuples on $N$ with $k=O(1)$. Then, a uniformly random sample (with repetitions) of $\frac{|N|}{|M|^{1/k}}$
	points contains a member of $M$ with probability $\Theta(1)$.
\end{lemma}

A direct consequence of \cref{obs:matching} and \cref{lem:31} is that for any $\pi \in \mathcal{S}_k$ and any
$d$, if the deletion distance of any $f:[n]^d \to \R$ from $P_{\pi}^d$
equals (or is of the order of) the Hamming distance of $f$ from
$P_{\pi}^d$, then $P_{\pi}^d$ can be tested with sublinear 
query complexity (in the domain size), as stated below. In particular,
this is true for monotone $\pi$ of any constant size. Our goal, in the rest of the
paper, is to improve this much further for $k \leq 3$.

\begin{observation}
	\label{obs:32}
	Let $\pi \in \mathcal{S}_k$. If for every 
	$f \colon [n]^d \to \R$, the deletion and the Hamming distances of $f$ from
	$P_{\pi}^d$ are equal, then there is a nonadaptive one-sided error $\epsilon$-tester for $P_{\pi}^d$ that makes
	$O(n^{d(1- 1/k)}/\epsilon^{1/k})$ queries. 
	In particular, this is true for every monotone
	$k$-permutation and every $d$. Additionally, this is true for every
	permutation $\pi \in \mathcal{S}_3$ and $d = 2$, implying a $O(n^{4/3}/\epsilon^{1/3})$-query tester.
\end{observation}

\section{Erasure-resilient monotonicity testing over high dimensions}\label{sec:mono-start-main}
In this section, we describe our erasure-resilient (ER) monotonicity tester for real-valued functions over hypergrids of arbitrary dimension. 
Such a tester for $2$-dimensional
hypergrids  is a key subroutine in our pattern freeness testers in~\cref{sec:123} and~\cref{sec:132}.

\begin{theorem} \label{thm:improved-mon-er-main}
	There is a one-sided error nonadaptive $\delta$-ER $\epsilon$-tester for monotonicity of functions $f \colon [n]^2 \to \R$ making $O\left(\log^{O(1)}(n)\cdot \dfrac{\log^2 (1/\epsilon)}{\epsilon (1-\delta)}\right)$ queries.
\end{theorem}

\cref{thm:improved-mon-er-main} asserts the correctness and complexity
of an ER monotonicity tester, which serves our purpose for $(1,2,3)$-freeness testing.
However, for testing $(1,3,2)$-freeness, we need the ability to test monotonicity of 
partial functions $f:P \to \R$, where $P$ is an arbitrary subset of $[n]^2$.
The subdomain $P$ is known to the tester. 
An ER monotonicity tester as described above can be used for 
partial function monotonicity testing. However, $P$ can be arbitrarily small in many cases and 
the inverse dependence of the query complexity on $|P|$ (from \cref{thm:improved-mon-er-main}) is not desirable. 
To address this issue, we design an alternate tester whose query complexity
does not depend on $\Abs{P}$. We stress that the tester crucially depends on the fact that $P$ is known (unlike in the erasure-resilient case).

\begin{theorem} \label{thm:improved-mon-partial-main}
	There is a one-sided error nonadaptive $\epsilon$-tester for monotonicity of partial functions $f \colon P \to
	\R$, where $P \subseteq [n]^2$, with query complexity
	$O\left(\log^{O(1)}(n)\cdot \dfrac{ \log^2 (1/\epsilon)}{{\epsilon}}
	\right)$.
\end{theorem}

The proofs of \cref{thm:improved-mon-er-main} and \cref{thm:improved-mon-partial-main} can be found in \cref{sec:mono} and \cref{sec:simple-mon2}.

For using the monotonicity testers described above in the $3$-pattern testing, we require
more. We need the testers to output not just an arbitrary violation to monotonicity, but rather a \enquote{typical} one as we remark below.

\begin{remark}\label{rem:needed-main}
	Let $M$ be a matching of $(1,2)$-pairs on a domain $D$. \cref{thm:improved-mon-er-main} describes a tester for monotonicity that finds a violation $(p_1,p_2)$ with
	$p_1$ and $p_2$ being the $1$-leg and the $2$-leg of some pairs in $M$ (but not
	necessarily a pair by itself) with constant probability.
	The tester induces a probability distribution
	$P_M$ on the $1$- and $2$-legs of the pair it produces and in
	particular, a distribution on the $1$-leg $p_1$. It is easy to see
	that this probability is uniform over the $1$-legs of pairs in $M$.
\end{remark}

\begin{remark}
	Our monotonicity testers work even when the input function is over a rectangular grid $[n_x]\times [n_y]$, and the dependence of the query complexity on the input size is then $\polylog (\max\{n_x,n_y\})$.
\end{remark}

\begin{remark}
	Our monotonicity testers can be generalized to functions defined on higher-dimensional hypergrids $[n]^d$, for any $d \geq 1$, thereby completing the proof of~\cref{thm:mon-ER}. The details can be found in \cref{sec:mon-general-d}.
\end{remark}

\section{Testing \texorpdfstring{$(1,2,3)$}{(1,2,3)}-freeness \texorpdfstring{of functions $f \colon [n]^2 \to \R$}{}}\label{sec:123}

Our main goal in this section is to present a
$\polylog n$-query, nonadaptive one-sided error
$\epsilon$-tester for $(1,2,3)$-freeness of real-valued functions on the $2$-dimensional grid $[n]^2$. 

\begin{theorem}\label{thm:3-pat-no-induction}
	There exists a nonadaptive one-sided error $\epsilon$-tester for $(1,2,3)$-freeness of functions $f \colon [n]^2 \to
	\R$ that uses $\log^{O(1)} n$ queries for any constant $\epsilon \in (0,1)$.
\end{theorem}

We first present some required preliminaries. Then, we describe a
simpler $\widetilde{O}(n)$ tester which is later generalized to the final tester.

\subsection{Preliminaries for the \texorpdfstring{$(1,2,3)$}{(1,2,3)}-freeness tester}\label{sec:prelim-123}
Our goal is to design a one-sided error tester for
$(1,2,3)$-freeness. That is, we are only going to reject upon finding a
$(1,2,3)$-appearance. Thus, in what follows, we may assume that $f$ is $\epsilon$-far from being
$(1,2,3)$-free.
By~\cref{obs:matching}, $f$ contains a matching $M$ of $(1,2,3)$-appearances
of size $\Omega(\epsilon n^2)$.

\subsubsection{Grid of boxes} \label{sec:grid_of_boxes}

For a fixed parameter $m \leq n$, let $G_m^{(2)}$ be  a partition of $[n]^2$ into $m \times m$ boxes of size
$n^2/m^2$ each, in the natural way. These
boxes are arranged naturally in rows and columns in $G_m^{(2)}$. Here, a row refers to a set of points defined by a range of $y$-coordinates and a column refers to a set of points defined by a range of $x$-coordinates. 
For ease of description, we refer to the direction of increase of the $x$-coordinate as being from left to right and the direction of increase of the $y$-coordinate as being from bottom to top of the grid $[n]^2$.

A major tool that we use is to partition $M$ into a relatively small number of certain types of matchings, and try to find
a $\pi$-appearance in one of the types.
This is done as follows. Let $\bar{p}=(p_1,p_2,p_3)$ be a $\pi$-appearance in $M$. 
We say that $\bar{p}$ is in $M_0$ if  the legs of $(p_1,p_2,p_3)$ all belong to the same box in $G_m^{(2)}$.
The point $\bar{p}$ is in $M_1$ if its legs belong to boxes that are all in the same row (or column) in $G_m^{(2)}$, but not all in the same box.  The matching $M_2$  contains all appearances $\bar{p}=(p_1,p_2,p_3)$  where the boxes containing $p_1$ and $p_2$ do not share a row or column in $G_m^{(2)}$ (see \cref{fig_M2}). Finally,
$\bar{p}$ is in $M_3$ if $p_1,p_2$ share a row (or a column), and
$p_2,p_3$ share a column (or a row), altogether spanning $2$ rows
and $2$ columns (see \cref{fig_M3}).
With these definitions, the matching $M$ is partitioned into $M= M_0
\cup M_1 \cup M_2 \cup M_3$.
We note that the cases where we replace $p_1$
with $p_3$ are similar to the cases above. For instance, the case where $p_3$ is not
in the same row or column as $p_2$ is identical to $M_1$. Thus, these
cases cover all possible configurations of $(1,2,3)$-appearances.
Since  $|M| \geq \epsilon n^2$, it must be the case that 
\begin{enumerate}
	\item either $M_0$ has cardinality at least $\epsilon n^2 \cdot (1 - \frac{1}{\log n})$
	\item or one of $M_1, M_2, M_3$ has size at least $\epsilon n^2 /(3 \log n)$.
\end{enumerate}
We will now show how to develop an algorithm for each case.\footnote{It
	would be more natural to assume that one of the 4 cases has size at
	least $\epsilon n^2/4$. This will be enough for the $\widetilde{O}(n)$ algorithm, but
	this is not good enough for the $\polylog n$-query algorithm. Hence, we use this asymmetric assumption.}

\subsection{A \texorpdfstring{$\widetilde{O}(n)$}{O~(n)} tester for \texorpdfstring{$(1,2,3)$}{(1,2,3)}-freeness}\label{sec:O(n)}

The tester gets $\epsilon \in (0,1), n \in \mathbb{N}$ as inputs and has oracle access to a function $f \colon [n]^2 \to \mathbb{R}$.
Let $m \colonequals \sqrt{n}$, and let $G_m^{(2)}$ be the
corresponding $2$-dimensional grid of boxes.
The tester runs all of these procedures and rejects if any of them finds a $(1,2,3)$-appearance.

\subsubsection{\texorpdfstring{$M_0$}{M0} is large} \label{sec:M0_large-main}

We assume that $|M_0| \geq \epsilon n^2 \cdot (1 - \frac{1}{\log n})$. 
Let $\epsilon^{(i)} = \epsilon \cdot (1 - \frac{1}{\log n})^i$. 
A uniformly random box, in expectation, contains a matching of
$(1,2,3)$ appearances of size at least $\epsilon^{(1)} n^2/m^2$. By the
reverse Markov inequality, at least $\epsilon/(2\log n)$ fraction of boxes
each contain matchings of cardinality at least $\epsilon^{(2)}
n^2/m^2$.  
That is, if we sample $\Theta(\log^2 n)$ boxes uniformly and independently at random, then with probability at least $1 - n^{-\Omega(1)}$, one of them contains a matching of cardinality at least $\epsilon^{(2)} n^2/m^2$. Therefore, querying all the points in each of these boxes, will be sufficient to find a $(1,2,3)$-appearance. The query complexity is clearly $\widetilde{O}(n^2/m^2)$, which is $\widetilde{O}(n)$ by our setting of $m$.

\subsubsection{\texorpdfstring{$M_1$}{M1} is large} \label{sec:M1_large-main}

We assume that $M_1$ has cardinality at least
$\epsilon n^2 /(3 \log n)$.  Without loss of generality, we may assume
that each of these appearances have all their legs belonging to the
same row.  There are two main possibilities depending on whether the
legs occupy $2$ or $3$ boxes in the row.  Here, we outline the cases
where the legs occupy $2$ boxes. The case of the legs occupying $3$
boxes is similar and is omitted. Further, it could be that  the $1$-
and $2$-legs belong to the same box or the $2$- and $3$-legs belong
to the same box. We only consider the former case as the arguments to deal
with the latter case are analogous.

For two domain points $u, v \in [n]^2$, the $y$-spacing between $u$
and $v$ is the distance between the $y$-coordinates of $u$ and
$v$. The matching of $(1,2,3)$-appearances in which the $1$- and $2$-legs
are in the same box, but the $3$-leg is in a different box, can be further
split into these appearances for which 
the $y$-spacing between $2$- and $3$-legs is more than that between the $1$-
and $2$-legs. We denote this submatching by $M_{11}$, and denote by $M_{12}$
the submatching
for which the $y$-spacing between $1$- and $2$-legs in the
$(1,2,3)$-appearance is more than that between the $2$- and $3$-legs (see \cref{fig:M11_M12_2boxes}).

\begin{figure}[ht]
	\centering
	\includegraphics[scale=0.6]{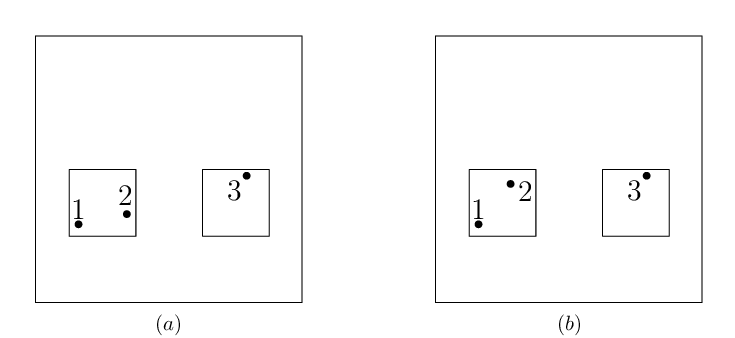}
	\caption{(a) $M_{11}$: $\Delta_y(2\text{-leg},3\text{-leg}) \geq \Delta_y(1\text{-leg},2\text{-leg})$ (b) $M_{12}$: $\Delta_y(2\text{-leg},3\text{-leg}) < \Delta_y(1\text{-leg},2\text{-leg})$}
	\label{fig:M11_M12_2boxes}
\end{figure}

\paragraph{The easy case: $M_{11}$ is large.} The subcase where $M_{11}$ is large (i.e., $\Abs{M_{11}} \geq \epsilon n^2/(c \log n)$ for some absolute constant $c$) is easy to handle.  At a high level, we can partition $M_{11}$
into logarithmically many buckets according to the $y$-spacing
between the $2$- and $3$-legs, much like in
\cite{NewmanRRS19}. For the right bucket, it is easy to first sample a box $B$ containing \emph{sufficiently many} $3$-legs and then
run a $(1,2)$-freeness tester to find a corresponding $(1,2)$-appearance in the same row as $B$. The existence of many such $(1,2)$-appearance
is guaranteed by the `median' argument.
A detailed discussion is deferred to the appendix (\cref{sec:M11_large}).

\paragraph{The hard case: $M_{12}$ is large.}
Suppose that the cardinality of $M_{12}$ is at least $\epsilon' n^2$,
where $\epsilon'$ is $\epsilon/(c \log n)$ for some absolute constant
$c$. The reason that this case is harder is because sampling an intended
$3$-leg in $B$ as above would determine a relatively small region for
the $2$-leg, but the $1$-leg could be in a much larger region, and so,
the corresponding $(1,2)$-appearances have too small a density.  For
this case, we will need to resort to the additional `typicality'
property of the output of the $(1,2)$-freeness tester.

We start by sampling one box $B$. With probability at least $\epsilon'/2$, a uniformly random box has at least $(\epsilon'/2) \cdot (n^2/m^2)$ many $(1,2)$-appearances. 
Thus, with probability at least $1-n^{-\Omega(1)}$, one out of $\widetilde{O}(1/\epsilon)$ sampled boxes satisfies this.
We condition on sampling such a box $B$. 

Let $R$ be the row containing $B$ and let $R'$ be the part of the row $R$ to the right of $B$. 
For each $k \in [\log n]$, we consider equipartitioning $R$ into $n/(2^k m)$ horizontal strips, each being a grid isomorphic to $[n] \times [2^k]$. 
Let $k' \in [\log n]$ be the scale that contains the largest number of the $(1,2,3)$-appearances 
whose $2$- and $3$-legs are on adjacent strips.  
That is, the row $R$ restricted to this scale contains a matching $M'$
of $(1,2,3)$-appearances of size at least $(\epsilon'/(2\log n)) \cdot
(n^2/m^2)$, where the $2$- and $3$-legs are on adjacent strips (but the
corresponding $1$-leg could be anywhere). Hence, a random strip of
this scale will contain at least $\epsilon'/(2\log n)$ fraction of
points that are the $2$-legs of appearances of that scale in $M_{12}$.

Now the argument is as follows: we first apply the monotonicity tester with proximity parameter\footnote{A monotonicity tester $T$ with \emph{proximity parameter} $\epsilon$ distinguishes the case that the queried function $f$ is monotone from the case that $f$ is $\epsilon$-far from monotone, with probability at least $2/3$.} $\Theta(\epsilon'/\log n)$
to $B$ to find a $(1,2)$-appearance corresponding to
$(1,2,3)$-appearances from $M_{12}$ of the fixed scale above. Further,
by the typicality feature stated in \cref{rem:needed-main}, this $2$-leg is
picked uniformly at random from all $2$-legs in the submatching of
$M_{12}$ restricted to the above fixed scale, and to the strip where
the $2$-leg is. In particular, the strip that contains the point
picked as the $2$-leg is a random strip among the strips of the above
fixed scale, and the value of the $2$-leg is below the median of all
$2$-legs of that scale in that strip.

Assuming the above, now the situation is much simpler: since by the
assumption above the corresponding strip contains at least $\epsilon ' /(4\log
n)$ fraction of points that are $2$-legs of corresponding
$(1,2,3)$-appearances, and assuming the median feature for the actual
picked $2$-leg, above, there are at least $\epsilon '/(8m \log n)$
fraction of points in the next strip of the same fixed scale and to
the right of $B$ that are the corresponding $3$-legs, and with value
above the value of our picked $2$-leg. Selecting a random point in
this strip will find one resulting in a $(1,2,3)$-appearance. Again,
working out the details (and using amplification to reach a high
success probability) results in a $\widetilde{O}(m/\epsilon^2)$-query algorithm for this
case. The pseudocode for this case can be found in the appendix (\cref{alg:M12-large-123}).

\subsubsection{\texorpdfstring{$M_2$}{M2} is large} \label{sec:M2_large-main}

$M_2$ is the matching of $(1,2,3)$-appearances where the boxes containing the $1$- and $2$-legs do not share a row or column. The discussion of the case where $M_2$ is large, i.e., $\Abs{M_2} \geq \epsilon n^2/(3\log n)$, is deferred to the appendix (\cref{sec:M2_large}).

\subsubsection{\texorpdfstring{$M_3$}{M3} is large}\label{step:3} 

In this case, the only possibilities are that the three boxes containing the legs of a $(1,2,3)$-appearance 
form a $\Gamma$-shape or an inverted $\Gamma$-shape, as shown in \cref{fig_M3}. The treatment of both cases is similar, so we assume that a majority of the $(1,2,3)$-appearances in
$M_3$ form a $\Gamma$-shape: that is, the $1,2$-legs are on the same column, and the $2,3$-legs are on the same row, as depicted in \cref{fig_M3}(a).

\begin{figure}[ht]
	\centering
	\includegraphics[scale=0.6]{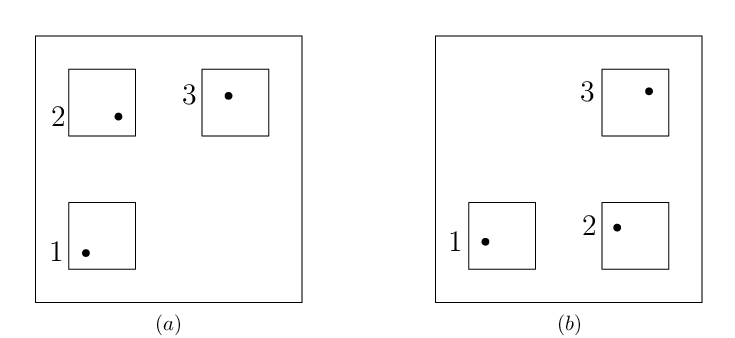}
	\caption{$M_3$: (a) $(1,2,3)$-appearance in a $\Gamma$-shape (b) $(1,2,3)$-appearance forming an inverted $\Gamma$-shape}
	\label{fig_M3}
\end{figure}

Consider a partition of $M_3$ as follows. For $\ell = 0,1,\ldots ,\log n$, let $(p_1,p_2,p_3) \in
M_3$ belong to $M^{(\ell)}$ for the smallest $\ell$ such that $x(p_2) -
x(p_1) < 2^\ell$, where $x(p)$ denotes the $x$-coordinate of the
point $p$.

Recall that by our assumption, $|M_3| \geq \epsilon n^2/(3 \log n)$. Let $\ell^* \in \{0,1,\ldots ,\log n\}$ be such that $|M^{(\ell^*)}| = \Omega(\epsilon
n^2/(\log^2 n))$.
We first consider the case when $\ell^* = 0$. That is, for $\Omega(\epsilon
n^2/(\log^2 n))$ many $(1,2,3)$-appearances in $M^{(\ell^*)}$  the $x$-coordinate of the $1$- and $2$-legs are equal. Now, a uniformly random box $B$ will contain the $2$-legs of $\Omega(\epsilon
n^2/(m^2\log^2 n))$ many $(1,2,3)$-appearances from $M^{(\ell^*)}$, in expectation.
In other words, with probability $\widetilde{\Omega}(\epsilon)$, a uniformly random box will contain the $2$-legs of $\Omega(\epsilon
n^2/(m^2\log^2 n))$ many $(1,2,3)$-appearances from $M^{(\ell^*)}$. Thus, in $\widetilde{\Theta}(1/\epsilon)$ iterations, with probability $1 - n^{-\Omega(1)}$, there exists a sampled box $B$ that contains the $2$-legs of $\Omega(\epsilon
n^2/(m^2\log^2 n))$ many $(1,2,3)$-appearances from $M^{(\ell^*)}$. In the following, we condition on such a box $B$ being sampled.

For this box $B$, let $R'$, as before, denote the row containing $B$ when restricted to the boxes to the right of $B$. We can see that $R'$ is $\Omega(\epsilon/(m\log^2 n))$-far from $(1,2)$-freeness, by virtue of the $2$- and $3$-legs of the $(1,2,3)$-appearances. Thus, running the $(1,2)$-freeness tester with proximity parameter $\Omega(\epsilon/(m\log^2 n))$ for $\widetilde{\Theta}(m/\epsilon)$ iterations, finds $\widetilde{\Theta}(m/\epsilon)$ many such $(2,3)$-appearances, with probability $1 - n^{-\Omega(\log n)}$. 

For each of the $\widetilde{\Theta}(m/\epsilon)$ many $(2,3)$-appearances
found above, the $2$-leg is a uniformly random point among the $2$-legs of $B$, by~\cref{rem:needed-main}.
Now, let $C$ be the column containing box $B$ and consider a partition of $C$ into subcolumns where each subcolumn, in this case, consists of 
points with the same $x$-coordinate.
These $n/m$ columns cover the set of all $2$-legs in $B$.
Let $(p_2, p_3)$ be one such $(2,3)$-appearance returned above. 
Since $p_2$ is a uniformly random point among the $2$-legs of $B$, there are, in expectation, $\Omega(\epsilon
n/(m\log^2 n))$ many $2$-legs in the subcolumn of $B$ containing $p_2$.
From this, one can conclude that, with probability $\Omega(\epsilon
/(m\log^2 n))$, (1) the point $p_2$ belongs to a subcolumn that has $\Omega(\epsilon
n/(m\log^2 n))$ many $2$-legs and (2) that $f(p_2)$ is at least the median value of $2$-legs in its subcolumn. 
Since we have $\widetilde{\Theta}(m/\epsilon)$ many such $(2,3)$-appearances, with probability at least $1 - n^{-\Omega(\log n)}$, one of them must satisfy the above conditions. Let $(p_2', p_3')$ denote that $(2,3)$-appearance. Now, it must be the case that
there are $\Omega(\epsilon n/(m\log^2 n))$ many $1$-legs in the subcolumn of $p_2'$ that can, together with $p_2'$, result in a $(1,2,3)$-appearance.
Since we sample $\widetilde{\Theta}(m/\epsilon)$ many points from the subcolumn of each such $2$-leg sampled, we succeed in finding a $(1,2,3)$-appearance with probability at least $1 - n^{-\Omega(\log n)}$. 

The argument for $\ell^* \geq 1$ is very similar to the one and is omitted for brevity.  The query complexity of a procedure resulting from the above idea is $\widetilde{O}(m^2 \poly(1/\epsilon))$, and its pseudocode can be found in the appendix (\cref{alg:M4-large-123}).

\begin{remark}\label{rem:750}
	In the description above, only the step that uses sampling to detect the $1$-leg is adaptive, as it is
	performed w.r.t.\ the $(2,3)$-appearances that are
	found. However, since such an appearance is guaranteed to be found
	with high probability, this adaptivity is not algorithmically
	needed. Hence the whole algorithm can be made nonadaptive with
	essentially the same complexity.
\end{remark}

\subsection{A recursive algorithm}\label{sec:recursive}

Our goal here is to improve the $\widetilde{O}(n)$ tester to a
$\polylog n$ query algorithm and prove \cref{thm:123_polylog}. 

Let $\epsilon^{(i)} = \epsilon \cdot (1 - \frac{1}{\log n})^i$. Let $m = \log^2 n$.
The idea is to use recursion in the case that $\epsilon^{(1)}$ fraction
of $(1,2,3)$-appearances each have all their legs in one box in $G_m^{(2)}$. This corresponds to the case that $M_0$ has large size in the description in the preceding sections. 
Consider a collection of $t = t^{(1)} = \Theta(\frac{\log^3 n}{\epsilon^{(1)}})$ boxes sampled uniformly and independently at random from $G_m^{(2)}$.
Let $X$ denote the size of the matching of $(1,2,3)$-appearances restricted to these boxes.
We know that $\E[X] = t \cdot \frac{\epsilon^{(1)} n^2}{m^2}$.
By the Hoeffding's inequality, we know that $X \geq t \cdot \frac{\epsilon^{(2)} n^2}{m^2}$, with probability
at least $1 - n^{-\Omega(1)}$. Let $B_1, \ldots, B_t$ denote the boxes sampled above and let $S$ denote the union of these boxes in $G_m^{(2)}$.  From the above, we can conclude that the region $S$ is $\epsilon^{(2)}$-far from
being $(1,2,3)$-free. We recursively call our algorithm on $S$. The important point is that
we only need to consider $(1,2,3)$-appearances that are fully contained within the boxes $B_1, \ldots, B_t$ and we can ignore any
appearance that has its legs belonging to different boxes. In other words, we are, in a sense, solving $t$ disjoint subproblems in parallel.

First, we consider a gridding of each of these boxes into $m \times m$ subboxes.
Suppose that at least $\epsilon^{(2)}/(3\log n)$ fraction of all the $(1,2,3)$-appearances in $S$ are such that
their legs all belong to subboxes in the same row but not all in the same subbox (as in the case when $M_1$ is large). 
It must be the case that there exists one box $B_i$ such that at least $\epsilon^{(2)}/(3t\log n)$ fraction of all the 
$(1,2,3)$-appearances have their legs belonging to rows of subboxes of $B_i$. Thus, running the procedure
for the case that $M_1$ is large on each box $B_i, i \in [t]$ separately with parameter $\epsilon^{(2)}/(3t\log n)$
will find such a $(1,2,3)$-appearance with an overall success probability of $1 - n^{-\Omega(1)}$. The number of queries is $\widetilde{O}(m^2 \cdot \poly(1/\epsilon^{(2)}))$. 
The other non-recursive cases can be dealt with in this same fashion. 

For the recursive case, we sample $t^{(2)} = \Theta(\frac{\log^3 n}{\epsilon^{(2)}})$ subboxes uniformly and independently at random from the region $S$ and as above, a Hoeffding's bound will guarantee that the union of these subboxes is $\epsilon^{(3)}$-far from
being $(1,2,3)$-free.

The base level of the recursion is when the region has size $\Theta(\log^4 n \cdot \poly(1/\epsilon))$, in which case, we query the entire region. 
Note that at the $i$-th recursion level, for $i \geq 1$, the size of the domain becomes $(n^2/m^{2i}) \cdot t^{(i)} = (n^2/m^{2i}) \cdot \Theta(\frac{\log^3 n}{\epsilon^{(i)}})$. Setting $m = \log^2 n$ bounds the recursion depth to at most $\log n$. The proximity parameter at the base level is $\epsilon/c$ for some absolute constant $c$.
The total query complexity in the non-recursive cases at the $i$-th recursion level is $\widetilde{O}(m^2 \cdot \poly(1/\epsilon^{(i)}))$ for all $i \geq 1$. Hence, the overall query complexity of the recursive algorithm is $\polylog n \cdot \poly(1/\epsilon)$. 

\begin{remark}
	As in \cref{rem:750}, this algorithm can also be made nonadaptive.
\end{remark}

This completes the proof of \cref{thm:123_polylog}, which we restate here for convenience.
\begin{theorem}
	There exists a one-sided error nonadaptive $\epsilon$-tester with query complexity $\log^{O(1)} n$ for $(1,2,3)$-freeness of functions $f \colon [n]^2 \to
	\R$ for constant $\epsilon \in (0,1)$. $\qed$
\end{theorem}

\section{Testing \texorpdfstring{$(1,3,2)$}{(1,3,2)}-freeness}\label{sec:132}

In this section we develop a sublinear-time tester for $(1,3,2)$-freeness and prove \cref{thm:3-pat-general}, where $\pi = (1,3,2)$ is
the only forbidden order in $\mathcal{S}_3$, other than $(1,2,3)$, up to
order-isomorphism. This case turns out to be harder than testing
$(1,2,3)$-freeness, much like the difference in the $1$-dimensional
case between the respective patterns. In particular, the sublinear-query
tester we develop is {\em adaptive}, as expected in view of the lower
bound on nonadaptive testers even for the $1$-dimensional case. Here, we need a different set of tools,
following~\cite{NV22}.

\subsection{Preliminaries}\label{sec:prelim-132-testing}

Let $f \colon [n]^2 \to \R$.  The set of values of $f$ along with the $2$-dimensional
domain naturally defines a $3$-dimensional box. This grid, the
\enquote{graph-grid} is isomorphic to $[n]^2 \times R(f)$ containing $n^2$
$f$-points of the form $(p,f(p)), ~ p \in [n]^2$.  For a parameter
$m \leq n$, let $G_m^{(2)}$ be
the $2$-dimensional $m \times m$ grid of boxes as defined in previous
sections.  One can then define a coarse $3$-dimensional
$m \times m \times m$ grid $G_m^{(3)}$ that partitions $[n]^2 \times R(f)$. In particular, $G_m^{(3)}$ partitions the set of
$f$-points, $\{(p,f(p)): ~ p \in [n]^2\}$ into $m^3$ parts. 
Let $I \subseteq [n]$ be an interval belonging to the partition of $x$-coordinates
in $G_m^{(3)}$ above. The region of the grid-graph with $x$-coordinates in the range $I$ 
is then referred to as an \emph{$x$-slice}. Each $x$-slice contains $m^2$ boxes. One can define $y$- and $z$-slices similarly. $z$-slices are also
referred to as \emph{layers}. The intersection of an $x$-slice and a $y$-slice is an \emph{$xy$-column} and contains $m$ boxes. One can define
$yz$ and $xz$-columns similarly.
We say
that a box $B \in G_m^{(3)}$ is nonempty if it contains a point
$(p,f(p))$ and empty otherwise.

We use \cite[Theorem 2.2]{JNO21}, restated for our
case, that generalizes the Marcus-Tardos Theorem
\cite{MarcusT04} as follows.

\begin{theorem}\cite{JNO21}\label{thm:MT} 
	For any integers $k,m, t$ such that $k \leq m$, there is a $\tau_t(k)$ such that if a function $g \colon [m]^t \to
	\{0,1\}$ evaluates to $1$ on more than $\tau_t(k) \cdot m^{t-1}$ entries, then, for any
	$k$-sized subposet $\pi_k$ of the $t$-dimensional grid $[k]^t$, there is a collection of
	$k$ $1$-entries in $[m]^t$ that form an isomorphic copy of $\pi_k$.   
\end{theorem}\label{thm:30}

We will use  \cref{thm:MT} with $k=3$ (corresponding to the forbidden $3$-pattern)
and $t=3$ corresponding to the $3$-dimensional Boolean grid of boxes that
corresponds to the graph of the function $f \colon [n]^2
\to \R$, in a similar way that $t=2$ was used in
\cite{NV22} for functions on the line. Namely, we view $G_m^{(3)}$
with its $f$-points as a {\em Boolean} grid  isomorphic to $[m]^3$
where a point (box) is $0$ if it is empty and $1$ otherwise.

The basic tool that allows this test is that the Hamming distance and
the deletion distance are equal for the property of
$(1,3,2)$-freeness by~\cref{cor:hamdel3}. Hence, in what follows,  we may
assume that $f$ is $\epsilon$-far from $(1,3,2)$-free, and contains a matching $M$ of $(1,3,2)$-appearances
with $|M| = \Omega(\epsilon  n^2)$, since we are going to develop a one-sided error tester for this property. We note that any
$(1,3,2)$-appearances in $f$ corresponds to a $(1,3,2)$-appearance in the
graph-grid (but not necessarily the other way, as in previous sections). 

\subsection{A \texorpdfstring{$\widetilde{O}(n)$}{O~(n)} tester for \texorpdfstring{$(1,3,2)$}{(1,3,2)}-freeness}\label{sec:O(n)-132}
Recall that a trivial deterministic tester for $(1,3,2)$-freeness takes
$O(n^2)$ queries. A baseline nonadaptive sublinear-query tester following \cref{lem:31} makes $O(n^{4/3})$ queries. Our aim in this section is to improve this by designing an $\widetilde{O}(n)$-query tester.

\begin{theorem}\label{thm:3-pat-n-tester}
	There exists a one-sided error
	$\epsilon$-tester with query complexity $\widetilde{O}(n)$ for $(1,3,2)$-freeness of functions $f \colon [n]^2 \to
	\R$ for constant $\epsilon \in (0,1)$.
\end{theorem}

\subsubsection{High-level outline}\label{sec:top-level}
Let $m \colonequals \sqrt{n}$. We consider the $2$-dimensional grid of boxes $G_m^{(2)}$ that partitions the
domain $[n]^2$. 
As in \cref{sec:123}, we can partition the matching $M$ into a
constant number of submatchings according to the configurations that the
$(1,3,2)$-appearances form. Moreover,  the cases in
which all the legs share a row, or share a column, are treated
identically as in \cref{sec:123}, and will
not be discussed here. 

The same is true for  the cases of the partial matchings in which the $1$-leg in
every appearance does not share a column or row with the $2$- or $3$-legs (analogous to the cases where $M_2$ is large in \cref{sec:123}).

The harder cases are when the $1$-leg and the $3$-leg share a column or a row (see \cref{Fig_new1}),
while the $2$-leg could be in a shared row (column) or in a third row (column). The harder cases among
these is when the $2$-legs share a row or a column with the
$3$-leg. There are two such configurations -- the $\Gamma$-shape (\cref{Fig_new1}(a)), and an inverted $\Gamma$-shape (\cref{Fig_new1}(b)). The treatment of these two cases is similar, so we only describe the case in which there is a large
matching $M_{\Gamma}$ of size $\Omega(\epsilon n^2)$ of $\Gamma$-shape $(1,3,2)$-appearances (i.e., the case shown in \cref{Fig_new1}(a)). The cases in which the $2$-leg does not share a column or row with the $3$-leg (see \cref{Fig_new1}(c) and \cref{Fig_new1}(d)) are similar.

\begin{figure}[ht]
	\centering
	\includegraphics[scale=0.6]{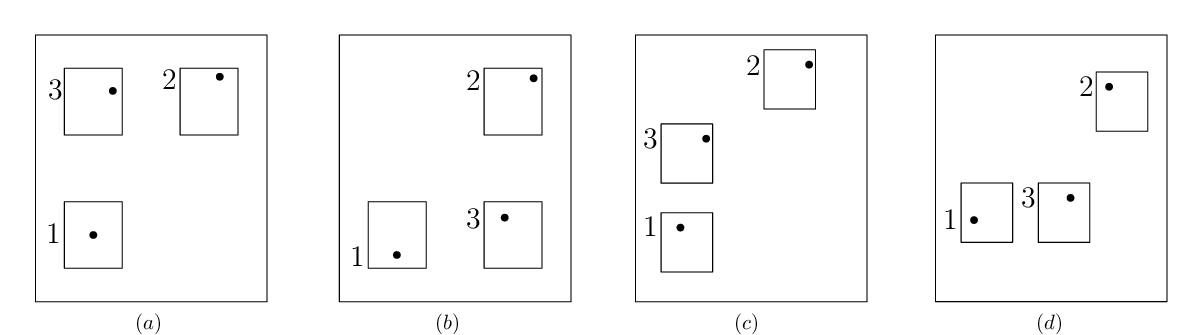}
	\caption{(a), (b) The harder cases -- Boxes in $G_m^{(2)}$ containing the legs of a $(1,3,2)$-appearance forming a (a) $\Gamma$-shape, and (b) an inverted $\Gamma$-shape.
		(c), (d) The easier cases -- Box in $G_m^{(2)}$ containing the $2$-leg is in a row and column different from that of the box containing the $3$-leg.}
	\label{Fig_new1}
\end{figure}

\vspace{0.5cm}
\noindent
{\bf The high level description:} Let $m \colonequals \sqrt{n}$.
\begin{enumerate}
	
	\item{\bf Gridding:} We first perform $\mathsf{Gridding}(S,I,m,\beta)$ (detailed in \cref{sec:subsec_gridding}) with $S \colonequals {[n]}^2$, $I \colonequals (-\infty,\infty)$, $m \colonequals \sqrt{n}$, and $\beta = \Theta(1)$. The algorithm $\mathsf{Gridding}(S,I,m,\beta)$ uses $\mathsf{Layering}(S, I, m)$ (detailed in \cref{sec:layering}) as a subroutine, which is used to obtain a \enquote{nice} partition of $R(f)$ into $m' \leq 2m = \Theta(\sqrt{n})$ layers (see \cref{def:nice-partition}). $\mathsf{Layering}$ is carried out by making $O(m' \cdot {\log}^4 {n}) = \widetilde{O}({m})$ queries. Via $\mathsf{Gridding}(S,I,m,\beta)$, we obtain the $3$-dimensional grid
	$G_m^{(3)}$ that partitions the graph-grid ${[n]}^2 \times R(f)$ into a set  $\mathcal{B}_3$ of $\Theta(m^3)$
	boxes with \enquote{well-controlled} densities (see \cref{cl:q_layering}). Each subbox $B \in \mathcal{B}_3$ is tagged by $\mathsf{Gridding}(S,I,m,\beta)$ as
	\emph{empty} if it contains no sampled points, 
	and \emph{nonempty} otherwise. Further, a nonempty box $B$ is tagged as \emph{dense} by $\mathsf{Gridding}$ if it estimates that $B$ contains $\Theta(n^2/m^2) = \Theta(n)$
	$f$-points. Note that the projection of any $B$ on the domain
	$[n]^2$ is of size $\Theta(n^2/m^2)$. Hence, \emph{dense} means that a constant
	fraction of the points in the projection of $B$ on the domain are
	indeed inside $B$. Otherwise, $B$ is \emph{not dense}. We prove that our
	estimates are accurate enough with high probability, and hence assume that the
	\emph{dense} boxes are indeed dense.
	
	We note that although there are $\Theta(m^3) = \Theta(n^{3/2})$ boxes in $G_m^{(3)}$,
	the process will only use \\$O\left(m \cdot {\log}^4 {n} + \frac{{\log}^4 {n}}{{\beta}^2} \cdot 4m^2\right) = \widetilde{O}(m^2) = \widetilde{O}(n)$ queries, and this process works as intended with high probability (at least $1 - 1/n^{\Omega(\log n)}$).
	
	\item As a result of $\mathsf{Gridding}$, we either find $\omega(m^2)$ nonempty
	boxes, or find out $O(m^2)$ dense boxes that cover all but a
	negligible fraction of the $f$-points.  In the first case, \cref{thm:MT} will guarantee the existence of a $(1,3,2)$-appearance
	in the already sampled points, in which case we stop and reject. Otherwise, we move on to the next item.
	
	\item\label{item:main} If we reach here, we have $O(m^2) = O(n)$ dense boxes in
	$\mathcal{B}_3$. We note that every slice ($x$-, $y$-, or $z$-slice) contains $O(m)$
	dense boxes, since, by averaging, we ignore other cases. By the same
	reasoning, every $xy$-column contains $\Theta(1)$ dense boxes (but
	this is not necessarily true for $xz$-columns or $yz$-columns).
	
	We ignore all $f$-points outside these
	boxes 
	and will show how to find a $(1,3,2)$-appearance with legs in the
	dense boxes.  Since we ignore an insignificant fraction of the $f$-points,
	there is still a matching $M'$ of size $\Theta(\epsilon n^2)$ of
	$\Gamma$-shaped $(1,3,2)$-appearances among the remaining points. 
	We now partition $M'$ into a constant number of
	submatchings, according to the configuration the appearances form
	in $G_m^{(3)}$. Each case will be treated separately, as
	explained in the next section, in $\widetilde{O}(n)$ queries.  
\end{enumerate}

\subsubsection{Finding a \texorpdfstring{$(1,3,2)$}{(1,3,2)}-appearances in a large \texorpdfstring{$\Gamma$}{Gamma}-shaped matching}
We assume in what follows that the matching $M$ of $\Gamma$-shaped
$(1,3,2)$-appearances is of size $\Omega(\epsilon n^2)$. A similar treatment can be carried out
when the matching of inverted-$\Gamma$-shaped appearances is large (by
possibly switching the $x$-coordinates and the $y$-coordinates in the
discussion), and for shapes as in \cref{Fig_new1}(c) and \cref{Fig_new1}(d).

\begin{enumerate}
	\item {\bf $M$ contains $\Omega(\epsilon n^2)$ appearances, where the
		legs of each appearance share a layer (that is,  a $z$-slice)}  
	
	In this case, a uniformly random
	layer will have $\Omega(\epsilon n^2/m)$ appearances in expectation. We choose a random
	layer $S_i, ~ i \in [m]$. Then, with probability $\Omega(\epsilon)$, it will
	contain $\Omega( \epsilon n^2/m)$ appearances. Thus, sampling
	$\widetilde{\Theta}(\frac{n^{4/3}}{m^{4/3}\epsilon^{1/3}}) =
	\widetilde{O}(n/\epsilon^{1/3})$  points uniformly and independently at random, will find
	an appearance with probability $1 - n^{-\Omega(\log n)}$, by \cref{lem:31}.
	
	\item {\bf $M$ contains $\Omega(\epsilon n^2)$ appearances, each
		of which has its legs on two layers}  
	\begin{itemize}
		\item Most appearances have their $1$-leg in a layer, and the
		$2,3$-legs in a different layer. See \cref{Fig_new2}.

		\begin{figure}[ht]
			\centering
			\includegraphics[scale=0.4]{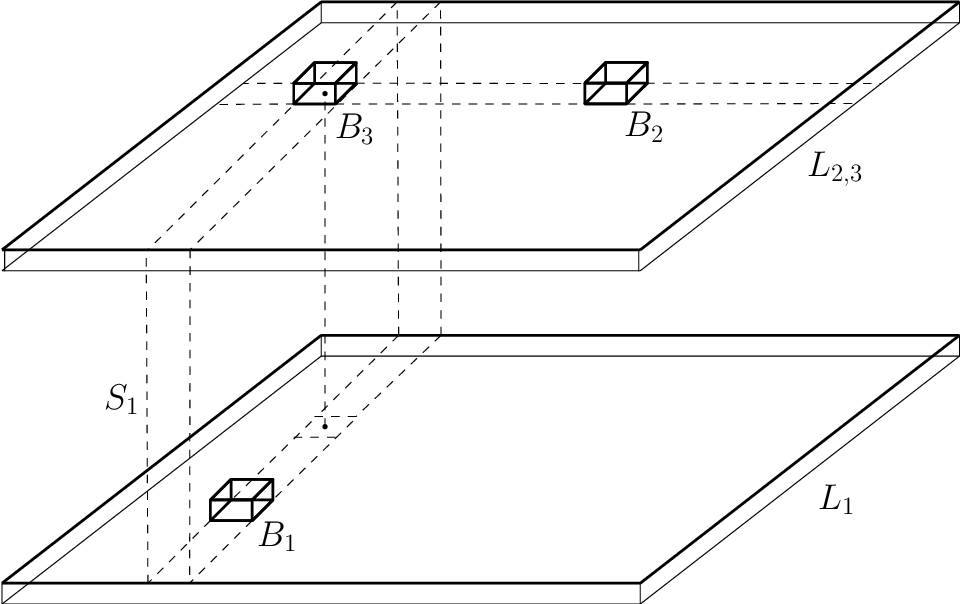}
			\caption{A typical occurrence: $1$-leg on layer $L_1$ and
				$2,3$-legs are on layer $L_{2,3}$}
			\label{Fig_new2}
		\end{figure}
		
		In this case, a typical dense box $B \in G_m^{(3)}$ will contain
		$N=\Omega(\epsilon n^2/m^2)$ 3-legs of $(1,3,2)$-appearances in
		$M$. Let $M(B_3)$ be this set of $N$ appearances. We choose such a box $B_3$, and assume it is \enquote{typical} (an
		event which will occur with probability
		$\Omega(\epsilon)$). Choosing $B_3$ specifies the layer
		$L_{2,3}$ in which the corresponding $2$-legs of the appearances
		in $M(B_3)$ lie, and the $x$-slice ($S_1$ in the
		figure) that should contain the $1$-leg.
		
		We note that the corresponding $x$-slice $S_1$ should contain $N$ many
		$1$-legs that correspond to the $3$-legs in $B_3$. However, we do not know
		the specific $B_1$ and $B_2$ that contain the $1,2$-legs.
		
		We choose a uniformly random point $p_3$ in $B_3$. With probability $\Omega(\epsilon)$, 
		it will be a corresponding $3$-leg, and with probability $1/4$,
		given that it is a $3$-leg, will have $x$-coordinate $x(p_3)$
		that is above but close to the median of the $x$-coordinates of $1$-legs in
		$S_1$ that correspond to the $N$ appearances in $B_3$. 
		Assuming
		that this happens, there
		are $\Omega(\epsilon n^2/m^2)$-appearances whose $3$-leg
		in $B_3$ has $x$-coordinate at least $x(p_3)$. 
		Each such
		$3$-leg, together with the corresponding $2$-leg,  forms a
		$(3,2)$-appearance in $L_{2,3}$, and to find one, we use our
		monotonicity tester with distance parameter $\Omega(\epsilon/m)$, restricted
		to the points with $x$-coordinate larger than $x(p_3)$. Once such
		an appearance $(p_3^*,p_2)$ is found, the slice $S_1$ should
		contain $N/2$ points $p_1$ that have $x$-coordinates lower than
		$x(p_3)$ and are $1$-legs of appearances in $M$. Namely, a
		random point satisfies this with probability
		$\Omega(\epsilon/m)$. Then, sampling
		$\tilde{\Theta}(m/\epsilon)$  such points will find a point $p_1$ that 
		together with $p_3^*,p_2$ forms a $(1,3,2)$-appearance.
		
		The resulting query complexity for this case is $\widetilde{O}(m/\epsilon)$, which is the query complexity of running the monotonicity tester (see \cref{thm:improved-mon-partial-main}) with parameter $\Omega(\epsilon/m)$.
				
		\item  Most appearances have their $1,2$-legs in a layer, and the
		$3$-legs in a different layer (see~\cref{Fig_new3}).
		We
		note that this is the only other possibility for a
		$\Gamma$-shape to have its two legs over two layers. 
		
		\begin{figure}[ht]
			\centering
			\includegraphics[scale=0.4]{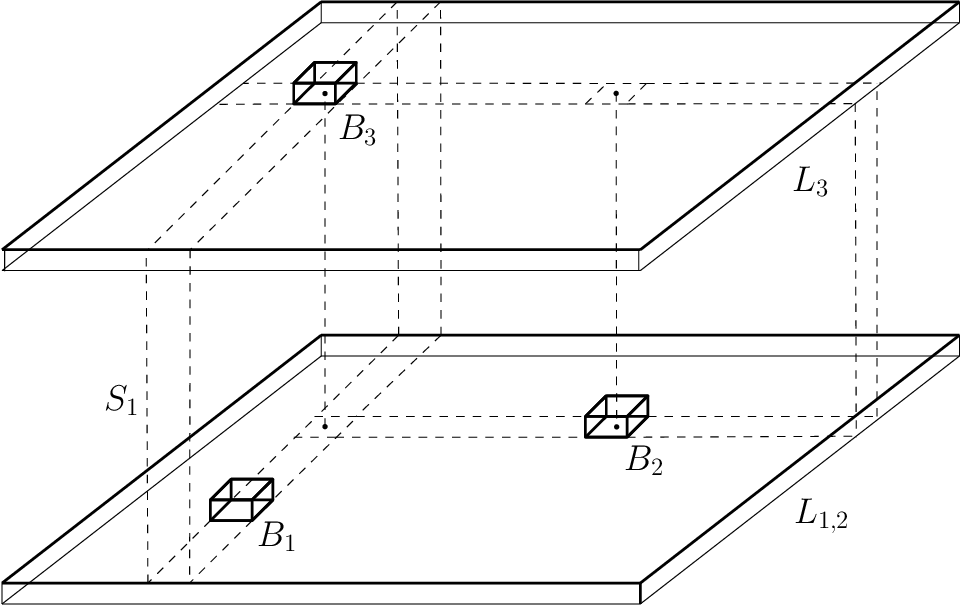}
			\caption{A typical occurrence: $1,2$-legs on layer $L_{1,2}$ and
				$3$-legs are on layer $L_{3}$}
			\label{Fig_new3}
		\end{figure}
		
		In this case, we present two solutions. The first, described in \textbf{A1}
		below, is conceptually easier and good enough for the
		$\widetilde{O}(n)$-algorithm. Its drawback is that its query complexity
		depends inversely on $m$ and hence requires a relatively large value of
		$m$. The second is proportional to $m$ and is therefore better for small values of
		$m$ which will allow for better recursion. 
		
		In both cases we choose a random layer $L_{1,2}$, which, with
		probability $\Omega(\epsilon)$, will contain $N=\Omega(\epsilon n^2/m)$
		$1,2$-legs of the appropriate $N$ appearances  $M(L_{1,2})$ in
		$M$.
		
		{\bf A1:} Next, we choose random $B_1,B_2$ that should contain the
		corresponding $1,2$-legs of such appearances,
		respectively. As there are $O(m^2)$ pairs of boxes, in expectation, a uniformly random pair of boxes will
		contain $\Omega(\epsilon n^2/m^3)$ 
		such appearances whose corresponding pairs of $1,2$ legs are in
		$B_1,B_2$ respectively. Now, fixing $B_1,B_2$, the $xy$-column containing
		$B_3$ is fixed, and by our assumption, contains $\Theta(1)$ dense
		boxes. Now, there are two ways to find a $(1,3,2)$-appearance in the
		corresponding boxes.
		We focus on the constant number of possible $B_3$, which, together with
		the fixed $B_1,B_2$ have $\Theta(n^2/m^2)$ points in total and contains
		$\Omega(\epsilon n^2/m^3)$ many $\pi$-appearances. Hence, by \cref{lem:31}, it is enough to sample $O(\frac{n^2}{m^2}
		\cdot (\frac{m^3}{\epsilon n^2})^{1/3}) =
		O(\frac{n^{4/3}}{\epsilon^{1/3}m})$ points uniformly and independently at random from this union of boxes
		in order to obtain a $\pi$-appearance with high probability. The sample complexity is $O(n^{5/6})$ for constant $\epsilon$ by our setting of $m$. 
				
		{\bf A2:}  We start as described above, except that the order of the
		sampling is different. After choosing $L_{1,2}$ we first choose $B_1$
		at random in $L_{1,2}$, which will contain, in expectation, $\Omega(\epsilon n^2/m^2)$ many $1$-legs from
		the matching $M$. Now, before choosing $B_2$, we first sample
		$\Theta(1/\epsilon)$ points from $B_1$.  We expect that at least
		one of these points is from the set of
		$1$-legs above, and further, that this point has value close to the median
		value of the $1$-legs.
		
		Assuming we have chosen such a $p_1$ successfully, let $M_H$ contain all
		appearances whose $1,2$-legs are in $L_{1,2}$ and the $1$-leg value is above $f(p_1)$. Since we assume
		that $f(p_1)$ is close to the median, we conclude that $|M_H| = \Omega
		(\epsilon n^2/m^2)$.  We now choose $B_2$, and correspondingly $B_3$
		(the latter from constantly many possibilities in its $xy$-column).  Since $B_2$ is chosen
		at random, we expect that there is a matching $M_H^1 \subseteq M_H$ of
		size $\Omega(\epsilon n^2/m^3)$ whose $3,2$-legs are in $B_3,B_2$,
		respectively.
		
		We now apply our partial function monotonicity tester (\cref{thm:improved-mon-partial-main}) in $B_2,B_3$ but restricted to
		points of $f$-values above $f(p_1)$, with distance parameter $\epsilon/m$, to
		find such a $(3,2)$-appearance $(p_3,p_2)$. This can be done in $\widetilde{O}(m/\epsilon)$ queries,
		and moreover, according to \cref{rem:needed-main}, the point $p_3$ is
		essentially a random point in $M_H^1$. In particular, its
		$x$-coordinate is larger than the $x$-coordinates of at least half the
		points in $M_H^1$. In turn, this implies that a random point in $B_1$
		is likely to be a $1$-leg of an appearance in $M_H^1$, but with
		$x$-coordinate smaller than $x(p_3)$ with probability
		$\Omega(1/m)$. Hence sampling $\widetilde{\Theta}(m)$ points will find such a
		point $p_1^*$ with high probability. Thus, $(p_1^*,p_3,p_2)$ form a
		$(1,3,2)$-appearance.
		
		We end this case by noting that the discussion above assumed that we knew such a
		\enquote{successful} $p_1$ to begin with. Since we do not know the correct $p_1$
		out of our sample, we try for all possible sampled points. Hence, the
		total query complexity in this case is $\widetilde{O}(m/\epsilon)$.
	\end{itemize}
	
	\item {\bf $M$ contains $\Omega(\epsilon n^2)$ appearances, each
		with its legs on three layers.}  
	This case, for a $\Gamma$-shaped appearance, is depicted in~\cref{Fig_new4}. 
	\begin{figure}[ht]
		\centering
		\includegraphics[scale=0.4]{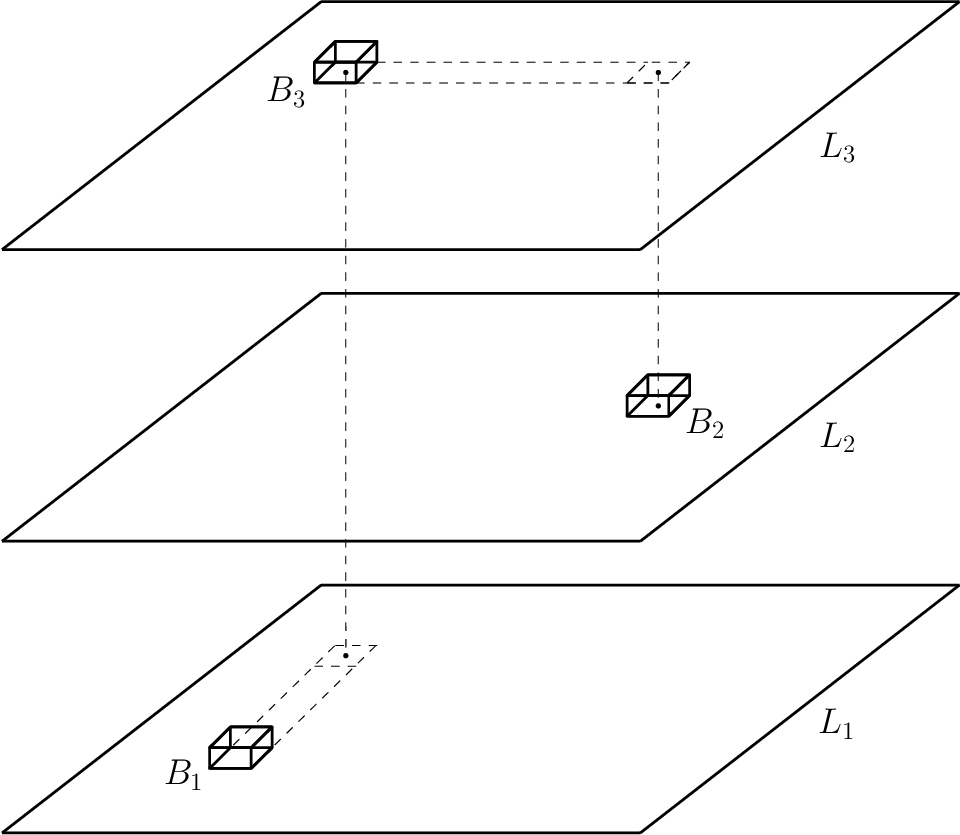}
		\caption{A typical occurrence: $1,2,3$ legs on layer $L_1,L_2,L_3$ respectively}
		\label{Fig_new4}
	\end{figure}

	In this case, a typical $x$-slice $S_1$ will contain $N=\Omega(\epsilon
	n^2/m)$ many $1$- and $3$-leg pairs from appearances in $M$.
	With constant probability, a random dense box $B_1$ in $S_1$ and a random index $p_1 \in B_1$ will have $x(p_1)$ below the median of the $x$-coordinates of the $1$-legs of
	these $N$ appearances.
	Assuming this happens, any
	$(3,2)$-appearance above the layer $L_1$ and with $x$-value larger
	than $x(p_1)$ will result in a $(1,3,2)$-appearance with
	$p_1$. Thus, using the monotonicity tester, this can be found in
	$\widetilde{O}(m/\epsilon)$ queries as the corresponding distance parameter is
	$\epsilon/m$. This is also the total query complexity for
	this case.
\end{enumerate}

Since all of these cases, as well as the \textsf{Layering} and \textsf{Gridding} are done using
$\widetilde{O}(n)$ queries, the claim on the query complexity of the test
follows. This proves~\cref{thm:3-pat-n-tester}. We remark that the queries made in Cases 2 and 3 above are adaptive, as they depend on which boxes are marked as dense by the \textsf{Gridding} phase.

\subsection{Improving the query complexity \texorpdfstring{to $\widetilde{O}\left(n^{\frac{4}{5} + o(1)}\right)$}{}}
In this section, we discuss the ideas to obtain our improved $(1,3,2)$-freeness tester and prove \cref{thm:3-pat-general}.

\begin{proof}[Proof of~\cref{thm:3-pat-general}]
From our discussion in the preceding section, the query complexity of all cases except Case 1 (where all legs are in a single layer), 
is upper bounded by $\widetilde{O}(m^2 \poly(1/\eps))$. The query complexity of Case 1 has an inverse dependence on $m$, 
and this is the case where recursion will help improve the complexity.

We focus on Case 1. Namely, we assume Case 1 happens at the top level of the recursion, and that we have
sampled a layer $L$ containing $\Omega(\epsilon n^2/m)$ appearances. 
Using the fact that the number of dense boxes in every layer is $O(m)$, we restrict our attention to the function restricted to the
part of the domain defined by the $O(m)$-many $xy$-columns defined by
these dense boxes in the current layer we focus on. 
The challenge in the recursive call is in performing \textsf{Layering} and \textsf{Gridding}, since we do not have a rectangular domain anymore. 
We instead do the following. We first perform \textsf{Layering} on this restricted function and form $O(m)$ layers. 
Now, we perform \textsf{Gridding} on the domains of each of the $O(m)$-many $\frac{n}{m} \times \frac{n}{m}$ sized rectangular subgrids separately with the same gridding parameter $m$. 
This results in $O(m^3)$ dense boxes. 
The overall query complexity for these steps is  $\widetilde{O}(m^3
\poly(1/\eps))$. In addition, within the same query complexity, all
cases except Case 1, can be done with respect to the function restricted
to the layer $L$.

Recursing to a depth of $d$ in the above fashion, 
the query complexity for \textsf{Gridding} and all cases except Case 1 at the lowest level of recursion is $\widetilde{O}(m^{d+2} \poly(1/\eps_d))$, where $\eps_d$
is the distance parameter at that level. Since the distance parameter drops by a constant factor in each successive level of recursion, we have $\eps_d = \eps/c^d$
for some constant $c$. 

At the $d^{\text{th}}$ recursion level, if we are in Case 1, we restrict our attention to $O(m^{d+1})$ dense boxes all in the same layer, where each dense box has asymptotically $\Theta(n^2/m^{2d+2})$ many points. The $(1,3,2)$-appearances that we look for are present in an $\eps_d$-fraction of these points. Thus, by \cref{lem:31}, making 
$\widetilde{O}\left(\frac{n^{4/3}}{m^{\frac{2d+2}{3}}}\cdot
\frac{1}{\eps_d^{1/3}}\right)$ queries will enable us to find these
appearances with high probability. Setting $m$ to be $n^{4/(5d+8)}$,
and plugging in the expression above results in
$\widetilde{O}\left(m^{\frac{4d+8}{5d+8}}\cdot \poly(1/\eps_d)\right)$ complexity. 
Finally, the total number of nodes in the recursion tree is at most $(\widetilde{\Theta}(1/\eps_d))^d$ and the number of queries in each node is upper bounded by the work done at the leaf nodes. 
Therefore, the overall query complexity for a $d$-level recursive procedure is $O(m^{\frac{4d+8}{5d+8}}\cdot ((\polylog n)/\eps_d)^d)$. By setting $d$ to be $\Theta(\log \log n)$, one can see that the second factor in the product is $n^{o(1)}$ and the first factor is an expression that tends to $n^{4/5}$. Thus, the overall complexity is $O(n^{\frac{4}{5} + o(1)})$.
\end{proof}

\section{Lower bounds for \texorpdfstring{$(1,3,2)$}{(1,3,2)}-freeness testers on hypergrids} \label{sec:132_lbs-main}

In this section, we prove the following lower bounds for nonadaptive and adaptive $\epsilon$-testers for $(1,3,2)$-freeness on the hypergrid $[n]^2$.

\begin{theorem} \label{thm:132_na_lb-main}
	Any one-sided error nonadaptive $\epsilon$-tester for $(1,3,2)$-freeness on $[n]^2$ has query complexity $\Omega(n)$, for every $\epsilon \leq 1/81$.
\end{theorem}

The existence of a $o(n)$-query \emph{adaptive} tester for $(1,3,2)$-freeness (\cref{thm:3-pat-general}) shows that adaptivity does indeed lead to more efficient testers. 
On the other hand, the following theorem states that even with adaptivity, testing $(1,3,2)$-freeness cannot be done much more efficiently.

\begin{theorem} \label{thm:132_a_lb-main}
	Any one-sided error adaptive $\epsilon$-tester for $(1,3,2)$-freeness on $[n]^2$ has query complexity $\Omega(\sqrt{n})$, for every $\epsilon \leq 4/81$.
\end{theorem}

To prove \cref{thm:132_na_lb-main} and \cref{thm:132_a_lb-main}, the
high level idea is to follow ~\cite{NewmanRRS19}, in which they prove an
$\Omega(\sqrt{n})$-query lower bound for {\em nonadaptively} 
testing $(1,3,2)$-freeness on the line $[n]$. Namely, we reduce an \enquote{intersection-search} problem (\cref{problem:intersection_search-main}) for two \emph{$2$-dimensional}
$3n \times 3n$ `monotone arrays' (w.r.t.\ $\prec$), to the problem of testing $(1,3,2)$-freeness on the
$2$-dimensional hypergrid $[9n]^2$. This reduction (\cref{lem:reduction_int-search_to_132-main}) is blackbox deterministic, and as such, any lower bound for it
(against nonadaptive/adaptive randomized algorithms), implies a corresponding lower bound
for testing $(1,3,2)$-freeness. However, unlike in \cite{NewmanRRS19},
where the corresponding $1$-dimensional problem could be solved
adaptively in $O(\log n)$ queries, we show that our $2$-dimensional variant
requires $\Omega(\sqrt{n})$ adaptive queries. This is shown by
observing that \cref{problem:intersection_search-main} `embeds'
in it $\Omega(n)$ independent instances of a variant of the \enquote{birthday problem}
(\cref{problem:birthday-search-main}), each of which are to be solved on unstructured sets of size $\Omega(n)$.

\begin{problem}[Intersection-search for two $2$-dimensional arrays] \label{problem:intersection_search-main}
	The input to this problem consists of two $3n \times 3n$ arrays $A$ and $B$, and a constant $\gamma$ with $0 \leq \gamma \leq 1$ (which we call the intersection parameter). $A$ and $B$ each consist of $9n^2$ distinct integers sorted in ascending order, with respect to the grid (partial) order $\prec$. It is promised that $\Abs{\{A[\veci] \colon \veci \in [3n]^2\} \cap \{B[\veci] \colon \veci \in [3n]^2\}} \geq \gamma \cdot (9n^2)$. The goal is to find $\veci,\vecj \in [3n]^2$ such that $A[\veci] = B[\vecj]$.
\end{problem}

\begin{problem}[Birthday-search]\label{problem:birthday-search-main}
	The input is a pair of $m$-sized arrays $A,B$, 
	each of which contain the numbers $[m]$, and are permuted according to some \emph{unknown} permutations $\pi_A, \pi_B$ respectively. That is, the
	input is the pair of permutations $\pi_A$ and $\pi_B$. The goal is
	to find a pair of indices $(i,j)$ such that $A[i]=B[j]$. 
\end{problem}

It is quite clear that birthday-search on $m$-size arrays can be
solved (with high probability) nonadaptively by querying $O(\sqrt{m})$ random indices in each
array. It is folklore that adaptivity does not help in this case.

\begin{claim}
	\label{cl:birth-main}
	Any adaptive algorithm for birthday-search on $m$-sized arrays
	requires $\Omega(\sqrt{m})$ queries to achieve a success probability
	of $2/3$. $\qed$
\end{claim}

The intersection-search problem can be reduced to the problem of testing $(1,3,2)$-freeness of real-valued functions on $2$-dimensional hypergrids. We will now formally state this reduction, and defer its proof to the appendix (\cref{lem:reduction_int-search_to_132}).

\begin{lemma} \label{lem:reduction_int-search_to_132-main}
	Let $(A,B,\gamma)$ be an input instance of \cref{problem:intersection_search-main}. Then, there exists a function $f \colon [9n]^2 \to \mathbb{R}$ that is at least $\gamma/9$-far from $(1,3,2)$-freeness, and there is a one-to-one correspondence between valid solutions of \cref{problem:intersection_search-main}, and $(1,3,2)$-appearances in $f$.
\end{lemma}

Observe that \cref{problem:intersection_search-main} can be solved by a simple nonadaptive randomized algorithm that makes a total of $O(n)$ queries to the arrays. One can simply query $O(n)$ indices in $A$ and $O(n)$ indices in $B$ uniformly and independently at random. With probability $\Theta(1)$, this will ensure that we sampled $\veci, \vecj \in [3n]^2$, such that $A[\veci] = B[\vecj]$, provided $\gamma = \Theta(1)$. This is nothing but a variant of the birthday problem. We will prove that this bound is tight and then directly conclude \cref{thm:132_na_lb-main}, thanks to the reduction \cref{lem:reduction_int-search_to_132-main}.

\begin{lemma} \label{lem:intersection_search_na_lb-main}
	Let $\mathcal{A}$ be any nonadaptive randomized algorithm for the intersection-search problem \cref{problem:intersection_search-main}, with intersection parameter $\gamma \geq 1/9$. If $\mathcal{A}$ makes $q$ queries and $q < n^2 / 2$, then the success probability of $\mathcal{A}$ is at most $q^2 / (4n^2)$.
\end{lemma}

The proof of \cref{lem:intersection_search_na_lb-main} is deferred to the appendix (\cref{lem:intersection_search_na_lb}).
\cref{lem:intersection_search_na_lb-main} implies that any nonadaptive randomized algorithm must query $\Omega(n)$ indices from $A$ and $B$ so as to solve \cref{problem:intersection_search-main} with probability $\Theta(1)$. The $\Omega(n)$-query lower bound \cref{thm:132_na_lb-main} directly follows from \cref{lem:reduction_int-search_to_132-main} and \cref{lem:intersection_search_na_lb-main}. $\hfill\qed$

%%%%%%%%%%%%%%%%%%%%%%%%%%%%%%%%%%%%%%%%%%%%%%%%%%%%%%%%%%%%%%%%%%

The adaptive lower bound \cref{thm:132_a_lb-main} is obtained by first proving an $\Omega(\sqrt{n})$-query lower bound for \emph{adaptive} randomized algorithms for \cref{problem:intersection_search-main} that succeed with probability $\Theta(1)$. We state this formally below and defer its proof to the appendix (\cref{lem:intersection_search_a_lb}).

\begin{lemma} \label{lem:intersection_search_a_lb-main}
	Let $\mathcal{A}$ be any adaptive randomized algorithm for the intersection-search problem \cref{problem:intersection_search-main}, with intersection parameter $\gamma \geq 4/9$. If $\mathcal{A}$ makes $q$ queries, then the success probability of $\mathcal{A}$ is at most $q^2 / (4n + 4)$.
\end{lemma}

Hence, any adaptive randomized algorithm must query $\Omega(\sqrt{n})$ indices from $A$ and $B$ so as to solve \cref{problem:intersection_search-main} with probability $\Theta(1)$. The $\Omega(\sqrt{n})$-query lower bound \cref{thm:132_a_lb-main} directly follows from \cref{lem:reduction_int-search_to_132-main} and \cref{lem:intersection_search_a_lb-main}. $\hfill\qed$

\bibliographystyle{alphaurl}
\bibliography{references_soda26}

%%%%%%%%%%%%%%%%%%%%%%%%%%%%%%%%%%%%%%%%%%%%%%%%%%%%%%%%%%%%%%
\newpage
\begin{appendices}
	\crefalias{section}{appendix}
	\crefalias{subsection}{appendix}
	\crefalias{subsubsection}{appendix}
	
	\section{Erasure-resilient monotonicity testing over high dimensions}\label{sec:mono-start}
	In this section, we describe an ER monotonicity tester for real-valued functions over hypergrids of arbitrary dimension. 
	Such a tester for $2$-dimensional
	hypergrids  is a key subroutine in our pattern freeness testers in~\cref{sec:123} and~\cref{sec:132}.
	
	\begin{theorem} \label{thm:improved-mon-er}
		There is a one-sided error $\delta$-erasure-resilient $\epsilon$-tester for monotonicity of functions $f \colon [n]^2 \to \R$ making $\widetilde{O}\left(\frac{\log^2(1/\epsilon)}{\epsilon (1-\delta)}\right)$ queries.
	\end{theorem}
	
	\subsection{A simple tester} \label{sec:mono}
	
	We first describe a simple one-sided error erasure-resilient monotonicity tester for real-valued functions over the $2$-dimensional hypergrid ${[n]}^2$. For convenience, we describe a $(1,2)$-freeness tester, which can be easily modified to be a monotonicity (i.e., $(2,1)$-freeness) tester. Specifically, we prove the following theorem and later discuss the additional ideas needed to show \cref{thm:improved-mon-er}.
	
	\begin{theorem} \label{thm:basic-mon}
		There is a one-sided error $\delta$-erasure-resilient $\epsilon$-tester for $(1,2)$-pattern freeness of functions $f \colon [n]^2 \to \R$ with query complexity $\widetilde{O}\left(\frac{1}{{\epsilon}^2 (1 - \delta)}\right)$ for all $\epsilon, \delta \in (0,1)$.
	\end{theorem}
	
	Let $f \colon {[n]}^2 \to \mathbb{R}$. For simplicity, let us initially assume that $f$ has \emph{no} erased points. Since our goal is to design a one-sided error tester, we may also assume that $f$ is $\epsilon$-far from being $(1,2)$-free.
	The tester we describe rejects upon finding a $(1,2)$-appearance. 
	By~\cref{obs:matching}, there exists a matching $M$ of $(1,2)$-appearances of size $\Abs{M} \geq \epsilon n^2 / 2$. Throughout this section, fix such a matching $M$.
	
	Given a tuple $(p,q)$ with $p, q \in [n]^2$ and $p \prec q$, we place $(p,q)$ in the \emph{bucket} $(\alpha,\beta)$, where $\alpha, \beta \in \set{0,1,\ldots,\log_2 n}$ in the following manner -- at level $i$, partition $[n]$ in $2^{i}$ equal parts in the natural way (for instance, at level $2$, the parts are $[1,n/4], (n/4,n/2], (n/2,3n/4]$, and $(3n/4,n]$), and stop at the lowest level $\alpha$ in which the $x$-coordinates $x(p)$ and $x(q)$ of $p$ and $q$ fall into different parts. That is, the natural partition of the array into $2^{\alpha}$ parts is the \emph{coarsest} partition that separates $x(p)$ and $x(q)$. Note that since $p \prec q$, this would ensure that $x(q)$ belongs to the part of $[n]$ immediately to the right of the part which contains $x(p)$, provided $x(p) < x(q)$. If $x(p) = x(q)$, we set $\alpha \colonequals 0$. Similarly, let $\beta$ be the lowest level that separates the $y$-coordinates $y(p)$ and $y(q)$ of $p$ and $q$. Then, we say that the tuple $(p,q)$ belongs to the bucket $(\alpha,\beta)$. Denoting $\ell = \log_2 n + 1$, we see that every such tuple belongs to one of $\ell^2$ buckets.
	
	Then we define
	\[
	M_{\alpha,\beta} \colonequals \set{(p,q) \in M \colon (p,q) \text{ belongs to the bucket } (\alpha,\beta)}.
	\]
	
	The sets $M_{\alpha, \beta}$ partition $M$ into $O({\log}^2 {n})$ submatchings. An immediate conclusion is that for some ${\alpha}^*, {\beta}^* \in \set{0,1,\ldots,\log_2 n}$, $\Abs{M_{{\alpha}^*,{\beta}^*}} \geq \epsilon n^2 / 2 {\ell}^2$, i.e., $\Abs{M_{{\alpha}^*,{\beta}^*}} = \Omega(\epsilon n^2 / {\log}^2 n)$.
	
	Indeed, we do not have prior knowledge of such $M$, $\alpha^*$, and $\beta^*$. Therefore, we first design a tester for a specified pair of values $\alpha$ and $\beta$ within the range $\mathcal{S} \colonequals \set{0,1,\ldots,\log_2 n}$ and a proximity parameter $\gamma \in (0,1)$. Our final tester is obtained by running this tester for all possible values of $\alpha, \beta \in \mathcal{S}$ and for the specific proximity parameter $\Omega(\epsilon / {\log}^2 n)$ (see \cref{thm:basic-mon}). The erasure-resilience of the tester comes easily at a cost of $O(1/(1-\delta))$, where $\delta$ is an upper bound on the fraction of erasures in the domain.
	
	\paragraph{The basic test $T(\alpha, \beta, \gamma)$ given $\alpha,\beta \in \mathcal{S}$ and $\gamma \in (0,1)$.} Let $\alpha, \beta \in \mathcal{S}$. We partition the domain ${[n]}^2$ into \emph{boxes} of side lengths $n/2^{\alpha}$ and $n/2^{\beta}$ along the $x$-axis and the $y$-axis respectively, in the natural way, when $\alpha, \beta \geq 1$.\footnote{If $\alpha = 0$, then the side length of each box on the $x$-axis is set to $0$. The case $\beta = 0$ is taken care of similarly. Note that at least one of $\alpha$ and $\beta$ is nonzero, as we only place tuples $(p,q)$ with $p \neq q$ in our buckets.} Let $\mathcal{B}$ be this set of boxes. For each box $B \in \mathcal{B}$, let $x_{\max}(B)$ and $y_{\max}(B)$ denote the maximum $x$-coordinate and $y$-coordinate values of $B$ respectively. We define the \emph{diagonally adjacent box} $B'$ of $B$ as follows:\footnote{If $\alpha = 0$, we set $x_{\max}(B') \colonequals x_{\max}(B) = x(p)$. The case $\beta = 0$ is taken care of similarly.}
	
	\begin{equation*} \label{defn:adjacent_region}
		B' \colonequals \set{p \in {[n]}^2 \colon x_{\max}(B) < x(p) \leq x_{\max}(B) + n/2^{\alpha} \text{, and } y_{\max}(B) < y(p) \leq y_{\max}(B) + n/2^{\beta}}
	\end{equation*}
	
	For each edge $(p,q) \in M_{\alpha,\beta}$ with $p \in B$, it follows that $q \in B'$. Let $M(B,B')$ denote the set of all $(p,q) \in M_{\alpha,\beta}$ with $p \in B$. Then, $M_{\alpha, \beta} = {\bigcup}_{B \in \mathcal{B}} M(B,B')$. Further, for any pair of distinct boxes $B_1, B_2 \in \mathcal{B}$, the sets $M(B_1,B_1')$ and $M(B_2, B_2')$ are disjoint and so, the sets $M(B,B')$ form a partition of $M_{\alpha, \beta}$.
	
	\begin{enumerate}
		\item[] \noindent The test $T(\alpha,\beta, \gamma)$ is as follows:
		\item Sample $B \in \mathcal{B}$ uniformly at random.
		\item  TEST$(B,B',\gamma)$: Sample $4/\gamma$ points uniformly at random from $B$ and $4/\gamma$ points uniformly at random from $B'$. Then query $f$ at all of these $8/\gamma$ points and \textbf{reject} if a $(1,2)$-appearance is found among these queried points.\label{stp:testBB'}
	\end{enumerate}
	
	We first prove the following claim about the procedure TEST$(B,B',\gamma)$, which is Step~\ref{stp:testBB'} of test $T(\alpha,\beta,\gamma)$.
	
	\begin{claim} \label{claim:median_argument}
		If $M(B,B') \geq \gamma \cdot \Abs{B}$, then $\operatorname{TEST}(B,B',\gamma)$ finds a $(1,2)$-appearance in $f$ with probability $\Theta(1)$. Its query complexity is $O(1/\gamma)$.
	\end{claim}
	
	\begin{proof}
		The query complexity is clear from the description of
		$\operatorname{TEST}(B,B',\gamma)$. Let the $1$-legs of edges in
		$M(B,B')$ be $p_1,\ldots,p_k$ where $k
		\geq \gamma \cdot \Abs{B}$  and $f(p_1) \leq  \cdots \leq  f(p_k)$.
		Let $Q$
		denote the set of corresponding $2$-legs (in $M(B,B')$) to these $1$-legs. Recall
		that the points in $Q$ are all in $B'$.
		Let $m \colonequals (k+1)/2$. Then there are at least
		$m$ indices $i$ such that $f(p_i) \leq f(p_m)$ and at least $m$
		indices $q \in Q$ for which $f(p_m) \leq f(q)$ (those that correspond
		to $1$-legs $p_i, ~ i > m$). Let $Q_H$ contain these $q$'s. Then, any
		pair $\{(p_i,q) \colon ~ i \leq m, ~ q \in Q_H \}$ forms a
		$(1,2)$-appearance. Thus, the probability that our sample of
		$4/\gamma$ points from $B$ intersects $\{p_i, ~ i \leq m\}$ is
		at least $1-(1-\frac{\gamma}{2})^{4/\gamma} \geq 1 - e^{-2}$, and similarly, the same holds for
		the probability that we sample a point in $Q_H$. Altogether, with
		$\Theta(1)$ probability, we succeed in finding a $(1,2)$-appearance.
	\end{proof}
	
	Let ${\alpha}^*, {\beta}^* \in \mathcal{S}$ be such that $\Abs{M_{{\alpha}^*,{\beta}^*}} \geq \epsilon n^2 / 2 {\ell}^2$. We let ${\epsilon}_1 \colonequals \epsilon / {\ell}^2$. We will now show that $T({\alpha}^*, {\beta}^*,\epsilon_1/4)$ finds a $(1,2)$-appearance in $f$ with probability $\Theta({\epsilon}_1)$.
	
	\begin{lemma} \label{lem:one_iteration}
		The algorithm $T({\alpha}^*,{\beta}^*, \epsilon_1/4)$  finds a $(1,2)$-appearance in $f$ with probability $\Theta({\epsilon}_1)$. Its query complexity is $O(1/{\epsilon}_1)$.
	\end{lemma}
	
	\begin{proof}
		The query complexity is clear from the description of
		$T(\alpha,\beta,\gamma)$. Since
		$\Abs{M_{{\alpha}^*,{\beta}^*}} \geq {\epsilon}_1 n^2/2$ and
		$\Abs{M_{{\alpha}^*,{\beta}^*}} = \sum_{B \in \mathcal{B}}
		\Abs{M(B,B')}$, we have
		${\mathbb{E}}_{B \in \mathcal{B}}[\Abs{M(B,B')}] \geq
		{\epsilon}_1 \Abs{B} / 2$.  Hence, by the reverse Markov inequality, the
		tester samples a box $B$ for which
		$\Abs{M(B,B')} \geq {\epsilon}_1 \Abs{B} / 4$ with probability
		at least ${\epsilon}_1/4$. Conditioning on the event that
		$B$ indeed has this many appearances,~\cref{claim:median_argument}
		implies that the success probability in finding a violation is
		$\Theta(1)$. Hence, $T({\alpha}^*,{\beta}^*, \epsilon_1/4)$
		finds a $(1,2)$-appearance in $f$ with probability $\Theta({\epsilon}_1)$.
	\end{proof}
	
	We will now describe an erasure-resilient tester for $(1,2)$-freeness based on \cref{lem:one_iteration}.
	
	\begin{proof}[Proof of \cref{thm:basic-mon}]
		Let $\epsilon_1$ as defined above. Our tester runs
		$T(\alpha, \beta, \epsilon_1/4)$, independently 
		$\Theta\left(\frac{1}{{\epsilon}_1} \log\log n\right)$ times, for each
		$\alpha, \beta \in \mathcal{S}$, to
		have a $\Theta(1)$ success probability. Every time we
		need to query a point and it happens to be erased, we repeat the
		process by choosing another point independently. So in expectation,
		we invest $O(1/(1 - \delta))$ queries in order to
		query a single point, and this reduces the test to the case where no
		point is erased. Hence the tester makes a total of
		$O\left(\frac{{\log}^6 {n}\log \log n}{{\epsilon}^2 (1 - \delta)}\right)$
		queries.
	\end{proof}

	\subsection{Testing monotonicity of partial functions \texorpdfstring{over ${[n]}^2$ and improving dependence on $\epsilon$}{}}\label{sec:simple-mon2}
	
	\cref{thm:basic-mon} asserts the correctness and complexity
	of an erasure-resilient monotonicity tester, which serves our purpose for $(1,2,3)$-freeness testing.
	However, for testing $(1,3,2)$-freeness, we need the ability to test monotonicity of 
	partial functions $f \colon P \to \R$, where $P$ is an arbitrary subset of $[n]^2$.
	The subdomain $P$ is known to the tester. 
	An erasure-resilient monotonicity tester as described above can be used for 
	partial function monotonicity testing. However, $P$ can be arbitrarily small in many cases and 
	the inverse dependence of the query complexity on $|P|$ (from \cref{thm:basic-mon}) is not desirable. 
	To address this issue, we design an alternate tester whose query complexity
	does not depend on $\Abs{P}$. Furthermore, we
	improve upon the quadratic dependence of the query complexity on
	$1/\epsilon$. We stress that the tester crucially depends on the fact that $P$ is known (unlike in the erasure-resilient case).
	
	\begin{theorem} \label{thm:improved-mon-partial}
		There is a one-sided error $\epsilon$-tester for
		$(1,2)$-pattern freeness of partial functions $f
		\colon P \to \R$, where $P \subseteq [n]^2$, with
		query complexity $\widetilde{O}\left(\frac{ \log^2
			(1/\epsilon)}{{\epsilon}}\right)$. 
	\end{theorem}

	Recall that $f$ has a matching $M$ of violations of size at least
	$\epsilon |P|/2$.
	Just as in~\cref{sec:mono}, let $n$ be a power of $2$, $\ell \colonequals
	\log_2 n + 1$, $\mathcal{S} \colonequals \set{0,1,\ldots,\log_2 n}$, $\alpha, \beta \in \mathcal{S}$, 
	${\epsilon}_1 \colonequals \epsilon / {\ell}^2$, so that there exists a pair $(\alpha^*,
	\beta^*)$ and a corresponding submatching
	$M_{{\alpha}^*, {\beta}^*}$ of size $\Abs{M_{{\alpha}^*, {\beta}^*}} \geq {\epsilon}_1
	\Abs{P} / 2$. Further, 
	we partition the domain ${[n]}^2$ (regardless of $P$) into
	boxes of appropriate sizes, as in~\cref{defn:adjacent_region},  and
	define the \emph{adjacent regions} of boxes exactly as we did there. This implies
	that all the  $1$- and $2$-legs of each $(1,2)$-appearance in $M_{{\alpha}^*, {\beta}^*}$ are in a box $B
	\in \mathcal{B}$ and the corresponding diagonally adjacent box $B'$ (see \cref{defn:adjacent_region}) respectively. As in~\cref{sec:mono}, we only describe a basic tester ${\operatorname{TEST}}^* (\alpha, \beta)$ for the scales $\alpha, \beta$.
	
	We note that we did not make any assumptions on $P$ so far, and the
	partition of the domain into boxes is oblivious of $P$. In particular,
	some of the boxes may not contain any points of $P$. For a box $B$,
	let $P(B) \colonequals P \cap B$. We define the following natural
	probability distribution ${\mathcal{D}}_P$ on the set $\mathcal{B}$ of
	boxes by ${\operatorname{Pr}}_{{\mathcal{D}}_P}(B) \colonequals
	\frac{\Abs{P(B)}}{\Abs{P}}$. 
	
	\paragraph{The improved algorithm ${\operatorname{TEST}}^* (\alpha, \beta)$.} \label{para:improved-mon-tester} 
	
	\begin{enumerate}
		\item Let $\ell \colonequals  \log_2 n + 1, \epsilon_1 \colonequals \frac{\epsilon}{{\ell}^2}, r \colonequals \log_2 {\left(\frac{4}{{\epsilon}_1}\right)}$. 
		
		\item For each $k \in [2r]$, perform the following steps:

		\begin{enumerate}
			\item \label{item-sample-boxes} Sample $r_k \colonequals \frac{20r}{2^k {\epsilon}_1}$ boxes $B_1,\ldots,B_{r_k}$ from $\mathcal{B}$ independently according to the distribution ${\mathcal{D}}_P$.
			
			\item For each pair $(B_i, B_i')$, sample $q_k \colonequals 4 \cdot 2^k$ points from $B_i \cap P$ and $q_k$ points from $B_i' \cap P$ independently and uniformly at random, and query $f$ at all of these points. If a violation to monotonicity is found, we \textbf{reject}.
		\end{enumerate}
	\end{enumerate}
	
	Just as in~\cref{sec:mono}, we fix a matching $M$ of violations and ${\alpha}^*, {\beta}^* \in \mathcal{S}$ for which $\Abs{M} \geq \frac{\epsilon \Abs{P}}{2}$ and $\Abs{M_{{\alpha}^*,{\beta}^*}} \geq \frac{\epsilon \Abs{P}}{2 {\ell}^2}$. We will now show that ${\operatorname{TEST}}^* ({\alpha}^*, {\beta}^*)$ finds a violation with high probability, and compute its query complexity.
	
	\vspace{0.6 cm}
	\noindent
	
	\begin{claim} \label{claim:improved}
		The algorithm ${\operatorname{TEST}}^* ({\alpha}^*, {\beta}^*)$ finds a violation with probability $\Theta(1)$. Its query complexity is $O\left(\frac{{\log}^2 \left(\frac{1}{{\epsilon}_1}\right)}{{\epsilon}_1}\right)$.
	\end{claim}
	
	\begin{proof}
		The total number of queries is bounded by $\underset{k \in [2r]}{\sum} r_k \cdot (2q_k) = (2r) \cdot \frac{20r}{{\epsilon}_1} \cdot 8$, so that the query complexity is $O\left(\frac{{\log}^2 \left(\frac{1}{{\epsilon}_1}\right)}{{\epsilon}_1}\right)$ as claimed.
		
		To estimate the success probability, for each $B \in \mathcal{B}$, we denote by ${\epsilon}_B \colonequals \frac{\Abs{M(B,B')}}{\Abs{P(B)}}$ the normalized number of violations from $M_{{\alpha}^*, {\beta}^*}$ with their $1$-leg in $B$ and $2$-leg in $B'$. Then
		\begin{gather}
			\underset{B \in \mathcal{B}}{\mathbb{E}} [{\epsilon}_B]  = \underset{B \in \mathcal{B}}{\sum} {\operatorname{Pr}}_{{\mathcal{D}}_P} (B) \cdot {\epsilon}_B = \underset{B \in \mathcal{B}}{\sum} \frac{\Abs{P(B)}}{\Abs{P}} \cdot \frac{\Abs{M(B,B')}}{\Abs{P(B)}} \geq \frac{{\epsilon}_1}{2}, \label{eqn:E(eps_B)}
		\end{gather}
		where the expectation is taken with respect to the distribution ${\mathcal{D}}_P$ on $\mathcal{B}$.
		
		Now, for each $k \in [2r]$, let $p_k$ be the probability that $B$ (chosen according to ${\mathcal{D}}_P$) has $2^{-k} < {\epsilon}_B \leq 2^{-k+1}$.
		
		Assume that for some fixed $k$, we have $p_k \geq \frac{2^k {\epsilon}_1}{16r}$. We will bound the success probability by the event that the algorithm will succeed in iteration $k$. Indeed, at iteration $k$, the algorithm will pick a box $B$ with ${\epsilon}_B > 2^{-k}$ with probability at least $p_k$. Since $r_k \geq \frac{1}{p_k}$ by our assumption on $p_k$, with probability $\Theta(1)$, at least one of the sampled boxes $B$ in the $k^{\text{th}}$ iteration will have ${\epsilon}_B > 2^{-k}$. Conditioned on this,~\cref{claim:median_argument} (with $\gamma = 2^{-k}$) then implies that for such a $B$, picking $q_k$ points each from $B$ and $B'$ uniformly at random and querying them will ensure that we find a violation with probability $\Theta(1)$. Altogether, this proves that the success probability is $\Theta(1)$ under the above assumption on $p_k$.
		
		We are now left with the case that, for every $k \in [2r]$, $p_k < \frac{2^k{\epsilon}_1}{16r}$. In this case, we have
		\[
		\underset{B \in \mathcal{B}}{\mathbb{E}} [{\epsilon}_B] \leq \underset{k \in [2r]}{\sum} p_k \cdot 2^{-k+1} + 1 \cdot 2^{-2 \log {\left(\frac{4}{{\epsilon}_1}\right)}}
		\]
		
		where the last term accounts for an upper bound on the expected value of ${\epsilon}_B$ over all the boxes for which ${\epsilon}_B < 2^{-2 \log {\left(\frac{4}{{\epsilon}_1}\right)}}$. 
		
		Simplifying the bound above gives us $\underset{B \in \mathcal{B}}{\mathbb{E}} [{\epsilon}_B] \leq 2 \cdot 2r \cdot \frac{{\epsilon}_1}{16r} + {\left(\frac{{\epsilon}_1}{4}\right)}^2 < \frac{{\epsilon}_1}{2}$, contradicting \cref{eqn:E(eps_B)}. This completes the proof.
	\end{proof}
	
	\cref{claim:improved} directly implies \cref{thm:improved-mon-partial}, and can be easily adapted and used in place of \cref{lem:one_iteration} to get \cref{thm:improved-mon-er}, which is an improved version of \cref{thm:basic-mon}.
	
	\begin{proof}[Proof of \cref{thm:improved-mon-er}] \label{proof:making-partial-ER}
		
		We can obtain a basic erasure-resilient monotonicity tester ${\operatorname{ER-TEST}}^* (\alpha, \beta)$ corresponding ot each bucket $(\alpha, \beta)$ by only making a couple of modifications to the basic tester ${\operatorname{TEST}}^* (\alpha, \beta)$ for partial functions given in the proof of \cref{thm:improved-mon-partial}. In particular, (i) in Step~\ref{item-sample-boxes}, the boxes $B_1,\dots,B_{r_k}$ should be sampled independently and \emph{uniformly at random} (as opposed to sampling according to the distribution $\mathcal{D}_P$ in the case of partial functions, since we do not have access to the list of nonerased points apriori), and (ii) the distance parameter $\epsilon_1$ must be set to $\frac{\epsilon \cdot (1 - \delta)}{\ell^2}$, where $\delta \in (0,1)$ is an upper bound on the fraction of erased points in the domain $[n]^2$, and $\ell \colonequals \log_2 n + 1$. Then, as was done in the case of partial functions, the basic tester ${\operatorname{ER-TEST}}^* (\alpha, \beta)$ is run independently for each of the $\ell^2$ buckets $(\alpha, \beta)$, where $\alpha, \beta \in \set{0,1,\dots,\log_2 n}$.
		
		In this case, observe that $\underset{B \in_R \mathcal{B}}{\mathbb{E}} [\epsilon_B] \geq \frac{\epsilon_1}{2}$, where $\epsilon_B \colonequals \frac{\Abs{M(B,B')}}{\Abs{B}}$ denotes the  normalized distance parameter, $\epsilon_1 \colonequals \frac{\epsilon \cdot (1 - \delta)}{\ell^2}$, and the expectation is taken with respect to the uniform distribution on the set of boxes $\mathcal{B}$. The rest of the proof is analogous to \cref{claim:improved}. Therefore, we have a $\delta$-ER $\epsilon$-tester for monotonicity of real-valued functions on $[n]^2$ with query complexity $\ell^2 \cdot O\left(\frac{{\log}^2 \left(\frac{1}{{\epsilon}_1}\right)}{{\epsilon}_1}\right)$, which is $\widetilde{O}\left(\frac{\log^2(1/\epsilon)}{\epsilon (1-\delta)}\right)$.
	\end{proof}
	
	For using the monotonicity testers described above in $3$-pattern testing, we require
	more. We need the testers to output not just an arbitrary violation to monotonicity, but rather a \enquote{typical} one as we remark below.
	
	\begin{remark}\label{rem:needed}
		Let $M$ be a matching of $(1,2)$-pairs on a domain $D$. \cref{thm:improved-mon-er} describes a tester for monotonicity that finds a violation $(p_1,p_2)$ with
		$p_1$ and $p_2$ being the $1$-leg and the $2$-leg of some pairs in $M$ (but not
		necessarily a pair by itself) with constant probability.
		The tester induces a probability distribution
		$P_M$ on the $1$- and $2$-legs of the pair it produces and in
		particular, a distribution on the $1$-leg $p_1$. It is easy to see
		that this probability is uniform over the $1$-legs of pairs in $M$.
	\end{remark}
	
	\begin{remark}
		Our monotonicity testers work even when the input function is over a rectangular grid $[n_x] \times [n_y]$, and the dependence of the query complexity on the input size is then $\polylog (\max\{n_x,n_y\})$.
	\end{remark}
	
	\subsection{Generalizing to any dimension} \label{sec:mon-general-d}
	
	The ideas described in \cref{sec:mono} and \cref{sec:simple-mon2} to construct partial function and ER monotonicity testers for functions $f \colon [n]^2 \to \R$ generalize in a natural way when the domain is $[n]^d$, for any $d \geq 1$. We simply extend the notion of bucketing described in \cref{sec:mono} to pairs of points in $[n]^d$. To this end, given a tuple $(p,q)$ with $p,q \in [n]^d$ and $p \prec q$, we place $(p,q)$ in the bucket $(\ell_1,\ldots,\ell_d)$ where $\ell_1,\ldots,\ell_d \in \set{0,1,\ldots,\log_2 n}$ if, for each $i \in [d]$, the coarsest partition of $[n]$ that separates $p[i]$ and $q[i]$ (i.e., the $i$-coordinates of $p$ and $q$ respectively) is constructed at level $\ell_i$ (i.e., when $[n]$ is partitioned into $2^{\ell_i}$ equal parts). If $p[i] = q[i]$, we set $\ell_i \colonequals 0$. Given this bucketing scheme, the \emph{boxes} into which we partition the domain $[n]^d$ will have side-length $n/2^{\ell_i}$ in the $i$-coordinate if $\ell_i \geq 1$ and $0$ otherwise, for each $i \in [d]$.
	
	Let $f \colon [n]^d \to \R$ be $\epsilon$-far from being $(1,2)$-free. Then, there exists a matching $M$ of $(1,2)$-appearances of size $\Abs{M} \geq \epsilon n^d / 2$, which we fix. With the bucketing scheme described above, there then exists a submatching $M_{\mathbf{\ell^*}}$ of size $\Abs{M_{\mathbf{\ell^*}}} \geq \epsilon n^d / 2 {(\log_2 n + 1)}^d$ which consists of all the $(1,2)$-appearances in $M$ that belong to the bucket $\mathbf{\ell^*} \colonequals (\ell_1^*,\ldots,\ell_d^*)$.
	
	We can then obtain a basic monotonicity tester  ${\operatorname{TEST}}^* (\ell_1,\ldots,\ell_d)$ (that corresponds to the bucket $(\ell_1,\ldots,\ell_d)$) for partial functions $f \colon P \to \R$, where $P \subseteq [n]^d$, that is analogous to the algorithm ${\operatorname{TEST}}^* (\alpha, \beta)$ for $d = 2$ which is given in \cref{sec:simple-mon2}. The only modification to be made is to the distance parameter $\epsilon_1$, which should now be set to $\frac{\epsilon }{\ell^d}$, where $\ell \colonequals \log_2 n + 1$. The final tester is then simply ${\operatorname{TEST}}^* (\ell_1,\ldots,\ell_d)$ run independently on each of the $\ell^d$ bucket types $(\ell_1,\ldots,\ell_d)$, where $\ell_1,\ldots,\ell_d \in \set{0,1,\ldots,\log_2 n}$.
	
	The final query complexity of this monotonicity tester for real-valued partial functions defined on some given subset $P$ of $[n]^d$ is then $\ell^d \cdot O\left(\frac{{\log}^2 \left(\frac{1}{{\epsilon}_1}\right)}{{\epsilon}_1}\right)$, which is $O\left(\frac{\log^{O(d)} n}{\epsilon} \cdot \log^2 \frac{1}{\epsilon}\right)$. The query complexity is independent of $\Abs{P}$, as desired. We have proved the following.
	
	\begin{theorem} \label{thm:partial-mon-final}
		There is a one-sided error $\epsilon$-tester for monotonicity of partial functions $f \colon P \to \R$, where $P$ is a (known) subset of $[n]^d$, with query complexity $O\left(\frac{\log^{O(d)} n}{\epsilon} \cdot \log^2 \frac{1}{\epsilon}\right)$, for all $\epsilon \in (0,1)$.
	\end{theorem}
	
	One can then obtain a basic ER-tester ${\operatorname{ER-TEST}}^* (\ell_1,\ldots,\ell_d)$ (that corresponds to the bucket $(\ell_1,\ldots,\ell_d)$) for functions $f \colon [n]^d \to \R$, from the tester ${\operatorname{TEST}}^* (\ell_1,\ldots,\ell_d)$ for partial functions, just like it was done in the proof of \cref{thm:improved-mon-er} towards the end of \cref{sec:simple-mon2}. Namely, the boxes should be sampled independently and uniformly at random, and the distance parameter $\epsilon_1$ should be set to $\frac{\epsilon \cdot (1 - \delta)}{\ell^d}$, where $\delta \in (0,1)$ is an upper bound on the fraction of erased points in the domain $[n]^d$, and $\ell \colonequals \log_2 n + 1$. The final query complexity of this $\delta$-ER $\epsilon$-tester for monotonicity of real-valued functions on $[n]^d$ is then $\ell^d \cdot O\left(\frac{{\log}^2 \left(\frac{1}{{\epsilon}_1}\right)}{{\epsilon}_1}\right)$, which is $O\left(\frac{\log^{O(d)} n}{\epsilon (1 - \delta)} \cdot \log^2 \frac{1}{\epsilon}\right)$. We have proved the following.
	
	\begin{theorem} \label{thm:er-mon-final}
		There is a $\delta$-erasure-resilient $\epsilon$-tester for monotonicity of functions $f \colon [n]^d \to \R$ with one-sided error, and query complexity $O\left(\frac{\log^{O(d)} n}{\epsilon (1 - \delta)} \cdot \log^2 \frac{1}{\epsilon}\right)$, for all $\epsilon, \delta \in (0,1)$.
	\end{theorem}

	\section{Missing proofs from \texorpdfstring{\cref{sec:123}}{Section 4}} \label{sec:123_missing_proofs}
	
	In this section, we provide the pseudocodes and omitted proofs for some of the cases of our $(1,2,3)$-freeness tester outlined in \cref{sec:123}.
	
	Let $f \colon [n]^2 \to \R$ be $\epsilon$-far from $(1,2,3)$-free. Then, $f$ contains a matching $M$ of $(1,2,3)$-appearances of size $\Abs{M} \geq \epsilon n^2 / 3$. This matching $M$ is partitioned into submatchings $M = M_0 \cup M_1 \cup M_2 \cup M_3$, as explained in \cref{sec:grid_of_boxes}. The subcase in which $M_0$ is large, i.e., $\Abs{M_0} \geq \epsilon n^2 \cdot \left(1 - \frac{1}{\log n}\right)$ has been described in \cref{sec:M0_large-main} and \cref{sec:recursive}. The subcase in which $M_3$ is large, i.e., $\Abs{M_3} \geq \epsilon n^2 / (3 \log n)$, has been described in \cref{step:3}.
	In what follows, we complete the description of the other cases.
	
	\subsection{\texorpdfstring{$M_1$}{M1} is large}
	
	Let $M_1 \subseteq M$ denote the set of $(1,2,3)$-appearances with all their legs in the same row (or column) of boxes in $G_m^{(2)}$, but not all in the same box.
	
	The matching of $(1,2,3)$-appearances such that the $1$- and $2$-legs
	are in the same box, but the $3$-leg is in a different box, can be further
	split into these appearances for which 
	the $y$-spacing between $2$- and $3$-legs is more than that between the $1$-
	and $2$-legs. We denote this submatching as  $M_{11}$, and denote by $M_{12}$
	the submatching
	for which the $y$-spacing between the $1$- and $2$-legs in the
	$(1,2,3)$-appearance is more than that between the $2$- and $3$-legs. This leads to two subcases.
	
	\subsubsection{\texorpdfstring{$M_{11}$}{M11} is large} \label{sec:M11_large}
	
	When $M_{11}$ is large, \cref{alg:M11-large-123} will find a $(1,2,3)$-appearance with probability at least $1 - n^{-\Omega(\log n)}$, by making $\widetilde{O}(m/{\epsilon}^3)$ queries.
	
	\begin{algorithm} [ht]
		\begin{algorithmic}
		\Statex	\item Repeat $\widetilde{\Theta}(1/\epsilon)$ times:
			\begin{enumerate}
				\item Sample a uniformly random box $B$. Let $R$ denote the row containing $B$.
				\item For each scale $k \in [\log (n/m)]$, 
				\begin{enumerate}
					\item Divide $R$ equally into $n/(2^k m)$ horizontal strips.
					\item Repeat $\widetilde{\Theta}(1/\epsilon)$ times:
					\begin{enumerate}
						\item Sample and query a uniformly random point $p_3$ in $B$.
						\item Let $H'$ be the strip containing $p_3$ and let $H$ be the horizontal strip immediately below $H'$. 
						\item Run a $(1,2)$-freeness tester with proximity parameter $\widetilde{\Theta}(\epsilon/m)$ in the restriction of $H$ to the region of the row $R$ to the left of box $B$.
					\end{enumerate}
				\end{enumerate}
				\item Return any $(1,2,3)$-appearance from among the points queried.
			\end{enumerate}
			\caption{$(1,2,3)$-freeness testing: $M_{11}$ is large}
			\label{alg:M11-large-123}
		\end{algorithmic}
	\end{algorithm}
	
	Suppose that the cardinality of $M_{11}$ is at least $\epsilon' n^2$,
	where $\epsilon'$ is $\epsilon/(c \log n)$ for some absolute constant
	$c$. This case is the easy case in which partitioning $M_{11}$
	into the exponentially growing buckets  according to the $y$-spacing
	between the $2$- and $3$-legs, much like in
	\cite{NewmanRRS19}, allows us to use the `median' argument. The
	details follow.

	By the assumption on the size of $M_{11}$, with probability at least $\epsilon'/2$, a uniformly random box has at least $(\epsilon'/2) \cdot (n^2/m^2)$ many $3$-legs of appearances in $M_{11}$. 
	Thus, with probability at least $1-n^{-\Omega(1)}$, one out of $\widetilde{\Theta}(1/\epsilon)$ sampled boxes contains this number of $3$-legs.
	Let $B$ denote such a box and 
	let $R$ be the row containing $B$. Let $R'$ be the part of the row $R$ to the left of $B$. 
	
	For each $k \in [\log n]$, we consider equipartitioning $R$ into $n/(2^k m)$ horizontal strips, each being a grid isomorphic to $[n] \times [2^k]$. 
	Let $k' \in [\log n]$ be the scale that contains the largest number of the $(1,2,3)$-appearances in $M_{11}$ whose $1, 2$ legs are on a strip and the corresponding $3$-leg is in the strip immediately above. That is, this scale contains at least $(\epsilon'/(2\log n)) \cdot (n^2/m^2)$ many $1$- and $2$-legs, and their corresponding $3$-legs, of appearances in $M_{11}$.
	
	The restriction to $B$ of a uniformly random horizontal strip from
	this scale contains, in expectation, at least $(\epsilon'/(2\log n))
	\cdot (n/m) \cdot 2^{k'}$ many $3$-leg appearances. Thus with (high)
	probability $p=\Omega(\epsilon'/\log n)$,
	a sampled strip when restricted to $B$ contains at least
	$(\epsilon'/(4\log n)) \cdot (n/m) \cdot 2^{k'}$ many $3$-leg appearances.
	Let $H$ denote such a horizontal strip and let $H'$ denote the
	horizontal strip immediately above $H$.  Sampling a random point from
	$H' \cap B$ will find a $3$-leg from $M_{11}$ with probability
	$\Omega(p)$. Further, with probability $p/2$ this point will have value higher
	than the median value of the $3$-legs in $H' \cap B$. Assuming the
	above happens, we only need to find a $(1,2)$-appearance in $H$, to
	the left of $B$ and below the sampled point in $B$ that is assumed to
	be a $3$-leg. Now, by the density argument above, this region is
	$p/(2m \log n)$-far from being $(1,2)$-free. Hence, the $(1,2)$-freeness
	tester applied on this region will find such corresponding
	$(1,2)$-appearance with high probability.  A successful
	$3$-leg as above will correspond to a $(1,2,3)$-appearance. 
	
	Working out the details, the success probability for one such sampled
	point in $B$, assuming that the right scale is known, is $\widetilde{\Omega}(\epsilon^3
	/m)$. Amplification can be done for all scales resulting in
	a high overall success probability and $\widetilde{O}(m/\epsilon^3)$ queries.

	\subsubsection{\texorpdfstring{$M_{12}$}{M12} is large} \label{sec:M12_large}
	
	When $M_{12}$ is large, \cref{alg:M12-large-123} will find a $(1,2,3)$-appearance with probability at least $1 - n^{-\Omega(\log n)}$, by making $\widetilde{O}(m/{\epsilon}^2)$ queries.
	
	\begin{algorithm}
		\begin{algorithmic}
		\Statex \item Repeat $\widetilde{\Theta}(1/\epsilon)$ times:
			\begin{enumerate}
				\item Sample a uniformly random box $B$. Let $R$ denote the row containing $B$ and let $R'$ be the restriction of the row to the boxes right of $B$.
				\item For each scale $k \in [\log (n/m)]$, 
				\begin{enumerate}
					\item Divide $R$ equally into $n/(2^k m)$ horizontal strips.
					\item Repeat $\widetilde{\Theta}(1)$ times:
					\begin{enumerate}
						\item Run the $(1,2)$-freeness tester on $B$ with proximity parameter $\epsilon'/(4\log n)$.
						\item If the tester returns a $(1,2)$-appearance, let $H$ be the horizontal strip containing the $2$-leg of the $(1,2)$-appearance.
						\item Let $H'$ be the strip immediately above $H$ (if it exists) and sample $\widetilde{\Theta}(m/\epsilon)$ points from $H' \cap R'$. 
					\end{enumerate}
				\end{enumerate}
				\item Return any $(1,2,3)$-appearance from among the points queried.
			\end{enumerate}
			\caption{$(1,2,3)$-freeness testing: $M_{12}$ is large}
			\label{alg:M12-large-123}
		\end{algorithmic}
	\end{algorithm}
	
	This is the hard case, and the proof of the correctness of \cref{alg:M12-large-123} can be found in \cref{sec:M1_large-main}.
	
	\subsection{\texorpdfstring{$M_2$}{M2} is large} \label{sec:M2_large}
	
	$M_2$ is the matching of $(1,2,3)$-appearances where the boxes containing the $1$- and $2$-legs do not share a row or column. Assume that $M_2$ has cardinality at least $\epsilon n^2/(3\log n)$. We partition $M_2 = M_{22} \cup M_{23}$ where in $M_{22}$ the legs
	span $2$ rows of boxes (the case of $2$ columns is similar). $M_{23}$ contains those
	$(1,2,3)$-appearances that span $3$ rows and $3$ columns.

	\paragraph{$M_{22}$ is large.} Assume first that $|M_{22}| \geq 
	\epsilon n^2/(6\log n)$.
	There are only two cases as shown in (a) and (b) of \cref{fig_M2}, and both cases are treated identically as follows.
	By our assumption, a uniformly random box $B$ will contain, in expectation, $\Omega(\epsilon n^2/m^2)$
	$2$-legs from  such $(1,2,3)$-appearances. 
	Therefore, with probability $\Omega(\epsilon)$, a uniformly random box
	$B$ will contain $\Omega(\epsilon n^2/(m^2\log n))$ many $2$-legs from  such
	$(1,2,3)$-appearances, which we call the `good event' and condition on it. 
	Let $R'$ be the restriction of the row containing $B$ to the boxes to the right of $B$ including $B$.
	Conditioned on the good event, $R'$ will have $\Omega(\epsilon n^2/(m^2\log n))$
	pairs that are the $2$- and $3$-legs of such $(1,2,3)$-appearances whose $2$-leg is in $B$.
	In other words, $R'$ is $\Omega(\epsilon/(m\log n))$-far from being $(2,3)$-free where the $2$-legs are all in $B$. 
	Hence, with probability at least $\Theta(1)$,  our $(1,2)$-freeness
	tester with proximity parameter $\epsilon' = \Theta(\epsilon/(m\log n))$ will
	find such a $(2,3)$-appearance. Moreover, by \cref{rem:needed-main},
	conditioned on finding a $(2,3)$-appearance from $M_{22}$, with probability
	$\widetilde{\Omega}(1)$, the $2$-leg found will have value at least as large as the values of the $3/4$ fraction of the $2$-legs in $B$ (i.e., the $2$-leg found will have value in the \enquote{top quartile} of the values of all the $2$-legs in $B$ w.h.p.).
	Assuming this happens, let the legs of the $(2,3)$-appearance that is
	found be denoted by $(p_2, p_3)$.
	
	Let $R''$ denote the region of the grid formed by
	taking the union of all boxes to the left of $B$ and below
	$B$.
	Now, there are
	$\Omega(\epsilon n^2/(m^2\log n))$ points $p_1$ in the region $R''$ such that $p_1 \prec p_2$ and $f(p_1) < f(p_2)$. Hence, sampling $\widetilde{\Theta}(m^2/\epsilon)$ points in the region $R''$ will find such a point
	$p_1$ with probability $1 - 1/n^{\Omega(1)}$. Such a point $p_1$ together with
	$(p_2,p_3)$ will form a $(1,2,3)$-appearance. 
	
	By a union bound over all bad events, we find a $(1,2,3)$-appearance with probability $\Omega(\epsilon)$. Thus, $\widetilde{\Theta}(1/\epsilon)$ iterations of the procedure outlined above guarantees the desired success probability. 
	
	The pseudocode for the procedure is given in~\cref{alg:M22-large-123}, and it has query complexity $\widetilde{\Theta}(m^2/\epsilon^2)$.
	
	\begin{figure}[ht]
		\centering
		\includegraphics[scale=0.6]{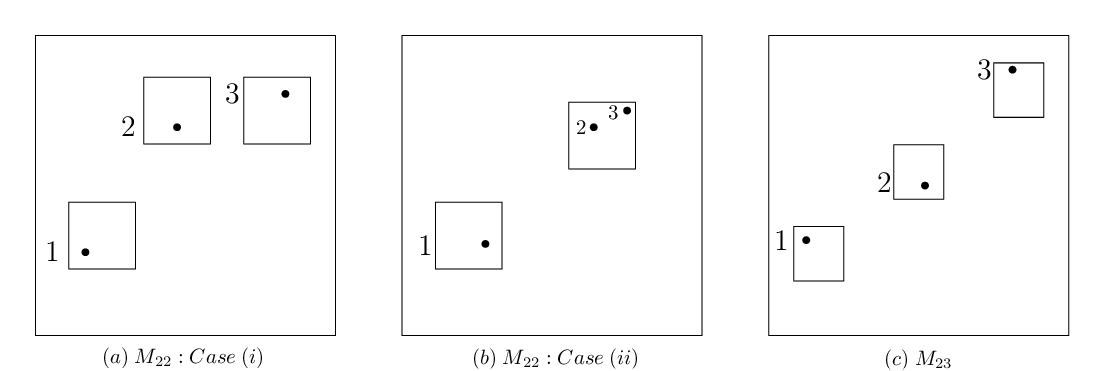}
		\caption{(a), (b) $M_{22}$: Boxes containing legs of $(1,2,3)$-appearances spanning $2$ rows of $G_m^{(2)}$
			(c) $M_{23}$: Boxes containing legs of $(1,2,3)$-appearances spanning $3$ rows and $3$ columns of $G_m^{(2)}$}
		\label{fig_M2}
	\end{figure}

	\begin{algorithm}[ht]
		\begin{algorithmic}
		\Statex	\item Repeat $t = \widetilde{\Theta}(1/\epsilon)$ times:
			\begin{enumerate}
				\item Sample a box $B$ uniformly at random from $G_{m}^{(2)}$. 
				\item Let $R'$ denote the restriction of the row containing
				$B$ to all the boxes to the right of $B$ and including
				$B$. Let $R''$ denote the region of the grid formed by
				taking the union of all boxes to the left of $B$ and below
				$B$. Namely, if $B$ is to contain the $2$-leg of an
				appearance from $M_{32}$, then the
				$3$-leg is in $R'$ and the $1$-leg is in $R''$.
				
				\item Run $t' = \widetilde{\Theta}(1/\epsilon)$ independent executions of a $(1,2)$-freeness tester on $R'$ with proximity parameter $\epsilon' = \Omega(\epsilon/m)$. 
				
				\item Sample $t'' = \widetilde{\Theta}(m^2/\epsilon)$ points uniformly and independently at random from $R''$.
				
				\item \textbf{Reject} if a $(1,2,3)$-appearance is found in the union of all points sampled in the steps above.
			\end{enumerate}
			\caption{$(1,2,3)$-freeness tester: $M_{22}$ is large}
			\label{alg:M22-large-123}
		\end{algorithmic}
	\end{algorithm}

	\paragraph{$M_{23}$ is large.}
	
	The case when $|M_{23}| = \Omega(\epsilon n^2/\log n)$ (see \cref{fig_M2}(c)) is similar, except that the violation
	$(p_2,p_3)$ should not be restricted to the row of the sampled box containing the $2$-leg, but rather to the part
	of the grid consisting of boxes to the right and above the sampled box. This decreases the density of the matching to
	$\widetilde{\Omega}(\epsilon / m^2)$ but the overall complexity is the same as above.
	
	\subsection{\texorpdfstring{$M_3$}{M3} is large} \label{sec:M3_large}
	
	$M_3$ is the matching of $(1,2,3)$-appearances where the boxes containing the $1$- and $2$-legs share a column (resp. row), and the boxes containing the $2$- and $3$-legs share a row (resp. column). That is, the boxes containing the legs of each such $(1,2,3)$-appearance form a \enquote{$\Gamma$} (resp. \enquote{inverted $\Gamma$}) shape. When $M_3$ is large, \cref{alg:M4-large-123} will find a $(1,2,3)$-appearance with probability at least $1 - n^{-\Omega(\log n)}$, by making $\widetilde{O}(m^2 \poly(1/\epsilon))$ queries. The detailed discussion of this case can be found in \cref{step:3}.
	
	\begin{algorithm}[ht]
		\begin{algorithmic}
		\Statex	\item For each $\ell \in \{0,1,\ldots ,\log n\}$:
			\begin{enumerate}
				\item Let $w = 2^{\ell -1}$ for $\ell \geq 1$ and $w = 1$ for $\ell = 0$.
				\item Repeat $t = \widetilde{\Theta}(1/\epsilon)$ times:
				\begin{enumerate}
					\item Sample a box $B$ uniformly at random from the grid $G_m^{(2)}$.
					\item Let $R'$ denote the restriction of $G_m^{(2)}$ to $B$ and the
					boxes to the right of $B$ and in the same row.
					Let $C$ denote the restriction of $G_m^{(2)}$  to $B$ and all the boxes below $B$ and in the same column.
					\item Run $\widetilde{\Theta}(m/\epsilon)$ independent iterations of the $(1,2)$-freeness tester with proximity parameter $\epsilon' = \Theta(\epsilon/(m\log n))$ on $R'$. \label{stp:12freenessGamma}
					
					\item Consider a partition of $C$ into subcolumns of width $w$ each.
					
					\item For each $(1,2)$-appearance $(p_1, p_2)$ returned in
					Step (b)iii such that $p_1 \in B$, perform the following:
					\begin{enumerate}
						\item Let $p_1$ belong to a subcolumn $C'$ of $C$ and let $C''$ be the subcolumn immediately to the left of $C'$.
						\item If $\ell = 0$, then sample $\widetilde{\Theta}(m/\epsilon)$ indices uniformly and independently at random from the region of the subcolumn $C'$ below $B$. 
						\item If $\ell \neq 0$, then sample $\widetilde{\Theta}(m/\epsilon)$ indices uniformly and independently at random from the region of the subcolumn $C''$ below $B$. 
					\end{enumerate}
				\end{enumerate}
				\item \textbf{Reject} if  a $(1,2,3)$-appearance is found.
			\end{enumerate}
			\caption{$(1,2,3)$-freeness tester: $M_3$ is large}
			\label{alg:M4-large-123}
		\end{algorithmic}
	\end{algorithm}

\section{Missing proofs from \texorpdfstring{\cref{sec:132_lbs-main}}{Section 6}} \label{sec:132_lbs}
	
	In this section, we provide the missing proofs in \cref{sec:132_lbs-main}, for the following lower bounds for nonadaptive and adaptive $\epsilon$-testers for $(1,3,2)$-freeness on the hypergrid $[n]^2$.
	
	\begin{theorem} \label{thm:132_na_lb}
		Any one-sided-error nonadaptive $\epsilon$-tester for $(1,3,2)$-freeness on $[n]^2$ has query complexity $\Omega(n)$, for every $\epsilon \leq 1/81$.
	\end{theorem}
	
	\begin{theorem} \label{thm:132_a_lb}
		Any one-sided-error adaptive $\epsilon$-tester for $(1,3,2)$-freeness on $[n]^2$ has query complexity $\Omega(\sqrt{n})$, for every $\epsilon \leq 4/81$.
	\end{theorem}
	
	Following \cite{NewmanRRS19}, we reduce 
	\cref{problem:intersection_search} (\enquote{intersection-search}) for two \emph{$2$-dimensional}
	$3n \times 3n$ `monotone arrays' (w.r.t. $\prec$), to the problem of testing $(1,3,2)$-freeness on the
	$2$-dimensional hypergrid $[9n]^2$.
	
	Throughout this section, We denote the coordinates of an index $\vecv$ of a $2$-dimensional array by $(v_1,v_2)$. Vector addition and subtraction are coordinate-wise.
	
	\begin{problem}[Intersection-Search for two $2$-dimensional arrays] \label{problem:intersection_search}
		The input to this problem consists of two $3n \times 3n$ arrays $A$ and $B$, and a constant $\gamma$ with $0 \leq \gamma \leq 1$ (which we call the intersection parameter). $A$ and $B$ each consist of $9n^2$ distinct integers sorted in ascending order, with respect to the grid (partial) order $\prec$. It is promised that $\Abs{A \cap B} \geq \gamma \cdot (9n^2)$. The goal is to find $\veci,\vecj \in [3n]^2$ such that $A[\veci] = B[\vecj]$.
	\end{problem}
	
	The intersection-search problem can be reduced to the problem of testing $(1,3,2)$-freeness of real-valued functions on $2$-dimensional hypergrids. We will now describe this reduction.
	
	\begin{lemma} \label{lem:reduction_int-search_to_132}
		Let $(A,B,\gamma)$ be an input instance of \cref{problem:intersection_search}. Then, there exists a function $f \colon [9n]^2 \to \mathbb{R}$ that is at least $\gamma/9$-far from $(1,3,2)$-freeness, and there is a one-to-one correspondence between valid solutions of \cref{problem:intersection_search}, and $(1,3,2)$-appearances in $f$.
	\end{lemma}
	
	\begin{proof} \label{proof:reduction_int-search_to_132}
		Let $(A,B,\gamma)$ be an input instance of \cref{problem:intersection_search} -- i.e., $A$ and $B$ are strictly increasing $3n \times 3n$ arrays, and have at least $\gamma \cdot (9n^2)$ elements in common. Let $A^r$ be the reversal of $A$ (i.e., $A^r[\veci] \colonequals A[\vec{i'}]$, where $\vec{i'} = (3n - i_1 + 1, 3n - i_2 + 1)$, if $\veci = (i_1, i_2)$).
		
		We now define a function $f \colon [9n]^2 \to \mathbb{R}$ in terms of the values in $A$ and $B$ as follows:
		
	\[
	f(\veci)=
	\begin{cases}
		A^r[\veci]-\tfrac14, & \text{if }\veci=(2i_1-1,i_2),\ i_1,i_2\in[3n],\\
		A^r[\veci]+\tfrac14, & \text{if }\veci=(2i_1,i_2),\ i_1,i_2\in[3n],\\
		B[\veci],            & \text{if }\veci=(6n+i_1,\,6n+i_2),\ i_1,i_2\in[3n],\\
		0,                   & \text{otherwise.}
	\end{cases}
	\]
		
		In the definition above, $\veci = (i_1,i_2)$. Note that for every valid solution $(\veci,\vecj)$ of \cref{problem:intersection_search} (i.e., for all indices $\veci,\vecj \in [3n]^2$, such that $A[\veci] = B[\vecj] = y$), the indices $(2i_1 - 1, i_2)$, $(2i_1, i_2)$, and $(6n + j_1, 6n + j_2)$, form a $(1,3,2)$-appearance in $f$, with $f$-values $y - 1/4$, $y + 1/4$, and $y$, respectively. It is also easy to check that these are the \emph{only} kinds of $(1,3,2)$-appearances in $f$, as $A^r$ is a strictly decreasing array consisting of positive integers, and $B$ is a strictly increasing array of integers. Since every solution to \cref{problem:intersection_search} uniquely defines one $(1,3,2)$-appearance in $f$ and vice versa, $f$ consists of at least $\gamma \cdot (9n^2)$ pairwise disjoint $(1,3,2)$-appearances, so that $f$ is $\epsilon$-far from $(1,3,2)$-freeness, where $\epsilon \geq (\gamma \cdot (9n^2))/(81n^2) = \gamma/9$.
	\end{proof}
	
	Now that we have the reduction \cref{lem:reduction_int-search_to_132}, we upper bound the success probability of any nonadaptive randomized algorithm for \cref{problem:intersection_search}, which would then directly give us \cref{thm:132_na_lb}.
	
	Note that we only consider algorithms which output a pair of indices $\veci, \vecj \in [3n]^2$ only if it actually queries $A[\veci]$ and $B[\vecj]$, and $A[\veci] = B[\vecj]$. It outputs \enquote{fail} if it does not actually hit such indices.  
	
	\begin{lemma} \label{lem:intersection_search_na_lb}
		Let $\mathcal{A}$ be any nonadaptive randomized algorithm for the intersection-search problem \cref{problem:intersection_search}, with intersection parameter $\gamma \geq 1/9$. If $\mathcal{A}$ makes $q$ queries and $q < n^2 / 2$, then the success probability of $\mathcal{A}$ is at most $q^2 / (4n^2)$.
	\end{lemma}
	
	\begin{proof} \label{proof:intersection_search_na_lb}
		We use Yao's minimax principle \cite{Yao77} to prove the lemma. We define a distribution $\mathcal{D}$ on the valid inputs to \cref{problem:intersection_search}, and then show that any \emph{nonadaptive deterministic} algorithm for this problem will succeed (i.e., output a pair $(\veci,\vecj)$ such that $A[\veci] = B[\vecj]$) with probability at most $q^2 / (4n^2)$ -- i.e., the proportion of inputs (under the distribution $\mathcal{D}$) for which the algorithm outputs a pair of indices is at most $q^2 / (4n^2)$.
		
		We now describe the distribution $\mathcal{D}$ from which we sample an input instance to \cref{problem:intersection_search}, with $\gamma = 1/9$ -- i.e., two monotone increasing (with respect to $\prec$) $3n \times 3n$ integer arrays $A$ and $B$, with exactly $n^2$ common elements. To this end, (i) pick a $3n \times 3n$ array $X$ with $0$-$1$ entries uniformly at random, (ii) independently pick a set $S \subset [2n]^2$ of size $n^2$, uniformly at random, and (iii) independently pick an index $\veck = (k_1,k_2) \in [n]^2$ uniformly at random.
		
		The idea behind these random choices is to construct the input arrays $A$ and $B$ such that $A[\veci] = B[\vecj]$ if and only if $\veci \in S$, and $\vecj = \veci \mathbf{+} \veck$, where `$\mathbf{+}$' refers to coordinate-wise addition. Since $S$ is picked uniformly at random, and $q < n^2 / 2$, it will not be possible for any nonadaptive algorithm to determine $S$ completely. Furthermore, we'll use the randomly chosen $0$-$1$ entries in the array $X$, to make sure that any nonadaptive deterministic algorithm outputs a pair $(\veci,\vecj)$ with $A[\veci] = B[\vecj]$, only if it actually hits the indices $\veci$ and $\vecj$. Choosing $\veck$ randomly ensures that such hits are unlikely, if $q = o(n)$.
		
		Given these random choices, the two input arrays $A$ and $B$ are defined as follows:
		\[
		A[\veci] \colonequals 6n(i_1+i_2)+2i_1+X[\veci],
		\quad \text{for } \veci=(i_1,i_2)\in[3n]^2.
		\]
		and
		\begin{align*}
			B[\veci] \colonequals
			\begin{cases} 
				6n((i_1 - k_1) + (i_2 - k_2)) + 2(i_1 - k_1) + X[\veci \mathbf{-} \veck], & \text{if } \veci \mathbf{-} \veck \in S, \\
				6n((i_1 - k_1) + (i_2 - k_2)) + 2(i_1 - k_1) + (1 - X[\veci \mathbf{-} \veck]), & \text{if } \veci - \veck \in [3n]^2 \setminus S, \text{ and}\\
				6n((i_1 - k_1) + (i_2 - k_2)) + 2(i_1 - k_1),  & \text{otherwise}
			\end{cases}
		\end{align*}
		
		It can be verified that the arrays $A$ and $B$ consist of $9n^2$ distinct integers each, and are sorted in ascending order with respect to the grid (partial) order $\prec$ on the indices. Further, $A$ only consists of positive integers. Note that $A[\veci] = B[\vecj]$ for some indices $\veci,\vecj \in [3n]^2$, if and only if $\veci \in S$, and $\vecj = \veci \mathbf{+} \veck$. Hence, $A$ and $B$ have exactly $n^2$ common values. Also, given $A$ and $B$, for every $\veci \in [2n]^2$, any nonadaptive algorithm can find such a common value only if it queries these indices, when the input is drawn from this distribution.
		
		Let $\mathcal{A}$ be a nonadaptive deterministic algorithm for \cref{problem:intersection_search}, with input $(A,B)$ drawn from the given distribution $\mathcal{D}$, and $\gamma = 1/9$. Let $\mathcal{A}$ query the array $A$ at the set of indices $Q_A \subset [3n]^2$, and query $B$ at the set of indices $Q_B \subset [3n]^2$. Let $q = \Abs{Q_A} + \Abs{Q_B}$. Let $D = \set{\vecj \mathbf{-} \veci \: \colon \: \veci \in Q_A, \vecj \in Q_B}$. Then, if $\mathcal{A}$ does indeed output a pair of indices $(\veci,\vecj)$ (i.e., it succeeds), then it must be the case that $\veck \in D$. Hence, the success probability of $\mathcal{A}$ is at most $\operatorname{Pr}[\veck \in D]$, which is at most $q^2 / (4n^2)$, as $\Abs{Q_A} \cdot \Abs{Q_B} \leq q^2 / 4$, and $\veck$ has been picked uniformly at random from $[n]^2$. This completes the proof.
	\end{proof}
	
	Hence, any nonadaptive randomized algorithm must query $\Omega(n)$ indices from $A$ and $B$ so as to solve \cref{problem:intersection_search} with probability $\Theta(1)$, provided $\gamma = \Theta(1)$. The $\Omega(n)$-query lower bound \cref{thm:132_na_lb} directly follows from \cref{lem:reduction_int-search_to_132} and \cref{lem:intersection_search_na_lb}. $\hfill\qed$
	
	%%%%%%%%%%%%%%%%%%%%%%%%%%%%%%%%%%%%%%%%%%%%%%%%%%%%%%%%%%%%%%%%%%
	
	The adaptive lower bound \cref{thm:132_a_lb} is directly obtained by first proving an $\Omega(\sqrt{n})$-query lower bound for \emph{adaptive} randomized algorithms for \cref{problem:intersection_search}, that succeed with probability $\Theta(1)$.
	
	\begin{lemma} \label{lem:intersection_search_a_lb}
		Let $\mathcal{A}$ be any adaptive randomized algorithm for the intersection-search problem \cref{problem:intersection_search}, with intersection parameter $\gamma \geq 4/9$. If $\mathcal{A}$ makes $q$ queries, then the success probability of $\mathcal{A}$ is at most $q^2 / (4n + 4)$.
	\end{lemma}
	
	\begin{proof} \label{proof:intersection_search_a_lb}
		As was done in \cref{lem:intersection_search_na_lb}, we use Yao's minimax principle to prove the lemma. We define a distribution $\mathcal{D}$ on the valid inputs to \cref{problem:intersection_search}, and then show that any \emph{adaptive deterministic} algorithm for this problem will succeed with probability at most $q^2 / (4n + 4)$ -- i.e., the proportion of inputs (under the distribution $\mathcal{D}$) for which the algorithm outputs a pair of indices is at most $q^2 / (4n + 4)$. We first describe the intuition behind the construction of such a distribution.
		
		Let $(A,B,\gamma \geq 4/9)$ be an input instance of \cref{problem:intersection_search}. Note that the set of indices $[3n]^2$ may be partitioned into $6n - 1$ antichains ${\set{S_r}}_{r = 2}^{6n}$, where $\mathrm{S}_r \colonequals \set{\vecx,\vecy \in [3n]^2 \: \colon \: x_1 + x_2 = y_1 + y_2 = r}$. Note that $S_{3n+1}$ is the main diagonal.
		
		The idea is to construct the distribution in such a
		way that the corresponding antichains $S_r$ in $A$ and $B$
		contain the same set of elements for each $r$, and
		for different $r$'s, the antichains are
		disjoint. This ensures that all solutions to the
		intersection-search problem belong to the same antichain. We will focus only on
		$r \in \set{n+2,\ldots,3n+1}$, for which $\Abs{S_r} =
		r - 1 = \Theta(n)$. Hence, the union of all these
		antichains constitute $4n^2 + n$ elements, implying
		that  $\gamma > 4/9$.
		
		To ensure that the matrices $A$ and $B$ are monotone with
		respect to $\prec$, we will pick the elements in $S_r$ (in both $A$ and $B$,
		as explained), much larger than the ones in $S_{r-1}$. Further, we will make
		sure that the elements outside our antichains of interest cannot be
		matched. Now, since each $S_r$ is an antichain, any order on its
		elements is consistent with the grid order $\prec$. Therefore, to
		find a match, one really needs to solve birthday-search (\cref{problem:birthday-search-main})
		on any one of the $S_r$ ($r \in \set{n+2,\ldots,3n+1}$).
		
		The input distribution $\mathcal{D}$ may be formally defined as follows -- for each $r \in \set{n+2,\ldots,3n+1}$, independently sample a permutation $\pi_r \in {\mathcal{S}_{r-1}}$ uniformly at random (i.e., a random permutation on $r - 1$ letters). Then, the input arrays $A$ and $B$ are defined as follows:
		
		\[
		A[\veci] \colonequals 6n(i_1+i_2)+2i_1, \quad \text{for } \veci=(i_1,i_2)\in[3n]^2.
		\]
		
		and
		
		\begin{align*}
			B[\veci] \colonequals
			\begin{cases} 
				6n(i_1 + i_2) + 2\pi_r(i_1), & \text{if } r = i_1 + i_2 \in \set{n+2,\ldots,3n+1},\\
				6n(i_1 + i_2) + 2i_1 + 1,  & \text{otherwise}
			\end{cases}
		\end{align*}
		
		It can be verified that $A$ and $B$ consist of $9n^2$ distinct elements each, are sorted in ascending order. $A$ and $B$ contain at least $4n^2$ common elements, and for each such element $y = A[\veci] = B[\vecj]$, $\veci$ and $\vecj$ belong to the same antichain $S_{i_1 + i_2}$ ( $= S_{j_1+j_2}$), where $n+2 \leq i_1 + i_2 = j_1 + j_2 \leq 3n + 1$.
		
		Let $\mathcal{A}$ be any adaptive deterministic algorithm for \cref{problem:intersection_search}. Consider an antichain $S_r$, where $r \in \set{n+2,\ldots,6n}$. Let $\mathcal{A}$ query $A$ at the indices $Q_A^{(r)} \subset S_r$, and query $B$ at the indices $Q_B^{(r)} \subset S_r$. Let $q_r = \Abs{Q_A^{(r)}} + \Abs{Q_B^{(r)}}$. Then, the probability (under the given input distribution $\mathcal{D}$) that $\mathcal{A}$ finds a valid pair of indices on $S_r$, with $r \in \set{n+2,\ldots,3n+1}$, is at most $\frac{\Abs{Q_A^{(r)}} \cdot \Abs{Q_B^{(r)}}}{\Abs{S_r}} \leq \frac{q_r^2}{4n+4}$, as $\Abs{Q_A^{(r)}} \cdot \Abs{Q_B^{(r)}} \leq q_r^2/4$, and $\Abs{S_r} \geq n+1$. The success probability is $0$ if $r \not\in \set{n+2,\ldots,3n+1}$. A union bound over these good events over all the antichains (i.e., $r \in \set{n+2,\ldots,6n}$) finishes the proof.
	\end{proof}
	
	Hence, any adaptive randomized algorithm must query $\Omega(\sqrt{n})$ indices from $A$ and $B$ so as to solve \cref{problem:intersection_search} with probability $\Theta(1)$, provided $\gamma = \Theta(1)$. The $\Omega(\sqrt{n})$-query lower bound \cref{thm:132_a_lb} directly follows from \cref{lem:reduction_int-search_to_132} and \cref{lem:intersection_search_a_lb}. $\hfill\qed$
	
	%%%%%%%%%%%%%%%%%%%%%%%%%%%%%%%%%%%%%%%%%%%%%%%%%%%%%%%%%%%%%%%%
	
	\section{\texorpdfstring{\textsf{Layering}}{Layering} and \texorpdfstring{\textsf{Gridding}}{Gridding}}\label{sec:gridding}
	
	In this section, we describe an algorithm that we call
	\textsf{Gridding}, which is a common subroutine in all of our algorithms. This
	machinery is from \cite{NV22}, tailored to the $2$-dimensional domain $[n]^2$.
	The output of \textsf{Gridding}, given oracle access to the function $f \colon [n]^2 \to \R$ and a parameter $m \leq n$, is an $m \times m \times m$ grid of boxes
	that partitions either the grid $[n]^2 \times R(f)$ or a region inside of it into boxes,
	with the property that the density of each box, which we define below, is \enquote{well controlled}.

	\begin{definition}[Density of a box]\label{def:density}
		Consider index and value subsets $S = X \times Y \subseteq [n]^2$ and $I \subseteq R(f)$, respectively, where $X \subseteq [n]$ and $Y \subseteq [n]$.
		The density of $\bx(S,I)$, denoted by $\den(S,I)$, is 
		the number of points in $\bx(S,I)$ normalized by its size $|S|$. 
	\end{definition}

	\begin{definition}[Nice partition of a box]\label{def:nice-partition}
		For index and value sets $S = X \times Y \subseteq [n]^2$ and $I \subseteq R(f)$ and a parameter $m \leq n$, we say that
		${\mathcal I} = \{I_1, I_2, \dots, I_{m'}\}$ forms a \textsf{nice} $m$-partition of $\bx(S,I)$ if
		\begin{itemize}
			\item $m' \leq 2m $,
			\item $I_1, \ldots ,I_{m'}$ are pairwise disjoint, and $\bigcup_{j
				\in [m']} I_j = I$. In particular, 
			the largest value in $I_j$ is less than the  smallest value in $I_{j'}$ for $j < j'$.
			\item for $j \in [m']$, either $\den(S,I_j) < 4/m$ $\mathrm{OR}$ $I_j$ contains exactly one value and $\den(S,I_j) \geq 1/2m$. 
			In the first case, we say that $\bx(S,I_j)$ is a
			\emph{multi-valued layer} of $\bx(S,I)$, and in the second case, we say that $\bx(S,I_j)$ is a \emph{single-valued layer} of $\bx(S,I)$.
		\end{itemize}
	\end{definition}

	\subsection{\texorpdfstring{\textsf{Layering}}{Layering}} \label{sec:layering}
	The main part of \textsf{Gridding} is an algorithm
	\textsf{Layering} which is described in~\cref{proc:layering}. A similar algorithm was used by Newman and Varma~\cite{NV20} for estimating the length of the longest increasing subsequence in an array.
	$\mathsf{Layering}(S,I,m)$ is given $S = X \times Y \subseteq [n]^2$, $I \subseteq R(f)$, and $m \leq n$ as inputs, and it outputs, with probability at least
	$1 - 1/n^{\Omega(\log n)}$, a set  $\mathcal I$ of intervals that is  a nice
	$m$-partition of  $\bx(S,I)$.  It works by  sampling
	$\widetilde{O}(m^2)$ points from $\bx(S,I)$ and outputs the set $\mathcal{I}$ based on these samples.
	Note that the sets $X$, $Y$, and $I$ are either contiguous index/value intervals.
	Additionally, we always apply the algorithm \textsf{Layering} to boxes of
	density $\Omega(1/\log n)$.

	\begin{algorithm}
		\caption{$\mathsf{Layering}(S,I,m)$}
		\label{proc:layering}
		\begin{algorithmic}[1]
			\State Sample a set of $m \log^4 n$  
			indices from $S$ uniformly and independently at random.
			\State Let $U$ denote the set of points in the sample that belong to $\bx(S,I)$. If $u \colonequals |U| < m \log^2 n$, then \textbf{FAIL}. \label{stp:layering-fail}
			
			\State We sort the multiset of values $V = \{f(p): ~ p \in U \}$ to form a strictly increasing sequence $\mathsf{seq}
			= (v'_1 <
			\ldots < v'_q)$, where, for each $i \in [q]$, we associate a weight $w_i$
			that equals the multiplicity of $v_i'$ in the multiset $V$ of
			values.
			\Comment{{\sf Note that
					$\sum_{i \in [q]} w_i = u$.}}
			
			\State We now partition the sequence $W = (w_1, \ldots ,w_q)$ into maximal
			disjoint contiguous subsequences $W_1, \ldots W_{m'}$ such that for each $j \in [m']$, either $\sum_{w \in
				W_j} w <  2u/m$, or $W_j$ contains only one member $w$ for
			which $w > u/m$. 
			\Comment{{\sf This can be done greedily as follows. 
					If $w_1 > u/m$, then $W_1$ will contain
					only $w_1$, and otherwise, $W_1$ will contain the maximal subsequence
					$(w_1, \ldots, w_i)$ whose
					sum is at most $2u/m$. We then delete the members of
					$W_1$ from $W$ and repeat the process. 
					For $i \in [m']$, let $w(W_i)$ denote the total weight in $W_i$.}}
			
			Correspondingly, we obtain a
			partition of the sequence $\seq$ of sampled values into at most $m'$ subsequences $\{\seq_j\}_{j \in [m']}$. 
			Some subsequences
			contain only one value of weight at least $u/m$ and are called \emph{single-valued}.  
			The remaining subsequences are called \emph{multi-valued}.
			
			For a subsequence $\seq_j$, let $\alpha_j = \min(\seq_j)$ and $\beta_j = \max(\seq_j)$.
			Let $\beta_0 = \inf(I)$.
			Note that $\alpha_j \leq \beta_j$  and $\beta_{j-1} < \alpha_j$ for all $j \in [m']$.
			
			\State For $j \in [m']$, we associate with the subsequence $\seq_j$, an interval $I_j \subseteq
			\mathbb{R}$, where $I_j = (\beta_{j-1}, \beta_j] \cap I$, and an approximate density
			$\widetilde{\den}(S,I_j) = w(W_j)/u$. The interval is multi-valued or single-valued depending on whether its corresponding sequence is multi-valued or single-valued, respectively.
			
			\State \textbf{Return} the set $\mathcal{I} = \bigcup_{\ell \in [m']} I_\ell$ of the $m'$ intervals.
		\end{algorithmic}
	\end{algorithm}

	\begin{claim}\label{clm:layering}
		If 
		$\den(S,I) > 1/\log n$, then with probability 
		$1 - 1/n^{\Omega(\log n)}$, \text{Layering}$(S,I,m)$ returns a collection of intervals
		$\mathcal I = \{I_j\}_{j=1}^{m'}$ such that $\mathcal I$  is a nice $m$-partition of  $\bx(S,I)$. Furthermore, it 
		makes a total of $m\log^4 n$ queries. 
	\end{claim}
	
	\begin{proof}
		
		Since $\den(S,I) > 1/\log n$, at least $m \log^2 n$ of the sampled points fall in
		$\bx(S,I)$, with probability at least $1 - \exp(-(m\log^3 n)/8)$, by a Chernoff bound. Hence, the algorithm does not fail in \cref{stp:layering-fail}. In the rest of the analysis, we condition on this event happening.
		
		The total number of intervals of weight at least $u/m$ is at most $m$ since the total weight is $u$. Other intervals have
		weight less than $u/m$ and for each such interval $I_j$, it must
		be the case that $I_{j-1}$ and $I_{j+1}$ are of weight at least $u/m$. It follows that $m' \leq 2m$.
		
		We now prove that the family $\mathcal{I}$ output by \textsf{Layering} is a nice $m$-partition of $\bx(S,I)$.
		It is clear from the description of~\cref{proc:layering} that the intervals output by the algorithm are disjoint.  
		Let $\mathcal{B} = \{[a,b]: a,b \in I \text{ and } \exists v,w \in S \text{ such that } f(v) = a, f(w) =b\}$ denote the set of all true intervals of points 
		from $\bx(S,I)$.  
		Consider an interval $[a,b] \in \mathcal{B}$ such that $\den(S,[a,b]) \ge 4/m$. The probability that 
		less than $2u/m$ points from the sample have values in the range $[a,b]$ is at most $1/n^{\Omega(\log n)}$ by a Chernoff bound.
		Conditioning on this event \emph{not} occurring implies that for every $I_j, j\in [m']$ output as a multi-valued interval by the algorithm, we have $\den(S,I_j) < 4/m$. 
		Finally, for a single-valued interval $[a,a]  \in \mathcal{B}$ such that $\den(S, [a,a]) < 1/2m$, with probability at least $1 - 1/n^{\Omega(\log n)}$, we have $\widetilde{\den}(S,[a,a]) \leq \frac{3}{2}\den(S, [a,a]) < 3/4m$, where $\widetilde{\den}(S,I')$ denotes the estimated density (as estimated in~\cref{proc:layering}) for a layer $\bx(S,I')$ when $I' \subseteq I$. 
		Conditioning on this event occurring implies that for every $I_j, j\in [m']$ output as a single-valued interval by the algorithm, we have $\den(S,I_j) \geq 1/2m$.
		
		The claim about the query complexity is clear from the description of the algorithm. 
	\end{proof}
	
	\subsection{\texorpdfstring{\textsf{Gridding}}{Gridding}} \label{sec:subsec_gridding}
	Next, we describe the algorithm \textsf{Gridding}.
	\begin{algorithm}[ht]
		\caption{$\mathsf{Gridding}(S,I,m,\beta)$}
		\label{proc:preprocess}
		\begin{algorithmic}[1]
			\Require $X \subseteq [n]$ is a union of disjoint $x$-stripes, $Y \subseteq [n]$ is a union of disjoint $y$-stripes, so that $S \colonequals X \times Y$ is a disjoint union of rectangles in the $xy$-plane, $I \subseteq R(f)$ is a disjoint union of intervals of values in
			$R(f)$,
			$m$ is a parameter defining the `coarse' grid size,
			$\beta < 1$ is a density threshold.
			
			\hspace{-1.5cm}
			{\bf Output:} A grid of boxes $G_{m'}^{(3)},~ m' \leq 2m $ in
			which there will be $\widetilde{O}({m}^2)$
			marked boxes.

			\State Call \textsf{Layering} (\cref{proc:layering}) on 
			inputs $S,I,m$. 
			This returns, with high probability, a set $\mathcal{I}$ of
			$m' \leq 2m$ value intervals $I = \bigcup_{j \in [m']} I_j$ that forms a nice $m$-partition  of $\bx(S,I)$.

			\State Partition $X$ into $m'$ contiguous
			intervals $X_1, \ldots X_{m'}$ each of size $|X|/m'$ and partition $Y$ into $m'$ contiguous
			intervals $Y_1, \ldots Y_{m'}$ each of size $|Y|/m'$. This defines the
			grid  of boxes $G_{m'}^{(3)}=\{\bx(X_i \times Y_j, I_k):~ (i,j,k)\in [m']^3\}$ inside
			the larger box $\bx(S,I)$. For each $(i,j) \in [m']^2$, denote by $S_{i,j}$ the rectangle $X_i \times Y_j$ in the $xy$-plane.

			\State Sample and query, independently at random, $\frac{\log^4 n}{\beta^2}$ points from each rectangle
			$S_{i,j}, (i,j) \in [m']^2.$   For each $(i,j,k) \in [m']^3$, if $\bx(S_{i,j}, I_k)$
			contains a sampled point, then tag that box as \emph{marked}.
			If $\bx(S_{i,j},I_k)$ contains at least $3\beta/4$ fraction of the sampled
			points in the rectangle $S_{i,j}$, tag that box as {\em dense}.
			
			\State \textbf{Return} the grid $G_{m'}^{(3)}$ along with the tags on the various boxes.
		\end{algorithmic}
	\end{algorithm}

	We note that initially, at the topmost recursion level of the algorithm
	for $\pi$-freeness, we call \textsf{Gridding} with $S=[n]^2,~ I = (-\infty,
	+\infty)$ and our preferred $m$ which is typically $m \colonequals n^{\delta}$,
	for some small $\delta <1$.

	We prove in~\cref{cl:q_layering} that running
	$\mathsf{Gridding}(S,I,m)$ results in  a partition of $\bx(S,I)$ into a grid of
	boxes $G_{m'}^{(3)}$ in which either the marked boxes 
	contain a $\pi$-appearance, or the union of
	points in the marked boxes 
	contain {\em all} but an $\eta$ fraction of the points in
	$[n]^2 \times R(f)$, for $\eta \ll \epsilon$.
	Additionally, with high probability, all boxes that are tagged {\em
		dense} have density at
	least $\frac{1}{8}$-th of the threshold $\beta$ for marking a box as dense.
	
	\begin{claim}\label{cl:q_layering}
		$\mathsf{Gridding}(S,I,m,\beta)$ returns a grid
		of boxes $G_{m'}^{(3)}$ that decomposes $\bx(S,I)$. It makes $\widetilde{O}(m^2/\beta^2)$ queries, and with high probability,
		\begin{itemize}
			\item The set of intervals corresponding to the layers of $G_{m'}^{(3)}$ form a nice $m$-partition of $I$.
			\item  For every $(i,j) \in [m']^2,$
			either $\bx(S_{i,j},I)$
			contains at least  $\frac{\log^2 n}{100\beta^2}$ marked boxes, or the number of
			points in the marked boxes in $\bx(S_{i,j},I)$ is at least  $
			(1-\frac{1}{\log^2 n}) \cdot |S_{i,j}|.$
			\item Every box that is tagged dense has density at
			least $\beta/8$, and every box of density at least $\beta$ is
			tagged as dense.    
		\end{itemize}
	\end{claim}
	\begin{proof}
		The bound on the query complexity as well as the first item follows directly from~\cref{clm:layering}. The third item follows by a simple application of the Chernoff bound followed by a union bound over all rectangles.

		For
		the second item, fix a rectangle $S_{i,j}$ of $G_{m'}^{(3)}$. 
		Let $T \subseteq [m']$ be the set of all $k \in [m']$ such that $\bx(S_{i,j},I_k)$ gets marked during Step 3 in $\textsf{Gridding}$. If $\sum_{k \in T}\den(S_{i,j}, I_k) \geq
		1-1/(\log^2 n)$, we are done.
		Otherwise, each query independently hits a box 
		that is not marked by any of the previous queries with probability greater than $1/(\log^2 n)$.
		Thus, the expected number of boxes marked is at least $\frac{\log^2 n}{\beta^2}$. Chernoff bound implies that, with probability at least $1 - n^{-\Omega(\log n)},$ at least $\frac{\log^2 n}{100\beta^2}$
		boxes are marked. 
		The union bound over all the rectangles implies the second item.  
	\end{proof}

	\section{Deletion and Hamming distances} \label{sec:hamdel}
	
	In this section, we prove relationships between deletion and Hamming distances from pattern freeness for various kinds of order patterns that we have discussed in the paper. In \cref{sec:hamdelproof}, we prove (\cref{clm:deletion}) that these distances are equal for certain types of permutation patterns over high-dimensional hypergrids. We also describe an example (\cref{clm:hamneqdelexample}) involving a $4$-pattern for which deletion and Hamming distances are significantly different even for the case of $2$-dimensional grids.
	
	In the case of patterns which are \emph{not} permutations, these distances need not be equal even for order patterns of length $3$. In \cref{sec:last}, we give an example of a \enquote{$1,2,3$-fork} pattern for which the gap between the Hamming and deletion distances is large.
	
	\subsection{Permutation patterns} \label{sec:hamdelproof}
	
	First, we prove (\cref{clm:deletion}) that deletion and Hamming distances are equal for all permutation patterns $\pi$ (of any length) for which $\{\pi(1),\pi(k)\}
	\cap \{1,k\} \neq \emptyset$, and for all $f \colon [n]^d \to \R$. We note that this case subsumes all monotone patterns (of any length), and all patterns of length at most $3$ -- this is clear either by looking at the pattern upside down or by multiplying all $f$-values by $-1$.
	
	\begin{claim}{\emph{(Restatement of \cref{clm:deletion})}}
		Let $k \geq 1$ and $\pi \in {\mathcal S}_k$ with $\{\pi(1),\pi(k)\}
		\cap \{1,k\} \neq \emptyset$. Let $n,d \ge 1$. The Hamming distance of $f \colon [n]^d \to \R$ to $P_{\pi}^d$ is equal to the deletion distance of $f$ to $P_{\pi}^d$.
	\end{claim}
	
	\begin{proof}\label{prf:hamdelproof}
		Let $f \colon {[n]}^d \to \mathbb{R}$, and $\pi \in {\mathcal{S}}_k$ with $\pi(1) = 1$. The other cases are analogous, either by looking at the domain upside down, or by multiplying the $f$-values by $-1$. It is easy to see that the deletion distance of $f$ to $P_{\pi}^d$ is at most the Hamming distance -- if changing the $f$-values on some $S \subseteq {[n]}^d$ makes it $\pi$-free, then $\restr{f}{{[n]}^d \setminus S}$ is $\pi$-free. We will prove the reverse inequality.
		
		Let $\Delta$ be the deletion distance of $f$ to $P_{\pi}^d$, and let
		$S \colonequals \set{i_1,\ldots,i_{\Delta}}$ be a set of indices in
		${[n]}^d$ such that $\restr{f}{{[n]}^d \setminus S}$ is $\pi$-free.
		
		Define a function $f^* \colon {[n]}^d \to \mathbb{R}$ such that $\restr{f^*}{{[n]}^d \setminus S} \colonequals \restr{f}{{[n]}^d \setminus S}$ and for all $j \in [\Delta]$,
		\[
		f^*(i_j) \colonequals \min \set{f(p) \colon p \not\in S \text{ and } p \prec i_j}
		\]
		
		We will now show that $f^*$ is
		$\pi$-free. This will imply that the Hamming distance of $f$ to
		$P_{\pi}^d$ is at most $\Delta$, which will complete the proof.  Note that we may assume that $\bar{1} = (1,\ldots,1) \not\in S$, for if $\bar{1} \in S$, setting $f^*(\bar{1}) \colonequals \max \set{f(p) \colon p \not\in S}$ would ensure that $\bar{1}$ cannot be a leg of any $\pi$-appearance in $f^*$.
		
		To show $f^*$ is $\pi$-free, we let $j \in \set{0,\ldots,\Delta}$, and prove that $\restr{f^*}{{[n]}^d \setminus \set{i_{j+1},\ldots,i_{\Delta}}}$ is
		$\pi$-free. For $j = \Delta$, we interpret this to mean $f^*$ is $\pi$-free, which is what we set out to prove. The proof proceeds by induction on $j$. The base case $j = 0$ follows directly from the definition
		of $f^*$ as
		$\restr{f^*}{{[n]}^d \setminus S} = \restr{f}{{[n]}^d \setminus S}$
		, and $\restr{f}{{[n]}^d \setminus S}$ is $\pi$-free.
		
		Now, let us assume that $j \geq 1$ and $\restr{f^*}{{[n]}^d \setminus \set{i_j,\ldots,i_{\Delta}}}$ is $\pi$-free but that $\restr{f^*}{{[n]}^d \setminus \set{i_{j+1},\ldots,i_{\Delta}}}$ is not $\pi$-free. Suppose $i_j$ was the $1$-leg of a $\pi$-appearance $(i_j,p_2,\ldots,p_k)$ in $\restr{f^*}{{[n]}^d \setminus \set{i_{j+1},\ldots,i_{\Delta}}}$, where $p_2,\ldots,p_k \in {[n]}^d \setminus \set{i_{j+1},\ldots,i_{\Delta}}$ and $i_j \prec p_2 \prec \ldots \prec p_k$. From the way we defined $f^*(i_j)$, there exists $p \prec i_j$  such that $f^*(p) = f^*(i_j)$, so that $(p,p_2,\ldots,p_k)$ is a $\pi$-appearance in $\restr{f^*}{{[n]}^d \setminus \set{i_j,\ldots,i_{\Delta}}}$, contradicting the induction hypothesis. Furthermore, $i_j$ cannot be the $t$-leg (for any $2 \leq t \leq k$) of any $\pi$-appearance in ${[n]}^d \setminus \set{i_{j+1},\ldots,i_{\Delta}}$ since, for all $p \prec i_j$, $f^*(i_j) \leq f(p)$. This completes the proof.
	\end{proof}
	
	We will now give an example of a pattern $\pi \in {\mathcal{S}}_4$ and a function $f \colon [n]^2 \to \R$ for which the difference between the Hamming and deletion distances of $f$ from $\pi$-freeness is $\Theta(n)$. This example generalizes straightforwardly to longer patterns and higher dimensions.
	
	\begin{claim}{\emph{(Restatement of \cref{clm:hamneqdelexample})}}
		There exists a $4$-pattern $\pi \in {\mathcal{S}}_4$ and a function $f \colon [n]^2 \to \R$ such that the deletion distance of $f$ from $\pi$-freeness is $1$ and the Hamming distance of $f$ from $\pi$-freeness is $\Theta(n)$.
	\end{claim}
	
	\begin{proof}[Proof of \cref{clm:hamneqdelexample}]
		Let $\pi \colonequals (2,4,1,3) \in {\mathcal{S}}_4$. Consider the function $f \colon {[n]}^2 \to \mathbb{R}$ whose $f$-values are represented by the matrix
		\[
		\begin{bmatrix}
			\; & \; & \; & 6 & \; & \; & \; & \; & \; & \;\\
			\; & \scalebox{1.5}{8} & \; & \vdots & \; & \; & \scalebox{1.5}{8} & \; & \; & \;\\
			\; & \; & \; & 6 & \; & \; & \; & \; & \; & \;\\
			2 & \cdots & 2 & \mathbf{4} & 1 & \cdots & 1 & 3 & \cdots & 3\\
			\; & \; & \; & 7 & \; & \; & \; & \; & \; & \;\\
			\; & \; & \; & \vdots & \; & \; & \; & \; & \; & \;\\
			\; & \scalebox{1.5}{8} & \; & 7 & \; & \; & \scalebox{1.5}{8} & \; & \; & \;\\
			\; & \; & \; & 5 & \; & \; & \; & \; & \; & \;\\
			\; & \; & \; & \vdots & \; & \; & \; & \; & \; & \;\\
			\; & \; & \; & 5 & \; & \; & \; & \; & \; & \;
		\end{bmatrix}
		\]
		
		Let the number of occurrences of $1$, $2$, and $3$ on the \enquote{special} row in the matrix each be $\Theta(n)$. Similarly, let the number of occurrences of $5$, $6$, and $7$ on the \enquote{special} column in the matrix each be $\Theta(n)$. There are thus $\Theta(n)$ $\pi$-appearances on the \enquote{$2,\mathbf{4},1,3$}-row. Similarly, on the \enquote{$5,7,\mathbf{4},6$} column, there are $\Theta(n)$ $\pi$-appearances.
		
		The deletion distance of $f$ from being $\pi$-free is $1$ as the deletion of the $\mathbf{4}$ marked in bold in the matrix is sufficient to remove all the $\pi$-appearances. However, the Hamming distance of $f$ from being $\pi$-free is $\Theta(n)$ since the modification of the $\mathbf{4}$ cannot simultaneously remove the $\pi$-appearances on the row and the column to which it belongs. This is because the $\mathbf{4}$ in the matrix is simultaneously the $4$-leg of each of the $\pi$-appearances along the row, and the $1$-leg of each of the $\pi$-appearances along the column. Hence, we must modify $\Theta(n)$ entries either on the row or the column to remove all the $\pi$-appearances.
	\end{proof}
	
	\subsection{Non-permutation patterns} \label{sec:last}
	
	As mentioned earlier, when $d = 1$, every pattern is a permutation pattern. 
	However, when $d \geq 2$, order patterns that are not permutations
	exist, since the partial order $\prec$ with which we work, is not a total order on the domain $[n]^d$ of the function being tested. As already remarked, even for $d = 2$, there are $3$-patterns (which are not permutations)
	for which the Hamming distance may be arbitrarily larger than the
	deletion distance. We will give an example of such a pattern which we call the \enquote{$1,2,3$-fork} pattern.
	
	A function $f \colon [n]^d \to \R$ is said to contain a $1,2,3$-fork pattern if there exist $p_1, p_2, p_3 \in [n]^d$ such that $p_2 \prec p_1$, $p_2 \prec p_3$, $p_1$ and $p_3$ are not comparable with respect to $\prec$, and $f(p_1) < f(p_2) < f(p_3)$. As usual, we call $p_1$, $p_2$, and $p_3$ the $1$-leg, $2$-leg, and the $3$-leg of the $1,2,3$-fork pattern that they form, respectively. For the rest of this section, let $\pi$ denote the $1,2,3$-fork pattern.
	
	\begin{claim}
		There exists a function $f \colon [n]^2 \to \R$ such that the deletion distance of $f$ from $\pi$-freeness is $1$ and the Hamming distance of $f$ from $\pi$-freeness is $\Theta(n)$.
	\end{claim}
	
	\begin{proof}
		Consider the function $f \colon {[n]}^2 \to \mathbb{R}$ whose $f$-values are represented by the matrix
		\[
		\begin{bmatrix}
			1 & \cdots & 1 & 4 & 4\\
			2 & \cdots & 2 & \mathbf{3} & 2\\
			\; & \; & \; & 4 & 5\\
			\; & \scalebox{1.5}{5} & \; & \vdots & \vdots\\
			\; & \; & \; & 4 & 5
		\end{bmatrix}
		\]
		
		Note that there are $\Theta(n)$ $\pi$-appearances formed by the values $1$, $2$, and $\mathbf{3}$, where the $\mathbf{3}$ marked in bold is the $3$-leg of each of these $\pi$-appearances. There are also $\Theta(n)$ $\pi$-appearances formed by the values $\mathbf{3}$, $4$, and $5$, where the $\mathbf{3}$ marked in bold is the $1$-leg of each of these $\pi$-appearances. The deletion distance of $f$ from $\pi$-freeness is $1$ as deleting the $\mathbf{3}$ would suffice. However, the Hamming distance of $f$ from $\pi$-freeness is $\Theta(n)$, thanks to $\mathbf{3}$ simultaneously being the $1$-leg and the $3$-leg of $\Theta(n)$ $\pi$-appearances each.
	\end{proof}
	
	We do not know how to design tests for such patterns with respect to the Hamming distance and conjecture that it
	requires $n^{\omega(1)}$ queries. This is to be compared
	with the only known provable lower bound on permutation-freeness testing which is $\Omega(\log n)$.
	
	Thus, it makes sense to test freeness of such order patterns with
	respect to the deletion distance. Our methods for $(1,3,2)$-freeness
	apply for such patterns too, and in particular, an $\widetilde{O}(n)$-query erasure-resilient test can be constructed along the lines of \cref{sec:O(n)-132}.
	
\end{appendices}

%%%%%%%%%%%%%%%%%%%%%%%%%%%%%%%%%%%%%%%%%%%%%%%%%%%%%%%%%%%%

\end{document}